\documentclass{lmcs}
\pdfoutput=1

\usepackage{lastpage}
\lmcsdoi{17}{1}{6}
\lmcsheading{}{\pageref{LastPage}}{}{}%
{Apr.~05,~2019}{Jan.~28,~2021}{}

\usepackage[active]{srcltx}
\usepackage[latin1]{inputenc}
\usepackage{mathpartir}  
\usepackage{latexsym}
\usepackage{amssymb}
\usepackage{amsmath}

\usepackage{amsthm}
\usepackage{array,longtable}
\usepackage{caption} 
\newcolumntype{C}{>{$}c<{$}}
\newcolumntype{R}{>{$}r<{$}}
\newcolumntype{L}{>{$}l<{$}}

\usepackage{stmaryrd}
\usepackage{booktabs}

\usepackage{graphics} 

\newcommand{\may}[1]{\langle{}#1\rangle}
\newcommand{\must}[1]{{[} #1 {]}}
\newcommand{\wmay}[1]{\langle\!\langle #1 \rangle\!\rangle}

\newcommand{\Set}[1]{\left\{#1\right\}}
\newcommand{\casesep}{\quad | \quad}
\newcommand{\setsep}{\mid}

\newcommand{\nameset}{\mathcal{N}}
\newcommand{\formset}{\mathcal{A}}

\newcommand{\actions}{\mbox{\sc act}}
\newcommand{\states}{\mbox{\sc states}}
\newcommand{\predicates}{\mbox{\sc pred}}

\newcommand{\conjuncset}[1]{\bigwedge_{i \in #1}}
\newcommand{\conjunc}{\conjuncset{I}}
\newcommand{\disjuncset}[1]{\bigvee_{i \in #1}}
\newcommand{\disjunc}{\disjuncset{I}}
\newcommand{\bisim}{\stackrel{\cdot}{\sim}}

\newcommand{\valid}[2]{#1 \nobreak \, \models \nobreak \, #2 }
\newcommand{\entails}[2]{#1 \, \vdash \, #2}

\newcommand{\trans}[1]{\xrightarrow{#1}}
\newcommand{\Trans}[1]{\stackrel{#1}{\Rightarrow}}
\newcommand{\HTrans}[1]{\Trans{\hat{#1}}}
\newcommand{\freshin}{\#}
\newcommand{\bn}{\operatorname{bn}}
\newcommand{\n}{\operatorname{supp}}

\newcommand{\flogeq}{\stackrel{F/L}{=}}
\newcommand{\logeq}{\stackrel{\cdot}{=}}

\newcommand{\ssubst}[2]{\{{}^{#1}\!/\!{}_{#2}\}}

\newcommand{\new}{\reflectbox{\ensuremath{\mathsf{N}}}}

\newcommand{\effect}[2]{#1 @ #2}
\newcommand{\effectfa}{\effect{f}{A}}

\newcommand{\effectset}{{\mathcal{P}}_{\mbox{fs}}({\mathcal{F}})}
\newcommand{\stateid}{\mbox{\tt id}}
\newcommand{\setid}{\{\stateid\}}
\newcommand{\substx}{\mbox{\tt subst}_x}

\newcommand{\fbisim}{\stackrel{F/L}{\sim}}
\newcommand{\fabisim}{\stackrel{L(\alpha,F,f)/L}{\sim}}
\newcommand{\falogeq}{\stackrel{L(\alpha,F,f)/L}{=}}

\newcommand{\wbisim}{\stackrel{\cdot}{\approx}}
\newcommand{\wlogeq}{\stackrel{\cdot}{\equiv}}

\newcommand{\ac}{{\mbox{\sc ac}}}
\newcommand{\ef}{{\mbox{\sc ef}}}

\newcommand{\finset}[1]{\mathcal{P}_{\mbox{fin}}(#1)}
\newcommand{\nameabs}[1]{<\!\!#1\!\!>\!\!}

\newcommand{\topic}[1]{

\vspace{1em}

\noindent
\textbf{#1.}}

\newcommand{\firsttopic}[1]{\textbf{#1.}}

\begin{document}
 
\title{Modal Logics for Nominal Transition Systems}
\author[Parrow, Borgstr\"om, Eriksson, Gutkovas, Weber]{Joachim Parrow, Johannes
Borgstr\"om, Lars-Henrik Eriksson,\\Ram{\=u}nas Forsberg Gutkovas, Tjark Weber}
\address{Uppsala University,
  Sweden}

\begin{abstract}
We define a general notion of transition system
 where states and action labels can be from arbitrary nominal sets, actions may bind names, and state predicates from an arbitrary logic define properties of states. 
 A Hennessy-Milner logic for these systems is introduced, and proved adequate and expressively complete for bisimulation equivalence. A main technical novelty is the use of finitely supported infinite conjunctions. We show how to treat different bisimulation variants such as early, late, open and weak in a systematic way, explore the folklore theorem that state predicates can be replaced by actions, and make substantial comparisons with related work. The main definitions and theorems have been formalised in Nominal Isabelle.
\end{abstract}
\maketitle

\section{Introduction}

Transition systems are ubiquitous in models of computing. Specifications about what may and must happen during executions are often formulated in a modal logic. There is a plethora of different versions of both transition systems and logics, including a variety of higher-level constructs such as updatable data structures, new name generation, alias generation, dynamic topologies for parallel components etc. In this paper we formulate a general kind of transition system where such aspects can be treated uniformly, and define accompanying modal logics. Our results are on adequacy, i.e., that logical equivalence coincides with bisimilarity, and expressive completeness, i.e., that any bisimulation-preserving property can be expressed.

\topic{States} In any transition system there is a set of {\em states} $P,Q,\ldots$ representing the configurations a system can reach, and a relation telling how a computation can move between them. Many formalisms, for example all process algebras, define languages for expressing states, but in the present paper we shall make no assumptions about any such syntax and just assume that the states form a set. 

\topic{Actions} In systems describing communicating parallel processes, the transitions  are often labelled with {\em actions} $\alpha, \beta$, representing the externally observable effect of the transition. A transition $P \trans{\alpha} P'$ thus says that in state $P$ the execution can progress to $P'$ while performing the action $\alpha$, which is visible to the rest of the world. For example, in CCS these actions are atomic and partitioned into output and input communications. In value-passing calculi the actions can be more complicated, consisting of a channel designation and a value from some data structure to be sent along that channel.

\topic{Scope openings}
With the advent of the pi-calculus~\cite{MPWpi} two important aspects of transitions were introduced: name generation and scope opening. The main idea is that names (i.e.,\ atomic identifiers) can be scoped to represent local resources. They can also be transmitted in actions, to give the environment access to this resource. In the monadic pi-calculus such an action is written $\overline{a}(\nu b)$, to mean that the local name $b$ is exported along the channel $a$.  These names can be subjected to alpha-conversion: if $P \trans{\overline{a}(\nu b)} P'$ and $c$ is a fresh name then also $P \trans{\overline{a}(\nu c)} P'\{c/b\}$, where $P'\{c/b\}$ is $P'$ with all free occurrences of $b$ replaced by~$c$. Making this idea fully formal is not entirely trivial and many papers gloss over it.
In the polyadic pi-calculus several names can be exported in one action, and in psi-calculi arbitrary data structures may contain local names.  In this paper we make no assumptions about how actions are expressed, and just assume that for any action $\alpha$ there is a finite set of names $\bn(\alpha)$, the {\em binding names}, representing exported names. In our formalisation we use nominal sets~\cite{PittsNominalSets}, an attractive theory to reason about objects depending on names on a high level and in a fully rigorous way.

\topic{State predicates}
The final general components of our transition systems are the {\em state predicates}, ranged over by $\varphi$, representing what can be concluded in a given state. For example state predicates can be equality tests of expressions, or connectivity between communication channels.
We write $\entails{P}{\varphi}$ to mean that in state $P$ the state predicate $\varphi$ holds. We make no assumptions about what the state predicates are, beyond the fact that they form a nominal set, and that $\entails{}{}$ does not depend on particular names.

A structure with states, transitions, and state predicates as discussed above we call a {\em nominal  transition system}.

\topic{Hennessy-Milner Logic}
Modal logic has been used since the 1970s to describe how facts evolve through computation. We  use the popular and general branching time logic known as Hennessy-Milner Logic~\cite{DBLP:journals/jacm/HennessyM85} (HML). Here the idea is that an action modality~$\may{\alpha}$ expresses a possibility to perform an action $\alpha$. If $A$ is a formula then $\may{\alpha}A$ says that it is possible to perform $\alpha$ and reach a state where $A$ holds. With conjunction and negation this gives a powerful logic shown to be {\em adequate} for bisimulation equivalence: two processes satisfy the same formulas exactly if they are bisimilar. 
In the general case, conjunction must take an infinite number of operands when the transition systems have states with an infinite number of outgoing transitions.

\topic{Contributions}
Our definition of nominal transition systems is very general since we leave open what the states, transitions and predicates are. The only requirement is that transitions satisfy alpha-conversion.
A technically important point is that we do not assume the usual {\em name preservation principle}, that if $P \trans{\alpha} P'$ then the names occurring in $P'$ must be a subset of those occurring in $P$ and $\alpha$. This means that the results are applicable to a wide range of calculi. For example, the pi-calculus represents a trivial instance where there are no state predicates. CCS represent an even more trivial instance where $\bn$ always returns the empty set. In the fusion calculus and the applied pi-calculus the state contains an environmental part which tells which expressions are equal to which. In the general framework of psi-calculi the states are processes together with  assertions describing their environments. All of these are special cases of nominal transition systems.

We define a modal logic with the $\may{\alpha}$ operator that binds the names in $\bn(\alpha)$. The fully formal treatment of this requires care in ensuring that infinite conjunctions do not exhaust all names, leaving none available for alpha-conversion. All previous works that have considered this issue properly have used uniformly bounded conjunction, i.e.,\ the set of all names in all conjuncts is finite.
Instead of this we use the notion of finite support from nominal sets: the conjunction of an infinite set of formulas is admissible if the set has finite support.
This results in a much greater generality and expressiveness. For example, we can now define quantifiers and the next-step modalities as derived operators. Also the traditional fixpoint operators from the modal mu-calculus are definable through an infinite set of approximants.

We establish {\em adequacy}: that logical equivalence coincides with bisimilarity. Compared to previous such adequacy results our  proof takes a new twist. We further establish {\em expressive completeness}: that all properties (i.e., subsets of the set of states) that are bisimulation-closed can be expressed as formulas. To our knowledge this result is the first of its kind. 

We provide versions of the logic adequate for a whole family of bisimulation equivalences, including late, early, their corresponding congruences,  open, and hyper. Traditionally these differ in how they take name substitutions into account. We define a general kind of effect function, including many different kinds of substitution, and show how all variants can be obtained by varying it. We also show how such effects can be represented directly as transitions. Thus these different kinds of bisimulation can all be considered the same, only on different transition systems.

Weak bisimulation, where the so called silent actions do not count, identifies many more states than strong bisimulation. We provide adequate and expressively complete 
logics for weak bisimulation. In the presence of arbitrary state predicates this is a particularly challenging area and there seems to be more than one alternative formulation. We formally prove the folklore theorem that for strong and weak bisimulation, state predicates can be encoded as actions on self-loops. The counterpart for weak logic is less clear and there appear to be a few different possibilities.

Finally we compare our logic to several other proposed logics for CCS and developments of the pi-calculus. A conclusion is that we can easily represent most of them. The correspondence is not exact because of our slightly different treatment of conjunction, but we certainly gain simplicity and robustness in otherwise complicated logics. We also show how our framework can be applied to obtain logics where none have been suggested previously.

Our main definitions and theorems have been formalised in Nominal Isabelle~\cite{nominal2}. This has required significant new ideas to represent data types with infinitary constructors like infinite conjunction, and their alpha-equivalence classes. As a result we corrected several details in our formulations and proofs, and now have very high confidence in their correctness. The formalisation effort has been substantial, and we consider it a very worthwhile investment. It is hard to measure it precisely since five persons worked on this on and off since 2015, but at a very rough estimate it is less than half of the total effort. The main hurdle was in the very beginning, to get a representation of the infinitely wide formulas with bound names, and accompanying induction schemes.

\topic{Exposition}
In the following section we provide the necessary background on nominal sets. In Section~\ref{sec:nominaltransitionsystems} we present our main definitions and results on nominal transition systems and modal logics. In Sections~\ref{sec:derived-formulas} and~\ref{sec:fixedpoints} we derive useful operators such as quantifiers and fixpoints, and indicate some practical uses. Section~\ref{sec:variants} shows how to treat variants of bisimilarity such as late and open in a uniform way, and in Section~\ref{section:weak} we treat a logic for weak bisimilarity. Section~\ref{sec:state-predicates-versus-actions} presents an encoding of state predicates as actions. In Sections~\ref{sec:applications} and~\ref{sec:related} we compare with related work and demonstrate how our framework can be applied to recover earlier results uniformly. Section~\ref{sec:formalisation} contains a brief account of the formalisation in Nominal Isabelle, and finally Section~\ref{sec:conclusion} concludes with some ideas for further work.

This paper is an extended and revised version of~\cite{DBLP:conf/concur/ParrowBEGW15} and~\cite{DBLP:conf/forte/ParrowWBE17}. 

\section{Background on nominal sets}\label{sec:background}
 {\em Nominal sets}~\cite{PittsNominalSets} is a general theory of objects which depend on names, and in particular formulates the notion of alpha-equivalence when names can be bound. The reader need not know nominal set theory to follow this paper, but some key definitions will make it easier to appreciate our work and we recapitulate them here.

We assume a countably infinite multi-sorted set of atomic identifiers or {\em names} $\nameset$ ranged over by $a, b,\ldots$. \footnote{In~\cite{PittsNominalSets} they are called {\em atoms} and the set of atoms is written $\mathbb{A}$.}
A {\em permutation} is a sort-preserving bijection on names that leaves all but finitely many names invariant. The singleton permutation which swaps names $a$ and $b$ and has no other effect is written $(a\, b)$, and the identity permutation that swaps nothing is written id. Permutations are ranged over by $\pi, \pi'$. The effect of applying a permutation $\pi$ to an object~$X$ is written $\pi \cdot X$. Formally, the permutation action $\cdot$ can be any operation that satisfies $\mbox{id} \cdot X = X$ and $\pi \cdot (\pi' \cdot X) = (\pi \circ \pi') \cdot X$, but a reader may comfortably think of $\pi \cdot X$ as the object obtained by permuting all names in $X$ according to $\pi$.

A set of names $N$ {\em supports} an object $X$ if for all $\pi$ that leave all members of $N$ invariant it holds $\pi \cdot X = X$. In other words, if $N$ supports $X$ then names outside $N$ do not matter to $X$.
If a finite set supports $X$ then there is a unique minimal set supporting $X$, called the {\em support} of $X$, written $\n(X)$, intuitively consisting of exactly the names that matter to $X$.
As an example the set of names textually occurring in a datatype element is the support of that element, and the set of free names is the support of the alpha-equivalence class of the element. Note that in general, the support of a set is not the same as the union of the support of its members. An example is the set of all names; each element has itself as support, but the whole set has empty support since $\pi\cdot\nameset=\nameset$ for any $\pi$.

We write $a \freshin X$, pronounced ``$a$ is fresh for $X$'', for $a \not\in \n(X)$. The intuition is that if $a \freshin X$ then $X$ does not depend on $a$ in the sense that $a$ can be replaced with any fresh name without affecting $X$. If $N$ is a set of names we write $N\freshin X$ for $\forall a \in N\,.\, a \freshin X$.

A {\em nominal set $S$} is a set equipped with a permutation action such that if $X$ is in $S$, then also $\pi\cdot X$ is in $S$, and where each member of $S$  has finite support. A main point is that then each member has infinitely many fresh names available for alpha-conversion.
Similarly, a set of names~$N$ supports a function $f$ on a nominal set if for all $\pi$ that leave $N$ invariant it holds $\pi \cdot f(X) = f(\pi\cdot X)$, and similarly for relations and functions of higher arity. Thus we extend the notion of support to finitely supported functions and relations as the minimal finite support, and can derive general theorems such as $\n(f(X)) \subseteq \n(f) \cup \n(X)$.

An object that has empty support we call {\em equivariant}. For example, a unary
 function~$f$ is equivariant if $\pi \cdot f(X) = f(\pi \cdot X)$ for all $\pi, X$. The intuition is that an equivariant object does not treat any name special.
 
In order to reason about bound names we adopt the notion of generalised name abstraction for finite sets of names (see~\cite[Chapter~4.6]{PittsNominalSets}). Let $S$ be a nominal set and $\finset{\nameset}$ be the finite subsets of names. The {\em nominal alpha-equivalence} $=_\alpha$ is the binary relation on $\finset{\nameset} \times S$ defined by $(N,X) =_\alpha (N',X')$ if there is a permutation $\pi$ such that $\pi \cdot(N,X) = (N',X')$ and $\pi(a)=a$ for all $a \in  \n(X) \setminus N$. It is easily proven that $=_\alpha$ is an equivariant equivalence. The intuition is that the pair $(N,X)$ means  $X$ with bound names $N$, and that alpha-converting the bound names  is allowed as long as there are no captures, i.e., collisions with names free in $X$.

The set of equivalence classes $(\finset{\nameset} \times S)/ \!\!=_\alpha$ is traditionally written $[\finset\nameset] S$, and the equivalence class containing $(N,X)$ is written $\nameabs{N}X$ (which has support $\n(X)\setminus N$). This is an unfortunate clash of notation with Hennessy-Milner logics, where these brackets signify a modal action operator.

\section{Nominal transition systems and Hennessy-Milner logic}
\label{sec:nominaltransitionsystems}

\subsection{Basic definitions}
\label{subsec:basicdefinitions}

We define nominal transition systems, bisimilarity, and a corresponding Hennessy-Milner logic in this subsection.

\begin{defi}
\label{def:nominaltransitionsystem}
A {\em nominal transition system} is characterised by the following
\begin{itemize}
\item \states: A nominal set of {\em states} ranged over by $P,Q$.
\item \predicates: A nominal set of {\em state predicates} ranged over by $\varphi$.
\item
An equivariant binary relation $\vdash$ on \states{} and \predicates. We write $\entails{P}{\varphi}$ to mean that  in state $P$ the state predicate $\varphi$ holds.
\item \actions: A nominal set of {\em actions} ranged over by $\alpha$.
\item An equivariant function $\bn$ from \actions{} to finite sets of names, which for each $\alpha$ returns a subset of $\n(\alpha)$, called the {\em binding names}.
\item An equivariant  subset of $\states \times [\finset\nameset](\actions \times \states)$, called the 
 {\em transition relation}, written $\rightarrow$. If  $(P, \nameabs{N}(\alpha, Q)) \in \mathop{\rightarrow}$ it must hold that $\bn(\alpha)=N$.
\end{itemize}
\end{defi}

We call $\nameabs{\bn(\alpha)}(\alpha, P')$ a {\em residual}.  A residual is thus an alpha-equivalence class of a pair of an action and a state, where the scope of the binding names in the action also contains the state. For $(P, \nameabs{\bn(\alpha)}(\alpha, P')) \in \mathop{\rightarrow}$ we write $P\trans{\alpha} P'$.  This follows the traditional notation in process algebras like the pi-calculus, although it hides the fact that the scope of the names bound by the action extends into the target state.  Transitions satisfy alpha-conversion in the following sense:  If $a \in \bn(\alpha)$, $b \freshin \alpha, P'$ and $ P\trans{\alpha} P'$ then also $P\trans{(a \,b)\cdot\alpha} (a\, b)\cdot P'$ denotes the same transition.

As an example, basic CCS from~\cite{MilnerCCS} is a trivial nominal transition system. Here the \states{} are the CCS agents, \actions{} the CCS actions, $\bn(\alpha) = \emptyset$ for all actions, and $\predicates = \emptyset$. For the pi-calculus, \states{} are the pi-calculus agents, and $\actions{}$ the four kinds of pi-calculus actions (silent, output, input, bound output). In the early semantics $\bn$ returns the empty set except for bound outputs where $\bn(\overline{a}(\nu x)) = \{x\}$. In the late semantics there are actions like $a(x)$ where $x$ is a placeholder so also $\bn(a(x)) = \{x\}$. In the polyadic pi-calculus each action may bind a finite set of names.

In the original terminology of these and similar calculi these occurrences of $x$ are referred to as ``bound.''  We believe a better terminology is ``binding,'' since they bind into the target state.  For higher-order calculi this distinction is important.  Consider an example where objects transmitted in a communication are processes, and a communicated object contains a bound name:
\[P \trans{\overline{a}(\nu b)Q} R\]
The action here transmits the process $(\nu b)Q$ along the channel~$a$.  The name~$b$ is local to~$Q$, so alpha-converting~$b$ to some new name affects only~$Q$.  Normally agents are considered up to alpha-equivalence; this means that~$b$ is not in the support of the action, and we have $\bn(\overline{a} (\nu b)Q) = \emptyset$.

In the same calculus we may also have a different action
\[P \trans{(\nu b) \overline{a}Q} R\]
Again this transmits a process along the channel~$a$, but the process here is just~$Q$.  The name~$b$ is shared between~$Q$ and~$R$, and is extruded in the action.  An alpha-conversion of~$b$ thus affects both~$Q$ and~$R$ simultaneously.  In the action~$b$ is a free name, in the sense that~$b$ is in its support, and it cannot be replaced by another name in the action alone.  Here we have $\bn((\nu b)\overline{a}Q) = \{b\}$.

In all of the above we have $\predicates = \emptyset$ since communication is the only way a process may influence a parallel process, and thus communications are the only things that matter for process equivalence. More general examples come from psi-calculi~\cite{PsiLMCS} where there are so called ``conditions'' representing what holds in different states; those would then correspond to $\predicates$. Other calculi, e.g.~those presented in~\cite{gardner.wischik:explicit-fusions,buscemi.montanari:cc-pi},
 also have mechanisms where processes can influence each other without explicit communication, such as fusions and updates of a constraint store. All of these are straightforward to accommodate as nominal transition systems. Section~\ref{sec:related} contains further descriptions of these and other examples.

\begin{defi}
\label{def:bisim}
A {\em bisimulation} $R$ is a symmetric binary relation on states in a nominal transition system satisfying the following two criteria: $R(P,Q)$ implies
\begin{enumerate}
\item {\em Static implication}: For each $\varphi \in \predicates$, $\entails{P}{\varphi}$ implies  $\entails{Q}{\varphi}$.
\item {\em Simulation}: For all $\alpha, P'$ such that $\bn(\alpha)\freshin  Q$ and $P \trans{\alpha} P'$ there exists some $Q'$ such that
 $Q \trans{\alpha} Q'$ and $R(P',Q')$
\end{enumerate}
We write $P \bisim Q$ to mean that there exists a bisimulation $R$ such that $R(P,Q)$.
\end{defi}

Static implication means that bisimilar states must satisfy the same state predicates; this is reasonable since these can be tested by an observer or parallel process. The simulation requirement is familiar from the pi-calculus. Note that this definition corresponds to ``early'' bisimulation in the pi-calculus. In Section~\ref{sec:variants} we will consider other variants of bisimilarity.

\begin{prop}
\label{bisimequiv}
$\bisim$ is an equivariant equivalence relation.
\end{prop}
\begin{proof}
The proof has been formalised in Isabelle.  Equivariance is a
simple calculation, based on the observation that if~$R$ is a
bisimulation, then $\pi \cdot R$ is a bisimulation.  To prove
reflexivity of~$\bisim$, we note that the identity relation is a bisimulation.
Symmetry is immediate from Definition~\ref{def:bisim}.  To prove
transitivity, we show that the composition of~$\bisim$ with itself is
a bisimulation; the simulation requirement is proved by a considering
an alpha-variant of~$P\trans{\alpha} P'$ where~$\bn(\alpha)$ is fresh
for~$Q$.
\end{proof}

We shall now define a Hennessy-Milner logic for nominal transition systems. As mentioned above, we shall use conjunctions with an infinite set of conjuncts, and need to take care to avoid set-theoretic paradoxes. For example, if we allow conjunctions over arbitrary subsets of formulas then the formulas will not form a set (because its cardinality would then be the same as the cardinality of the set of its subsets).
We thus fix an infinite cardinal~$\kappa$ to bound the conjunctions. We assume that~$\kappa$ is sufficiently large; specifically, we require $\kappa > \aleph_0$ (so that we may form countable conjunctions) and $\kappa > \lvert \states\rvert$. Our logic for nominal transition systems is the following.
\begin{defi}
\label{def:formulas}
The nominal set of formulas $\formset$ ranged over by $A$ is defined by induction as follows:
\[ A \quad ::= \quad
   \conjunc A_i \casesep
    \neg A \casesep
    \varphi \casesep
    \may{\alpha} A \]
\end{defi}
Name permutation distributes over all formula constructors.  
The only binding construct is in $\may{\alpha} A$ where $\bn(\alpha)$ is abstracted and binds into $\alpha$ and $A$. To be completely formally correct we should write the final clause of Definition~\ref{def:formulas} as $\nameabs{\bn(\alpha)}(\may{\alpha}A)$. As with the transitions, we abbreviate this to just~$\may{\alpha}A$, letting the scope of the names bound in~$\alpha$ tacitly extend into~$A$. This means that $\n(\may\alpha A) = \n(\alpha,A)\setminus\bn(\alpha)$. To avoid notation clashes, we shall in the following let name abstractions be implicit in all transitions and logical formulas.

The formula $\conjunc A_i$ denotes the conjunction of the set of formulas $\{A_i\setsep i \in I\}$.
This set must have bounded cardinality,
  by which we mean that $\lvert \{A_i\setsep i \in I\} \rvert < \kappa$.
It is also required that $\{A_i\setsep i \in I\}$ has finite support;
  this is then the support of the conjunction.

For a simple example related to the pi-calculus with the late semantics, consider the formula $\may{a(x)} \may{\overline{b}x} \top$ where $\top$ is the empty conjunction and thus always true. Since $\bn(a(x))) = \{x\}$ the name $x$ is abstracted in the formula, so
$\may{a(y)}\may{\overline{b}y} \top$ is an alpha-variant. The formula says that it must be possible to input something along $a$ and then output it along $b$. In the early semantics, where the input action contains the received object rather than a placeholder, the corresponding formula is
\[\bigwedge_{x\in \nameset} \may{ax}\may{\overline{b}x} \top\]
in other words, the conjunction is over the set of formulas $S = \{\may{ax}\may{\overline{b}x} \top \,|\, x \in \nameset\}$. This set has finite support, in fact the support is just $\{a,b\}$. The reason is that for $c,d \# \{a,b\}$ we have $(c\, d) \cdot S = S$. Note that the formulas in this set have no finite common support, i.e., there is no finite set of names that supports all elements, and thus the conjunction is not a formula in the usual logics for the pi-calculus.

This example highlights one of the main novelties in Definition~\ref{def:formulas}, that we use conjunction of a possibly infinite and finitely supported set of conjuncts. In comparison, the earliest HML for CCS, Hennessy and Milner (1985)~\cite{DBLP:journals/jacm/HennessyM85}, uses finite conjunction, meaning that the logic is adequate only for finitely branching transition systems. In his subsequent book (1989)~\cite{MilnerCCS} Milner admits arbitrary infinite conjunction.
Abramsky (1991)~\cite{DBLP:journals/iandc/Abramsky91} employs a kind of uniformly bounded conjunction, with a finite set of names that supports all conjuncts, an idea that is also used in the first HML for the pi-calculus (1993)~\cite{Milner:1993ys}. Almost all subsequent developments follow one of these three approaches. Note that with arbitrary infinite conjunction the formulas become a proper class and not a set, meaning we cannot reason formally about the formulas using set theories like ZF or HOL.

Our main point is that both finite and uniformly bounded conjunctions are expressively weak, in that the logic is not adequate for the full range of nominal transition systems, and in that quantifiers over infinite structures are not definable. In contrast, our use of finitely supported sets of conjuncts is adequate for all nominal transition systems (cf.~Theorems~\ref{thm:bisimlog} and \ref{thm:logbisim} below) and admits quantifiers as derived operators (cf. Section~\ref{sec:derived-formulas} below). Universal quantification over names $\forall x \in \mathcal{N}. A(x)$ is usually defined to mean that $A(n)$ must hold for all $n \in \mathcal{N}$. We can define this as the (infinite) conjunction of all these $A(n)$. This set of conjuncts is not uniformly bounded if $n \in \n(A(n))$. But it is supported by $\n(A)$ since, for any permutation $\pi$ not affecting $\n(A)$ we have $\pi \cdot A(n) = A(\pi(n))$ which is also a conjunct; thus the  set of conjuncts is unaffected by $\pi$.

The validity of a formula $A$ for a state $P$ is written
$\valid{P}{A}$ and is defined by recursion over~$A$ as follows.
\begin{defi}
\label{def:validity}
\[\begin{array}{ll}
\valid{P}{\conjunc A_{i}} & \mbox{if for all $A\in \{A_{i}\setsep i\in I\}$ it holds that $\valid{P}{A}$} \\
\valid{P}{\neg A} & \mbox{if not $\valid{P}{A}$}\\
\valid{P}{\varphi} & \mbox{if $\entails{P}{\varphi}$} \\
\valid{P}{\may{\alpha} A} &  \mbox{if there exists $P'$ such that $P \trans{\alpha} P'$ and $\valid{P'}{A}$}
\end{array}\]
\end{defi}
In the last clause we assume that $\may{\alpha}A$ is a representative of its alpha-equivalence class such that
$\bn(\alpha) \freshin P$.

\begin{lem}
$\valid{}{}$ is equivariant.
\end{lem}
\begin{proof}
By the Equivariance Principle in Pitts
(2013)~\cite[page~21]{PittsNominalSets}.  A detailed proof that
verifies $\valid{P}{A} \iff \valid{\pi\cdot P}{\pi\cdot A}$ for any
permutation~$\pi$ has been formalised in Isabelle.  The proof proceeds
by structural induction on~$A$, using equivariance of all involved
relations.  For the case~$\may{\alpha} A$ in particular, we use the
fact that if $\may{\alpha'} A' = \may{\alpha} A$, then
$\may{\pi\cdot\alpha'} (\pi\cdot A') = \may{\pi\cdot\alpha} (\pi\cdot
A)$.
\end{proof}

\begin{defi}
\label{def:logequiv}
Two states $P$ and $Q$ are {\em logically equivalent}, written $P \logeq Q$, if
for all $A$ it holds that $\valid{P}{A}$ iff $\valid{Q}{A}$.
\end{defi}

\subsection{Logical adequacy}
\label{sec:adequacy}

We show that the logic defined in Section~\ref{subsec:basicdefinitions} is adequate for bisimilarity; that is, bisimilarity and logical equivalence coincide.

\begin{thm}
\label{thm:bisimlog}
$P \bisim Q\ \implies\ P \logeq Q$
\end{thm}
\begin{proof}
The proof has been formalised in Isabelle.  Assume~$P \bisim
Q$.  We show~$\valid{P}{A} \iff \valid{Q}{A}$ by structural induction
on~$A$.
\begin{enumerate}
\item Base case: $A = \varphi$. Then $\valid{P}{A} \iff
  \entails{P}{\varphi} \iff \entails{Q}{\varphi} \iff \valid{Q}{A}$ by
  static implication and symmetry of~$\mathop{\bisim}$.
\item Inductive steps $\conjunc A_i$ and $\neg A$: immediate by
  induction.
\item Inductive step $\may{\alpha} A$: Assume
  $\valid{P}{\may{\alpha}A}$.  Then for some alpha-variant 
  $\may{\alpha'}A' = \may{\alpha}A$, $\exists P' \,.\, P
  \trans{\alpha'} P'$ and $\valid{P'}{A'}$.  Without loss of
  generality we assume also $\bn(\alpha') \freshin Q$, otherwise just
  find an alpha-variant of~$\may{\alpha'}A'$ where this holds.  Then
  by simulation $\exists Q' \,.\, Q \trans{\alpha'} Q'$ and $P' \bisim
  Q'$.  By induction and~$\valid{P'}{A'}$ we get~$\valid{Q'}{A'}$,
  whence by definition~$\valid{Q}{\may{\alpha}A}$.  The proof
  of~$\valid{Q}{\may{\alpha}A} \implies \valid{P}{\may{\alpha}A}$ is
  symmetric, using the fact that~$P \bisim Q$ entails~$Q \bisim P$. \qedhere
\end{enumerate}
\end{proof}

The converse result uses the idea of {\em distinguishing formulas}.

\begin{defi}
\label{def:distinguishing}
A {\em distinguishing formula} for $P$ and $Q$ is a formula $A$ such that $\valid{P}{A}$ and not $\valid{Q}{A}$.
\end{defi}

The following lemma says that we can find such a formula $B$ where, a bit surprisingly, the support of $B$ does not depend on $Q$.

\begin{lem}
\label{lemma:distsup}
If $P\not\logeq Q$ then there exists a distinguishing formula $B$ for $P$ and $Q$ such that $\n(B) \subseteq \n(P)$.
\end{lem}
\newcommand{\thepermutations}{\Pi_P}
\begin{proof}
The proof has been formalised in Isabelle.  If $P\not\logeq Q$ then there exists a distinguishing formula $A$ for $P$ and $Q$, i.e.,  $\valid{P}{A}$ and not $\valid{Q}{A}$.
Let $\thepermutations= \{\pi \setsep n \in\n(P) \Rightarrow \pi(n)=n\}$ be the group of name permutations that leave $\n(P )$ invariant and let $\mathcal{B}$ be the $\thepermutations$-orbit of
$A$, i.e.,
\[{\mathcal{B}} = \{ \pi \cdot A \setsep \pi \in \thepermutations\}\]
 In the terminology of Pitts~\cite{PittsNominalSets} ch.~5,  $\mathcal{B}$ is
$\mbox{hull}_{\n(P)} A$.
Clearly, if $a,b \freshin P$ and $\pi \in \thepermutations$ then $(a\, b) \circ \pi \in \thepermutations$. Thus $(a\, b) \cdot \mathcal{B} = \mathcal{B}$ and therefore $\n(\mathcal{B}) \subseteq \n(P )$.
This means
that the formula $B = \bigwedge {\mathcal{B}}$ is well-formed and $\n(B)
\subseteq \n(P)$. For all $\pi \in \thepermutations$ we have by definition $P =\pi \cdot P$ and by equivariance $\pi \cdot P \models \pi \cdot A$, i.e., $ P \models \pi \cdot A$.
Therefore
$\valid{P}{B}$.  Furthermore, since the identity permutation is in
$\thepermutations$ and not $\valid{Q}{A}$ we get not
$\valid{Q}{B}$.
\end{proof}
Note that in this proof, $B$ uses a conjunction that is not uniformly bounded.

\begin{thm}
 \label{thm:logbisim}
$P \logeq Q\ \implies\ P \bisim Q$
\end{thm}
\begin{proof}
The proof has been formalised in Isabelle.  We establish that
$\logeq$ is a bisimulation.  Obviously it is symmetric.  So assume
$P\logeq Q$, we need to prove the two requirements on a bisimulation.
\begin{enumerate}
\item Static implication. $\entails{P}{\varphi}$ iff
  $\valid{P}{\varphi}$ iff $\valid{Q}{\varphi}$ iff
  $\entails{Q}{\varphi}$.
\item Simulation. The proof is by contradiction. Assume that $\logeq$
  does not satisfy the simulation requirement. Then there exist
  $P,Q,P',\alpha$ with $\bn(\alpha) \freshin Q$ such that $P\logeq Q$
  and $P \trans{\alpha}P'$ and, letting ${\mathcal{Q}} = \{Q' \setsep Q
  \trans{\alpha} Q'\}$, for all $Q' \in {\mathcal{Q}}$ it holds that
  $P'\not\logeq Q'$. Assume $\bn(\alpha)\freshin P$, otherwise just find an alpha-variant of the transition satisfying this requirement. By $P' \not\logeq Q'$, for all $Q' \in {\mathcal{Q}}$ there exists a
  distinguishing formula for $P'$ and $Q'$. The formula may depend
  on~$Q'$, and by Lemma~\ref{lemma:distsup} we can find such a
  distinguishing formula $B_{Q'}$ for $P'$ and $Q'$ with $\n(B_{Q'})
  \subseteq \n(P')$.  Therefore the formula
 \[B = \bigwedge_{Q' \in {\mathcal{Q}}} B_{Q'}\]
 is well-formed with support included in $\n(P')$. We thus get that
 $\valid{P}{\may{\alpha} B}$ but not $\valid{Q}{\may{\alpha} B}$,
 contradicting $P\logeq Q$. \qedhere
\end{enumerate}
\end{proof}

This proof of the simulation property is different from other such proofs in the literature. For finitely branching transition systems, $\mathcal Q$ is finite so a finite conjunction is enough to define~$B$. For transition systems with the  name preservation property, i.e., that if $P \trans{\alpha} P'$ then $\n(P') \subseteq  \n(P) \cup \n(\alpha)$, a uniformly bounded conjunction suffices with common support $\n(P )  \cup \n(Q) \cup \n(\alpha)$. Without the name preservation property, we here use a non-uniformly bounded conjunction in Lemma~\ref{lemma:distsup}.

\subsection{Expressive completeness}
\label{subsec:expressive-completeness}

Theorem~\ref{thm:logbisim} shows that for every pair of non-bisimilar states, there is a formula that distinguishes them.  We can prove an even stronger result: the logic contains characteristic formulas for the equivalence classes of bisimilarity. Moreover, every finitely supported set of states that is closed under bisimilarity has a characteristic formula.

We first strengthen Lemma~\ref{lemma:distsup} by showing that there is an equivariant function that yields distinguishing formulas for non-equivalent states.

\begin{lem}
  \label{lemma:distinguishing-equivariant}
  If $P\not\logeq Q$, write~$B_{P,Q}$ for a distinguishing formula for~$P$ and~$Q$ such that $\n(B_{P,Q}) \subseteq \n(P)$.  Then 
  \[D(P,Q) := \bigwedge_\pi \pi^{-1} \cdot B_{\pi \cdot P, \pi \cdot Q}\]
  defines a distinguishing formula for~$P$ and~$Q$ with support included in~$\n(P)$.  Moreover, the function $D$ is equivariant.
\end{lem}

\begin{proof}
  The proof has been formalised in Isabelle.  Assume~$P\not\logeq Q$.  For any permutation~$\pi$, $\pi \cdot P \not\logeq \pi \cdot Q$ by equivariance of~$\mathop{\logeq}$.  Thus the formulas~$B_{\pi \cdot P, \pi \cdot Q}$ exist by Lemma~\ref{lemma:distsup}, and $\n(\pi^{-1}\cdot B_{\pi \cdot P, \pi \cdot Q}) \subseteq \n(P)$.  Hence the conjunction~$D(P,Q)$ is well-formed with support included in~$\n(P)$.

  From~$\valid{\pi \cdot P}{B_{\pi \cdot P, \pi \cdot Q}}$ we have $\valid{P}{\pi^{-1} \cdot B_{\pi \cdot P, \pi \cdot Q}}$ by equivariance of~$\mathop{\models}$.  Hence $\valid{P}{D(P,Q)}$. Also not $\valid{Q}{B_{P,Q}}$, hence (by considering the identity permutation) not $\valid{Q}{D(P,Q)}$.  Therefore~$D(P,Q)$ is a distinguishing formula for~$P$ and~$Q$.

  For equivariance, we have
  \begin{eqnarray*}
    \rho \cdot D(P, Q) 
    &=& \bigwedge_\pi (\rho \circ \pi^{-1}) \cdot B_{\pi \cdot P, \pi \cdot Q} \\
   &\stackrel{\{ \sigma \,:=\, \pi\circ\rho^{-1}\}}=& \bigwedge_\sigma \sigma^{-1} \cdot B_{(\sigma \circ \rho) \cdot P, (\sigma \circ \rho) \cdot Q} \\
    &=& D(\rho \cdot P, \rho \cdot Q)
  \end{eqnarray*}
  as required.
\end{proof}

\begin{defi}
  \label{def:characteristic}
  A \emph{characteristic formula} for $P$ is a formula $A$ such that for all~$Q$, 
\[P \bisim Q \quad\mbox{iff}\quad\valid{Q}{A}\]
\end{defi}

Characteristic formulas can be obtained as a (possibly infinite) conjunction of distinguishing formulas.

\begin{lem}
  \label{lemma:characteristic}
  Let~$D$ be defined as in Lemma~\ref{lemma:distinguishing-equivariant}.  The formula 
 \[\mathrm{Char}(P) := \bigwedge_{P\not\logeq Q} D(P,Q)\]
  is a characteristic formula for~$P$.
\end{lem}

\begin{proof}
  The proof has been formalised in Isabelle.  By Lemma~\ref{lemma:distinguishing-equivariant}, $\n(D(P,Q)) \subseteq \n(P)$ for all~$Q$ with~$P\not\logeq Q$.  Thus the conjunction is well-formed with support included in~$\n(P)$.

  We show that $P\logeq Q$ iff~$\valid{Q}{\mathrm{Char}(P)}$.  The lemma then follows because bisimilarity and logical equivalence coincide (Theorems~\ref{thm:bisimlog} and~\ref{thm:logbisim}).

  $\Longrightarrow$:  Assume~$P\logeq Q$.  By the definition of distinguishing formulas, $\valid{P}{D(P,Q')}$ for all~$Q'$ with~$P\not\logeq Q'$.  Hence~$\valid{P}{\mathrm{Char}(P)}$.  Therefore~$\valid{Q}{\mathrm{Char}(P)}$ by logical equivalence.

  $\Longleftarrow$:  Assume~$\valid{Q}{\mathrm{Char}(P)}$. If $P\not\logeq Q$,
  then~$\mathrm{Char}(P)$ has $D(P,Q)$ as a conjunct.
  By assumption we then have $\valid{Q}{D(P,Q)}$,
  which is a contradiction since~$D(P,Q)$ is a distinguishing formula for~$P$ and~$Q$ by Lemma~\ref{lemma:distsup}.
\end{proof}

\begin{lem}
  \label{lemma:characteristic-equivariant}
  Let~$\mathrm{Char(P)}$ be defined as in Lemma~\ref{lemma:characteristic}.  The function
  \[\mathrm{Char}\colon P \mapsto \mathrm{Char}(P)\]
   is equivariant.
\end{lem}

\begin{proof}
  The proof has been formalised in Isabelle.  We verify that $\pi \cdot \mathrm{Char}(P) = \mathrm{Char}(\pi\cdot P)$ for any permutation~$\pi$ using the equivariance of~$\mathop{\bigwedge}$, $\mathop{\not\logeq}$, and~$D$.
\end{proof}

\begin{defi}
A {\em property} of states is a subset of the states. A property $S$ is {\em well-formed} if it is finitely supported and bisimulation closed, i.e., if $P\in S$ and  $P \sim Q$ then also $Q \in S$.
\end{defi}
Well-formed properties can be described as a (possibly infinite) disjunction of characteristic formulas.  (Disjunction~$\bigvee_{i\in I} A_i$ is defined in the usual way as $\neg\bigwedge_{i\in I}\neg A_i$; see Section~\ref{sec:derived-formulas} for further details.)
\begin{thm}[Expressive Completeness]
  \label{thm:expressive-completeness}
  Let~$S$ be a well-formed property.  Then there exists a formula $A$
  such that
  \[\valid{P}A \quad\mbox{iff}\quad P \in S\]
 \end{thm}
\begin{proof}
  The proof has been formalised in Isabelle.  Let
  \[A= \bigvee_{P'\in S} \mathrm{Char}(P')\]
  $S$ is finitely supported by assumption, and by the equivariance of
  $\mathrm{Char}$ (Lemma~\ref{lemma:characteristic-equivariant}) we
  have $\n(\{ \mathrm{Char}(P')\mid P'\in S \})\subseteq\n(S)$.  Hence
  the disjunction is well-formed.

  Assume~$P\in S$.  Since~$\valid{P}{\mathrm{Char}(P)}$ we get~$\valid{P}{\bigvee_{P'\in S} \mathrm{Char}(P')}$.  Conversely, assume that~$\valid{P}{\bigvee_{P'\in S} \mathrm{Char}(P')}$.  Then for some~$P'\in S$, $\valid{P}{\mathrm{Char}(P')}$.  By Lemma~\ref{lemma:characteristic}, $P' \bisim P$.  Hence~$P\in S$ because~$S$ is closed under bisimilarity.
\end{proof}

There are many relative expressiveness results in connection with logics in general and modal logics in particular. A classic example is van Benthem's result that modal logic is the bisimulation-invariant fragment of first-order logic~\cite{vanBenthem}. Theorem~\ref{thm:expressive-completeness} is very different in that it establishes a kind of absolute expressiveness: any (finitely supported and bisimulation closed) set of states can be characterized by a single formula. This is clearly impossible in any logic with only countably many formulas, since the set of sets of states may be uncountable.

\newcommand{\interp}[1]{\llbracket#1\rrbracket}
\newcommand{\INTERP}{\interp{\text{~}}}
\newcommand{\interpe}[2][\varepsilon]{\interp{#2}_{#1}}
\newcommand{\encodefix}[1]{\underline{#1}}
\newcommand{\cof}{\textsc{cof}}
\newcommand{\deq}{\mathrel{:=}}
\newcommand{\powfin}[1]{{\mathcal P}_{\mathrm{fs}}(#1)}
\newcommand{\CTLstar}{$\mbox{CTL}^{*}$}
\newcommand{\lfp}{\operatorname{lfp}}

\section{Derived formulas}
\label{sec:derived-formulas}
\noindent\firsttopic{Dual connectives}
We define logical disjunction $\bigvee_{i\in I}A_i$ in the usual way as
$\neg \bigwedge_{i\in I}\neg A_i$, when $\{A_i\setsep i\in I\}$
has bounded cardinality and finite support. A special case is $I=\{1,2\}$: we then write
$A_1 \land A_2$ instead of $\bigwedge_{i\in I}A_i$, and dually for $A_1 \lor A_2$.
We write $\top$ for the empty conjunction $\bigwedge_{i\in\emptyset}$, and $\bot$ for $\neg\top$.
We also write $A\implies B$ for $B\lor\neg A$.

The must modality $\must{\alpha}A$ is defined as
$\neg\may\alpha\neg A$, and requires $A$ to hold after every possible
$\alpha$-labelled transition from the current state.
  Note that $\bn(\alpha)$ bind into $A$.
By the semantics of the logic, $\must\alpha (A\land B)$ is equivalent to  $\must\alpha
A\land \must\alpha B$, and dually $\may\alpha (A\lor B)$ is equivalent to  $\may\alpha
A\lor \may\alpha B$.

\topic{Quantifiers}
\newcommand{\quantdot}{.}
Let $S$ be any finitely supported set of bounded cardinality and use $v$ to range over members of $S$.
Write $A\ssubst{v}{x}$ for the substitution of $v$ for the name $x$ in $A$, and assume this substitution function is equivariant.
Then we define
$\forall x\in S \quantdot A$ as $\bigwedge_{v\in S}A\ssubst vx$. There is not necessarily a common finite support for
 the formulas $A\ssubst vx$, for example if $S$ is some term algebra over names, but
the set $\{A\ssubst vx \setsep v\in S\}$ has finite support bounded by
$\{x\}\cup\n(S)\cup\n(A)$.
In our examples in Section~\ref{sec:related}, substitution is defined inductively on the structure of formulas,
based on primitive substitution functions for actions and state predicates,
which are capture-avoiding and preserve the binding names of actions.

Existential quantification $\exists x\in S\quantdot A$ is defined as the dual $\neg \forall x\in S\quantdot \neg A$.
When $X$ is a metavariable used to range over a nominal set $\mathcal{X}$, we simply write $X$ for ``$X\in \mathcal{X}$''.
As an example, $\forall a\quantdot A$ means that the formula $A\ssubst na$ holds for all names $n\in\nameset$.

\topic{New name quantifier}
The new name quantifier $\new x.A$~\cite{pitts:nominal-logic} intuitively states that
$\valid{P}{A{\ssubst{n}x}}$ holds where $n$ is a fresh name for $P$.  For
example, suppose we have actions of the form $a\,b$ for input, and
$\overline{a}\,b$ for output where $a$ and $b$ are free names, then the formula $
\new x.  \must{a\,x}\may{\overline{b}\,x}\top $ expresses that whenever a process
inputs a fresh name $x$ on channel $a$, it has to be able to output that name on
channel~$b$. If the name received is not fresh (i.e., already present in $P$) then $P$ is not required to do anything. Therefore this formula is weaker than
$\forall x \quantdot
\must{a\,x}\may{\overline{b}\,x}\top$.

Since $A$ and $P$  have finite support, if $\valid{P}{A\ssubst{n}x}$ holds for some $n$ fresh for $P$, by equivariance it also holds for almost all $n$, i.e., all but finitely many $n$. Conversely, if it holds for almost all $n$, it must hold for some $n\freshin \n(P)$. Therefore $\new x$ is often pronounced ``for almost all $x$''.
In other words,  $P\models\new x.A$ holds if $\{x \setsep P\models A(x)\}$ is a cofinite set of
names \cite[Definition 3.8]{PittsNominalSets}.

To avoid the need for a substitution function,
we here define the new name quantifier using name swapping $(a\;n)$.
Letting $\cof=\{ S\subseteq\nameset \setsep \nameset \setminus S\text{ is finite} \}$
we thus encode  $\new x.A$ as $\bigvee_{S\in\cof}\bigwedge_{n\in S\setminus\n(A)}(x\,n)\!\cdot\! A$. This formula
states there is a cofinite set of names $n$ such that for all of them that are fresh for $A$,
$(x\,n)\!\cdot\! A$ holds. The support of $\bigwedge_{n\in S\setminus\n(A)}(x\,n)\!\cdot\! A$  is bounded by
$(\nameset \setminus S) \cup \n(A)$ where $S \in \cof$, and
the support of the encoding $\bigvee_{S\in\cof}\bigwedge_{n\in S\setminus\n(A)}(x\,n)\!\cdot\! A$ is
bounded by $\n(A)$.

\topic{Next step}
We can generalise the action modality to sets of actions:
if $T$ is a finitely supported set of actions that has bounded cardinality, we write $\may{T}A$ for $\bigvee_{\alpha \in T}\may\alpha A$.  The
support of $\{\may\alpha A \setsep \alpha \in T\}$ is bounded by
$\n(T)\cup\n(A)$ and thus finite.  Dually, we write $\must{T}A$ for
$\neg\may{T}\neg A$, denoting that $A$ holds after all transitions
with actions in $T$.  Note that binding names of actions in $T$ bind
into $A$, and so $\may{\alpha}A$ is equivalent to $\may{\{\alpha\}}A$
for any~$\alpha$.

To encode the next-step modality, let $\actions_A=\{\alpha \setsep
\bn(\alpha)\freshin A\}$.  Note that $\n(\actions_A) \subseteq \n(A)$
is finite.  If $\kappa$ (Definition~\ref{def:formulas}) is larger than
$\lvert \actions_A\rvert$, we write $\may{\,}A$ for
$\may{\actions_A}A$, meaning that we can make some (non-capturing)
transition to a state where $A$ holds.  As an example, $\may{\,}\top$
means that the current state is not deadlocked.  The dual modality
$\must{\,}A=\neg\may{\,}\neg A$ means that $A$ holds after every
transition from the current state.  Larsen~\cite{LarsenHMLrecursion}
uses the same approach to define next-step operators in HML, though
his version is less expressive since he uses a finite action set to
define the next-step modality.

\section{Fixpoint operators}
\label{sec:fixedpoints}

Fixpoint operators are a way to introduce recursion into a logic.  For
example, they can be used to concisely express safety and liveness
properties of a transition system, where by safety we mean that some
invariant holds for all reachable states, and by liveness that some
property will eventually hold.  Kozen~\cite{Kozen:1983} introduced the
least ($\mu X. A$) and the greatest ($\nu X. A$) fixpoint operators in
modal logic.

By combining the fixpoint and next-step operators, we can encode the
temporal logic CTL~\cite{Clarke:1981}, following
Emerson~\cite{Emerson:1997}.  The CTL formula $\mathsf{AG}\, A$, which
states that~$A$ holds along all paths, is defined as $\nu X. A \land
\must{\,} X$.  Dually the formula $\mathsf{E F}\,A$, stating the there
is some path where $A$ eventually holds, is defined $\mu X. A \lor
\may{\,} X$.  These are special cases of more general formulas: the
formula $\mathsf{A}[A \,\mathsf{U}\, B]$ states that for all paths~$A$
holds until~$B$ holds, and dually $\mathsf{E}[A \,\mathsf{U}\, B]$
states that there is a path along which~$A$ holds until~$B$.  They are
encoded as $\nu X. B \lor (\must{\,}X \land A)$ and $\mu X. B \lor
(\may{\,}X \land A)$, respectively.
For example, deadlock-freedom is given by the CTL formula
$\mathsf{AG}\,\may{\,}\top$ expressing that every reachable state has
a transition. The encoding of this formula is $\nu X. \may{\,}\top
\land \must{\,} X$.

We extend the logic of Definition~\ref{def:formulas} with the least fixpoint
operator and give meaning to formulas as sets of satisfying states. We show
that the meaning of the fixpoint operator is indeed a fixpoint. We proceed
to show that the least fixpoint operator can be directly encoded in the logic.
The greatest fixpoint operator can then be expressed as the dual of the least fixpoint.

\subsection{Logic with the least fixpoint operator}

\begin{defi}
  We extend the nominal set of formulas with the least fixpoint operator:
  \[
  A \quad ::= \quad
  \conjunc A_i \casesep
  \neg A \casesep
  \varphi \casesep
  \may{\alpha} A \casesep
  X \casesep
  \mu X. A
  \]
  where~$X$ ranges over a countably infinite set of equivariant (i.e.,
  $\pi \cdot X = X$ for all~$\pi$) variables.  We require that all
  occurrences of a variable~$X$ in a formula~$\mu X. A$ are in the
  scope of an even number of negations.
\end{defi}

An occurrence of a variable~$X$ in~$A$ is said to be free if it is not
a subterm of some~$\mu X.B$.  We say that a formula~$A$ is closed if
for every variable~$X$, none of its occurrences in~$A$ are free.  We
use a capture-avoiding substitution function~$[A/X]$ on formulas that
substitutes~$A$ for the free occurrences of the variable $X$.  In
particular, $(\may\alpha B)[A/X] = \may\alpha (B[A/X])$ when
$\bn(\alpha)$ is fresh for~$A$.

We give a semantics to formulas containing variables and fixpoint
modalities as sets of satisfying states.

\begin{defi} \label{def:interpretation}
  A \emph{valuation function}~$\varepsilon$ is a finitely supported
  map from variables to (finitely supported) sets of states.
  We write $\varepsilon[X \mapsto S]$ for the valuation function that
  maps~$X$ to~$S$, and any variable~$X' \neq X$ to~$\varepsilon(X')$.

  \noindent We define the \emph{interpretation} of formula~$A$ under
  valuation~$\varepsilon$ by structural induction as the set of
  states~$\interpe{A}$:
  \begin{longtable*}{RCL}
    \interpe{\conjuncset{I} A_i} &=& \bigcap_{i\in I} \interpe{A_i} \\
    \interpe{\neg A} &=& \states \setminus \interpe{A} \\
    \interpe{\varphi} &=& \Set{ P \setsep \entails{P}\varphi } \\
    \interpe{\may{\alpha} A} &=& \Set{ P \setsep \exists \alpha' \, A' \, P' \,.\, \may{\alpha} A = \may{\alpha'} A' \wedge \bn(\alpha') \freshin P,\varepsilon \wedge P \trans{\alpha'} P' \wedge P' \in \interpe{A'}} \\
    \interpe{X} &=& \varepsilon(X) \\
    \interpe{\mu X. A} &=& \bigcap\Set{ S \in\powfin{\states} \setsep \interp{A}_{\varepsilon[X \mapsto S]} \subseteq S }
  \end{longtable*}
  We write $\INTERP$ for the function $(A,\varepsilon)\mapsto\interpe{A}$.
\end{defi}

\begin{lem} \label{lem:interp-equiv}
  $\INTERP$ is equivariant.
\end{lem}

\begin{proof}
  By the Equivariance Principle~\cite[page~21]{PittsNominalSets}.
\end{proof}

\begin{lem}
  For any formula $A$ and valuation function $\varepsilon$,
  $\interpe[\varepsilon]{A}\in\powfin{\states}$.
\end{lem}

\begin{proof}
  By equivariance of $\INTERP$,
  $\n(\interpe[\varepsilon]{A}) \subseteq \n(A) \cup \n(\varepsilon)$.
\end{proof}

Temporal operators such as ``eventually'' can be encoded using the
least fixpoint operator.  For instance, the formula $\mu
X. \may{\alpha}X \lor A$ states that~$A$ eventually holds on some path
labelled with~$\alpha$.
We define the greatest fixpoint operator $\nu X. A$ in terms of the
least as $\neg \mu X. \neg A[\neg X/X]$.  Using the greatest fixpoint
operator we can state global invariants: $\nu X. \must\alpha X \land
A$ expresses that~$A$ holds along all paths labelled with~$\alpha$.
We can freely mix the fixpoint operators to obtain formulas like $\nu
X. \must\alpha X \land (\mu Y . \may\beta Y \lor A)$, which means that
for each state along any path labelled with~$\alpha$, a state
where~$A$ holds is reachable along a path labelled with~$\beta$.

As sanity checks for our definition, we prove that the interpretation
of formulas without fixpoint modalities is unchanged, and that the
interpretation of the formula~$\mu X.A$ is indeed the least fixpoint
of the function~$F_A^{\varepsilon}\colon
S\mapsto\interp{A}_{\varepsilon[X\mapsto S]}$.

\begin{prop} \label{prop:valid-iff-interpe}
  Let~$A$ be a formula as in Definition~\ref{def:formulas}.  Then for
  any valuation function~$\varepsilon$ and state~$P$, $\valid{P}{A}$
  if and only if~$P \in \interpe{A}$.
\end{prop}

\begin{proof}
  By structural induction on~$A$.  The clauses for~$X$ and~$\mu X.A'$
  in Definition~\ref{def:interpretation} are not used.  The
  interesting case is
  \begin{description}
  \item[Case $\may{\alpha}A'$] Assume $\valid{P}{\may{\alpha}A'}$.
    Without loss of generality assume also $\bn(\alpha) \freshin P,
    \varepsilon$, otherwise just find an alpha-variant
    of~$\may{\alpha}A'$ where this holds.  From
    Definition~\ref{def:validity} we obtain~$P'$ such that $P
    \trans\alpha P'$ and~$\valid{P'}{A'}$.  Then~$P' \in \interpe{A'}$
    by the induction hypothesis, hence $P \in
    \interpe{\may{\alpha}A'}$ by Definition~\ref{def:interpretation}.

    Next, assume~$P \in \interpe{\may{\alpha}A'}$.  From
    Definition~\ref{def:interpretation} we obtain an alpha-variant
    $\may{\alpha'}A'' = \may{\alpha}A'$ and~$P'$ such that
    $\bn(\alpha') \freshin P$, $P \trans{\alpha'} P'$ and $P' \in
    \interpe{A''}$.  Then~$\valid{P'}{A''}$ by the induction
    hypothesis.  Hence $\valid{P}{\may{\alpha'}A''} = \may{\alpha}A'$
    by Definition~\ref{def:validity}. \qedhere
  \end{description}
\end{proof}

\begin{lem} \label{lem:mu-fininite-supp}
  For any formula~$\mu X.A$ and valuation function~$\varepsilon$,
  $F_{A}^{\varepsilon}$ has finite support.
\end{lem}

\begin{proof}
  By equivariance of $\INTERP$,
  $\n(F_{A}^{\varepsilon})\subseteq\n(A)\cup\n(\varepsilon).$
\end{proof}

\begin{lem} \label{lem:mu-monotonic}
  For any formula $\mu X. A$ and 
  for any
  valuation function~$\varepsilon$, the
  function $F_{A}^\varepsilon \colon \powfin\states \to
  \powfin\states$ is monotonic with respect to subset inclusion.
\end{lem}

\begin{proof}
  By structural induction on~$A$, for arbitrary~$\varepsilon$.  Let
  $S$, $T \in \powfin\states$ such that~$S \subseteq\nobreak T$.  We
  prove a more general statement: if all occurrences of~$X$ in~$A$ are
  positive (i.e., within the scope of an even number of negations),
  $F_{A}^\varepsilon(S) \subseteq F_{A}^\varepsilon(T)$, and if all
  occurrences of~$X$ in~$A$ are negative (i.e., within the scope of an
  odd number of negations), $F_{A}^\varepsilon(T) \subseteq
  F_{A}^\varepsilon(S)$.  The interesting case is
  \begin{description}
  \item[Case $\mu X'.A'$] If $X = X'$, note that for any~$U$,
    $V\in\powfin{\states}$, $\varepsilon[X\mapsto U][X\mapsto\nobreak
    V] = \varepsilon[X\mapsto V]$.  Therefore, $\interp{\mu
    X'. A'}_{\varepsilon[X \mapsto S]} = \interp{\mu
    X'. A'}_{\varepsilon[X \mapsto T]}$ is immediate from
    Definition~\ref{def:interpretation}.

    Otherwise, $X \neq X'$.  Suppose that all occurrences of~$X$
    in~$A$ are positive.  Then all occurrences of~$X$ in~$A'$ are
    positive, and for any~$V\in\powfin{\states}$,
    $\interp{A'}_{\varepsilon[X'\mapsto V][X\mapsto S]} \subseteq
    \interp{A'}_{\varepsilon[X'\mapsto V][X\mapsto T]}$ by the
    induction hypothesis applied to~$A'$ and~$\varepsilon[X'\mapsto
      V]$.  Since $X\neq\nobreak X'$, for any~$U$,
    $V\in\powfin{\states}$, $\varepsilon[X\mapsto U][X'\mapsto V] =
    \varepsilon[X'\mapsto V][X\mapsto U]$.  Thus,
      \begin{multline*}
        \qquad\quad \interp{\mu X'. A'}_{\varepsilon[X \mapsto S]} = \bigcap\Set{
          S' \in\powfin{\states} \setsep \interp{A'}_{\varepsilon[X
            \mapsto S][X'\mapsto S']} \subseteq S' } \subseteq\\
        \bigcap\Set{ S' \in\powfin{\states} \setsep
          \interp{A'}_{\varepsilon[X\mapsto T][X' \mapsto S']}
          \subseteq S' } = \interp{\mu X'. A'}_{\varepsilon[X \mapsto
          T]}.
      \end{multline*}
      The case where all occurrences of~$X$ in~$A$ are negative is
      similar.\qedhere
  \end{description}
\end{proof}

We use a nominal version of Tarski's fixpoint
theorem~\cite{tarski1955} to show existence, uniqueness, and the
construction of the least fixpoint of~$F_{A}^{\varepsilon}$.  Note
that the usual Tarski fixpoint theorem does not apply, since the
lattice~$\powfin\states$ is not necessarily complete.  (For a simple
counterexample, consider $\states = \nameset$: sets that are neither finite nor cofinite are not elements of $\powfin\nameset$.)

\begin{thm} \label{thm:tarski}
  Suppose $X$ is a nominal set, and $f\colon\powfin X \to \powfin X$
  is finitely supported and monotonic with respect to subset
  inclusion.  Then~$f$ has a least fixpoint~$\lfp f$, and $$\lfp f =
  \bigcap\Set{ S \in\powfin{X} \setsep f(S) \subseteq S}.$$
\end{thm}
\begin{proof} (Due to Pitts~\cite{pitts15email})\\
  Since~$f$ is finitely supported and~$\bigcap$ is equivariant, also
  $\bigcap\Set{ S \in\powfin{X} \setsep f(S) \subseteq S}$ is finitely
  supported (with support contained in~$\n(f)$).  It then follows by a
  replay of the usual Tarski argument that $\bigcap\Set{ S
    \in\powfin{X} \setsep f(S) \subseteq S}$ is the least fixpoint
  of~$f$.
\end{proof}

Finally, we can show that the interpretation of a fixpoint
formula~$\mu X. A$ is the least fixpoint of the
function~$F_{A}^\varepsilon$.

\begin{prop} \label{prop:mu-fixpoint}
  For any formula~$\mu X. A$ and valuation function~$\varepsilon$,
  $\interpe{\mu X. A} = \lfp F_{A}^\varepsilon$.
\end{prop}

\begin{proof}
  Using Lemmas~\ref{lem:mu-fininite-supp} and~\ref{lem:mu-monotonic},
  the proposition is immediate from Theorem~\ref{thm:tarski}.
\end{proof}

\subsection{Encoding the least fixpoint operator}
\label{sec:encod-least-fixp}

The least fixpoint operator can be encoded in our logic of
Section~\ref{sec:nominaltransitionsystems}.  The basic idea here is
simple: we translate the fixpoint modality into a transfinite
disjunction that at each step~$\alpha$ unrolls the formula
$\alpha$~times.  This then semantically corresponds to a limit of an
increasing chain generated by a monotonic function, i.e., a least
fixpoint.

Recall that the cardinality of a set of conjuncts---and thus also of a
set of disjuncts---must be less than some fixed infinite
cardinal~$\kappa$ (see Definition~\ref{def:formulas}).  As before, we
require $\kappa > \aleph_0$ and $\kappa > \lvert \states\rvert$.

\begin{defi}
  We define the formula $\mathsf{unroll}_{\alpha}(\mu X.A)$ for all
  ordinals $\alpha < \kappa$ by transfinite induction.
  \[
  \begin{array}{rcl}
    \mathsf{unroll}_0(\mu X.A) &=& \bot \\
    \mathsf{unroll}_{\alpha+1}(\mu X.A) &=& A[\mathsf{unroll}_\alpha(\mu X.A)/X] \\
    \mathsf{unroll}_{\lambda}(\mu X.A) &=& \bigvee_{\alpha<\lambda} \mathsf{unroll}_\alpha(\mu X.A) \text{\ \ when $\lambda$ is a limit ordinal}
  \end{array}
  \]
\end{defi}

Since~$\mathsf{unroll}_\alpha$ is equivariant, the disjunction in the
limit case has finite support bounded by~$\n(A)$.  Note that
if~$A$ does not contain any fixpoint modalities then
$\mathsf{unroll}_\alpha(\mu X.A)$ also does not contain any fixpoint
modalities.

To show that for some ordinal~$\gamma < \kappa$, the interpretation of
$\mathsf{unroll}_{\gamma}(\mu X.A)$ indeed is the least fixpoint
of~$F^{\varepsilon}_{A}$, we use a nominal version of a chain fixpoint
theorem for sets by Kuratowski~(1922)~\cite{kuratowski22:transfinis}, augmented with a bound on the
depth of the unrolling.

\begin{thm} \label{thm:fixpoint}
  Suppose $X$ is a nominal set, and $f \colon\powfin X \to \powfin X$
  is finitely supported and monotonic with respect to subset
  inclusion.  Set $f^0 = \emptyset$, $f^{\alpha + 1} = f(f^\alpha)$,
  and $f^\lambda = \bigcup_{\alpha < \lambda} f^\alpha$ for limit
  ordinals~$\lambda$.  Then~$f$ has a least fixpoint~$\lfp f$, and
  if~$\nu$ is a cardinal with~$\nu > \lvert X \rvert$, there exists an
  ordinal~$\gamma < \nu$ such that
  \[
    \lfp f = f^\gamma.
  \]
\end{thm}

\begin{proof}
  First, each $f^\alpha$ is finitely supported, since $\n(f^\alpha)
  \subseteq \n(f)$ for every ordinal~$\alpha$ by transfinite
  induction.  Also, using monotonicity of~$f$, we have $f^\alpha
  \subseteq f^\beta$ for all~$\alpha\leq\beta$.

  We then show that $f$ has a fixpoint $f^\gamma$ for some ordinal
  $\gamma < \nu$, by contradiction.  Otherwise, for each $\gamma <
  \nu$ there is $x_{\gamma}\in f^{\gamma + 1}\setminus f^{\gamma}$.
  This yields an injective function $g \colon \nu \to X$ with
  $g(\gamma)=x_{\gamma}$, which is a contradiction since $\nu > \lvert
  X \rvert$.

  Let $y$ be any fixpoint of~$f$.  For every ordinal~$\alpha$,
  $f^\alpha\subseteq y$ by transfinite induction, so in particular
  $f^\gamma\subseteq y$.  Thus~$f^\gamma$ is the least fixpoint
  of~$f$.
\end{proof}

Let $\nu = \max \, \{\lvert\states\rvert^{+}, \aleph_0\}$ denote the
least infinite cardinal larger than~$\lvert\states\rvert$.  Note that
$\nu \leq \kappa$ by assumption.

\begin{lem} \label{lemma:fixpoint-F}
  For any formula~$\mu X.A$ and valuation function~$\varepsilon$,
  there exists an ordinal~$\gamma < \nu$ such that $\lfp
  F^{\varepsilon}_{A} = (F^{\varepsilon}_{A})^\alpha$ for all ordinals
  $\alpha \geq \gamma$.
\end{lem}

\begin{proof}
  Since $\nu >\lvert\states\rvert$, the lemma is an immediate
  consequence of Theorem~\ref{thm:fixpoint}, whose other assumptions
  follow from Lemmas~\ref{lem:mu-fininite-supp}
  and~\ref{lem:mu-monotonic}.
\end{proof}

From this lemma, we obtain an equivariant function~$\mathsf{conv}$
that maps each formula~$\mu X.A$ and each valuation
function~$\varepsilon$ to the least ordinal~$\mathsf{conv}(\mu X.A,
\varepsilon) < \nu$ such that $$\lfp F^{\varepsilon}_{A} =
(F^{\varepsilon}_{A})^{\mathsf{conv}(\mu X.A, \varepsilon)}.$$

When~$\mathcal{E}$ is a non-empty set of valuation functions, we write
$\mathsf{Conv}(\mu X.A, \mathcal{E})$ for
$\sup_{\varepsilon\in\mathcal{E}} \, \allowbreak \mathsf{conv}(\mu
X.A, \varepsilon)$.

\begin{lem} \label{lemma:sup-lt-kappa}
  Let~$\mathcal{E}$ be a non-empty set of valuation functions such
  that $\lvert\mathcal{E}\rvert < \nu$.  Then $\mathsf{Conv}(\mu X.A,
  \mathcal{E}) < \nu$.
\end{lem}

\begin{proof}
  Note that~$\aleph_0$ is regular, and every successor cardinal is
  regular.  Hence~$\nu$ is regular.  Therefore, the set $\{
  \mathsf{conv}(\mu X.A, \varepsilon) \mid \varepsilon\in\mathcal{E}
  \}$ (whose cardinality is less than~$\nu$) is not cofinal in~$\nu$.
\end{proof}

\begin{lem} \label{lem:interpe-substitution}
  For any formulas~$A$,~$B$ and valuation function~$\varepsilon$,
  if~$A$ does not contain any fixpoint operators, then
  $\interpe{A[B/X]} =
  \interpe[{\varepsilon[X\mapsto\interpe{B}]}]{A}$.
\end{lem}

\begin{proof}
  By structural induction on $A$.  The clause for~$\mu X.A'$ in
  Definition~\ref{def:interpretation} is not used.
\end{proof}

\begin{lem} \label{lem:unroll-fixpoint}
  For any formula~$A$ and valuation function~$\varepsilon$, if~$A$
  does not contain any fixpoint operators, then
  $\interpe{\mathsf{unroll}_\alpha(\mu X.A)} =
  (F_A^\varepsilon)^\alpha$ for all ordinals~$\alpha < \kappa$.
\end{lem}

\begin{proof}
  By transfinite induction on~$\alpha$.
  \begin{enumerate}
    \item Base case: $\interpe{\mathsf{unroll}_0(\mu X.A)} =
      \interpe{\bot} = \emptyset = (F_A^\varepsilon)^0$ by definition.

    \item Inductive step:
      \begin{align*}
        \interpe{\mathsf{unroll}_{\alpha+1}(\mu X.A)}
        &{}= \interpe{A[\mathsf{unroll}_\alpha(\mu X.A)/X]} \\
        &{}\stackrel{(1)}{=} \interp{A}_{\varepsilon[X\mapsto\interpe{\mathsf{unroll}_\alpha(\mu X.A)}]}  \\
        &{}\stackrel{(2)}{=} \interp{A}_{\varepsilon[X\mapsto (F_A^\varepsilon)^\alpha]}  \\
        &{}= F_A^\varepsilon((F_A^\varepsilon)^\alpha) \\
        &{}= (F_A^\varepsilon)^{\alpha+1}
      \end{align*}
      as required.  Above, equality~(1) follows from
      Lemma~\ref{lem:interpe-substitution} and equality~(2) follows
      from the induction hypothesis for~$\alpha$.

    \item Limit case: The limit case is straightforward.  We have
      \begin{align*}
        \interpe{\mathsf{unroll}_{\lambda}(\mu X.A)}
        &{}= \interpe{\bigvee_{\alpha<\lambda} \mathsf{unroll}_\alpha(\mu X.A)} \\
        &{}= \bigcup_{\alpha<\lambda} \interpe{\mathsf{unroll}_\alpha(\mu X.A)} \\
        &{}\stackrel{(1)}{=} \bigcup_{\alpha<\lambda} (F_A^\varepsilon)^\alpha \\
        &{}= (F_A^\varepsilon)^\lambda
      \end{align*}
      where equality~(1) follows from the induction hypothesis for all~$\alpha <
      \lambda$. \qedhere
  \end{enumerate}
\end{proof}

\begin{defi} \label{def:encoding}
  Given any non-empty set of valuation functions~$\mathcal{E}$ with
  $\lvert\mathcal{E}\rvert < \nu$, we define the
  formula~$\encodefix{A}_\mathcal{E}$ homomorphically on the structure
  of~$A$.  The encoding $\encodefix{\mu X.A}_\mathcal{E}$ of a
  fixpoint formula is its unrolling up to a sufficiently large
  ordinal.
  \[\begin{array}{rcll}
    \encodefix{\conjunc A_i}_\mathcal{E} &=& \conjunc \encodefix{A_i}_\mathcal{E} \\[.3em]
    \encodefix{\neg A}_\mathcal{E} &=& \neg\encodefix{A}_\mathcal{E} \\
    \encodefix{\varphi}_\mathcal{E} &=& \varphi \\[.3em]
    \encodefix{\may{\alpha} A}_\mathcal{E} &=& \may \alpha \encodefix{A}_\mathcal{E} \\[.3em]
    \encodefix{X}_\mathcal{E} &=& X \\[.3em]
    \encodefix{\mu X.A}_\mathcal{E} &=& \mathsf{unroll}_{\gamma}(\mu X.\encodefix{A}_{\mathcal{E}'}) & \text{where } \mathcal{E}' = \{\varepsilon[X\mapsto (F_A^\varepsilon)^\alpha] \mid \varepsilon\in\mathcal{E}, \alpha \leq \mathsf{conv}(\mu X.A, \varepsilon)\}\\
    & & & \text{and } \gamma = \mathsf{Conv}(\mu X.\encodefix{A}_{\mathcal{E}'}, \mathcal{E})
  \end{array}\]
\end{defi}

In the fixpoint case, since $\lvert\mathcal{E}\rvert < \nu$ and
$\mathsf{Conv}(\mu X.A, \mathcal{E}) < \nu$
(Lemma~\ref{lemma:sup-lt-kappa}), we also have
$\lvert\mathcal{E}'\rvert < \nu\cdot\nu = \nu$.  Since the encoding
function is equivariant, it preserves the finite support property for
conjunctions.  Clearly,~$\encodefix{A}_\mathcal{E}$ does not contain
any fixpoint operators.  Moreover, if~$A$ is closed,
then~$\encodefix{A}_\mathcal{E}$ does not contain any variables, and
is therefore a formula in the sense of Definition~\ref{def:formulas}.

\begin{thm} \label{thm:mu-admissable}
  Let~$\mathcal{E}$ be a non-empty set of valuation functions such
  that $\lvert\mathcal{E}\rvert < \nu$.  For any formula $A$ and
  valuation function $\varepsilon\in\mathcal{E}$,
  $\interpe{\encodefix{A}_\mathcal{E}} = \interpe{A}$.
\end{thm}

\begin{proof}
  By structural induction on~$A$, for arbitrary~$\mathcal{E}$
  and~$\varepsilon$.  The interesting case is
  \begin{description}
  \item[Case $\mu X.A'$] We need to show that $\interpe{\encodefix{\mu
      X.A'}_\mathcal{E}} = \interpe{\mu X.A'}$.  Let~$\mathcal{E}'$
    and~$\gamma$ be as in Definition~\ref{def:encoding}.  First, we
    compute the left-hand side to $\interpe{\encodefix{\mu
        X.A'}_\mathcal{E}} = \interpe{\mathsf{unroll}_\gamma(\mu
      X. \encodefix{A'}_{\mathcal{E}'})} =
    (F_{\encodefix{A'}_{\mathcal{E}'}}^\varepsilon)^\gamma = \lfp
    F_{\encodefix{A'}_{\mathcal{E}'}}^\varepsilon$ by
    Lemma~\ref{lem:unroll-fixpoint} and Lemma~\ref{lemma:fixpoint-F},
    where we use the fact that $\gamma \geq \mathsf{conv}(\mu
    X.\encodefix{A'}_{\mathcal{E}'},\varepsilon)$ since
    $\varepsilon\in\mathcal{E}$.
    For the right-hand side, we have $\interpe{\mu X.A'} = \lfp
    F_{A'}^\varepsilon$ by Proposition~\ref{prop:mu-fixpoint}.

    We now show
    $(F_{\encodefix{A'}_{\mathcal{E}'}}^\varepsilon)^\alpha =
    (F_{A'}^\varepsilon)^\alpha$ for all ordinals~$\alpha \leq
    \mathsf{conv}(\mu X.A', \varepsilon) + 1$ by transfinite induction
    on~$\alpha$.  The base case ($\alpha = 0$) and limit case ($\alpha
    = \lambda$) are straightforward.  For the inductive step, we have
      \begin{align*}
        (F_{\encodefix{A'}_{\mathcal{E}'}}^\varepsilon)^{\alpha+1}
        &{}= F_{\encodefix{A'}_{\mathcal{E}'}}^\varepsilon((F_{\encodefix{A'}_{\mathcal{E}'}}^\varepsilon)^\alpha) \\
        &{}\stackrel{(1)}{=} F_{\encodefix{A'}_{\mathcal{E}'}}^\varepsilon((F_{A'}^\varepsilon)^\alpha) \\
        &{}= \interp{\encodefix{A'}_{\mathcal{E}'}}_{\varepsilon[X\mapsto (F_{A'}^\varepsilon)^\alpha]}  \\
        &{}\stackrel{(2)}{=} \interp{A'}_{\varepsilon[X\mapsto (F_{A'}^\varepsilon)^\alpha]}\\
        &{}= F_{A'}^\varepsilon((F_{A'}^\varepsilon)^\alpha) \\
        &{}= (F_{A'}^\varepsilon)^{\alpha+1}
      \end{align*}
      as required.  Above, equality~(1) follows from the induction
      hypothesis for~$\alpha$, and equality~(2) follows from the outer
      induction hypothesis applied to~$\mathcal{E}'$
      and~$\varepsilon[X\mapsto(F_{A'}^\varepsilon)^\alpha]$.  Note
      that $\varepsilon[X\mapsto(F_{A'}^\varepsilon)^\alpha] \in
      \mathcal{E}'$ since $\alpha\leq\mathsf{conv}(\mu X.A',
      \varepsilon)$.

      It follows with Lemma~\ref{lemma:fixpoint-F} that
      $(F_{\encodefix{A'}_{\mathcal{E}'}}^\varepsilon)^{\mathsf{conv}(\mu
        X.A',\varepsilon)} = \lfp F_{A'}^\varepsilon =
      (F_{\encodefix{A'}_{\mathcal{E}'}}^\varepsilon)^{\mathsf{conv}(\mu
        X.A',\varepsilon)+1}$.  Hence $\lfp F_{A'}^\varepsilon$ is a
      fixpoint of $F_{\encodefix{A'}_{\mathcal{E}'}}^\varepsilon$, and
      thus $\lfp F_{\encodefix{A'}_{\mathcal{E}'}}^\varepsilon
      \subseteq \lfp F_{A'}^\varepsilon$.  Moreover,
      $(F_{\encodefix{A'}_{\mathcal{E}'}}^\varepsilon)^\alpha
      \subseteq \lfp F_{\encodefix{A'}_{\mathcal{E}'}}^\varepsilon$
      for any ordinal~$\alpha$ (as in the proof of
      Theorem~\ref{thm:fixpoint}), so in particular $\lfp
      F_{A'}^\varepsilon =
      (F_{\encodefix{A'}_{\mathcal{E}'}}^\varepsilon)^{\mathsf{conv}(\mu
        X.A',\varepsilon)} \subseteq \lfp
      F_{\encodefix{A'}_{\mathcal{E}'}}^\varepsilon$.  By combining
      both inclusions, $\lfp
      F_{\encodefix{A'}_{\mathcal{E}'}}^\varepsilon = \lfp
      F_{A'}^\varepsilon$. \qedhere
  \end{description}
\end{proof}

If~$A$ is closed, its semantics does not depend on~$\varepsilon$.  In
this case, we can pick an arbitrary valuation function to perform the
encoding: e.g., let $\varepsilon_\emptyset$ be the valuation that maps
every variable to~$\emptyset$.  We simply write~$\encodefix{A}$
for~$\encodefix{A}_{\{\varepsilon_\emptyset\}}$.

  Every closed formula containing fixpoint operators can be translated
into an equivalent formula without fixpoint operators.

\begin{cor}
  For any~$\varepsilon$,~$P$ and closed formula~$A$, we have
  $\valid{P}{\encodefix A}$ iff $P\in\interpe{A}$.
\end{cor}

\begin{proof}
  By Theorem~\ref{thm:mu-admissable} and
  Proposition~\ref{prop:valid-iff-interpe}.
\end{proof}

\section{Logics for variants of bisimilarity}
\label{sec:variants}

\subsection{Variants of bisimilarity}
\label{sec:bisimvariants}
There are variants of bisimilarity, differing in the effect the binding in a transition can have on the target state. A typical example is in the pi-calculus where the input action binds a name, and the target state must be considered for all possible instantiations of it. There is then a difference between the so-called {\em late} bisimilarity, where the target states must bisimulate before instantiating the input, and {\em early} bisimilarity, where it is enough to bisimulate after each instantiation. There are also corresponding congruences, obtained by closing bisimilarity under all substitutions of names for names. The original value-passing variant of CCS from 1989~\cite{MilnerCCS} uses early bisimilarity. The original bisimilarity for the pi-calculus (1992) is of the late kind~\cite{MPWpi}, where it also was noted that late equivalence is the corresponding congruence. Early bisimilarity and congruence in the pi-calculus were introduced in 1993~\cite{Milner:1993ys}, where HMLs adequate for a few different bisimilarity definitions are explored. Other ways to treat the name instantiations include Sangiorgi's {\em open bisimilarity} (1993)~\cite{open} and Parrow and Victor's {\em hyperbisimilarity} (1998)~\cite{fusion}. Hyperbisimilarity is the requirement that the bisimulation relation is closed under all name instantiations, and corresponds to a situation where the environment of a process may at any time instantiate any name. Open bisimilarity is between late and hyper: the environment may at any time instantiate any name except those that have previously been extruded or declared constant.

In our definition of nominal transition systems there are no particular input variables in the states or in the actions, and thus no a priori concept of replacing a name by something else. In order to cover all of the above variants of the pi-calculus and also of high-level extensions of it, we generalise name instantiation using a notion of effect functions.

\begin{defi}
 An {\em effect} is a finitely supported function from states to states. We let $\mathcal{F}$ stand for a nominal set of effects, and we let $\effectset$, ranged over by $F$, be the finitely supported subsets of~$\mathcal{F}$.
\end{defi}

For instance, in the monadic pi-calculus the effects would be the functions replacing one name by another. In a value-passing calculus the effects would be substitutions of values for variables. In the psi-calculi framework the effects would be sequences of parallel substitutions. Variants of bisimilarity then correspond to the use of various effects. For instance, if the action contains an input variable $x$, then the effects appropriate for late bisimilarity would be substitutions for $x$. Our only requirement is that the effects form a nominal set.

The definition of bisimilarity will now be through a $\effectset$-indexed family~$\{R_F\}$ of bisimulations. The index $F$ (``first'') simply says which effects must be taken into account before performing a transition. There is a function $L$ (``later''), which determines what effects should be considered after taking a transition. In its simplest form these effects depend only on the action of the transition, that is, $L$ has type  $\actions \rightarrow \effectset$. We require that~$L$ is equivariant.

\begin{defi}[Simple $L$-bisimulation and $\fbisim$]
\label{def:alternatebisim}
A {\em simple L-bisimulation}, for $L\colon\actions
\rightarrow \effectset$ equivariant, is a $\effectset$-indexed
family~$\{R_F\}$ of symmetric binary relations on states satisfying
the following:

For all $F \in \effectset$, $R_F (P, Q)$ implies
\begin{enumerate}
\item {\em Static implication}: For all $f \in F$, $\entails{f(P)}{\varphi}$ implies  $\entails{f(Q)}{\varphi}$.
\item {\em Simulation}: For all $f\in F$ and $\alpha$, $P'$ such that $\bn(\alpha)\freshin f(Q), F$ there exists $Q'$ such that
\[\mbox{if $f(P) \trans{\alpha} P'$ then $f(Q) \trans{\alpha} Q'$ and $R_{L(\alpha)}(P',Q')$}\]
\end{enumerate}

We write $P \fbisim Q$, called $F/L$-bisimilarity, to mean that there
exists a simple $L$-bisimulation $\{R_F\}_{F\in\effectset}$ such that
$R_F(P,Q)$.
\end{defi}

To exemplify $F/L$-bisimilarity we shall consider some of the popular bisimulation equivalences in the monadic pi-calculus; the ideas obviously generalize to more advanced settings. Thus the states are pi-calculus agents, and there are no state predicates. In the monadic pi-calculus there are input actions written $a(x)$, where $a$ and $x$ are names and the intention is that the binding input object $x$ shall be instantiated with another name received in a communication with a parallel process. There are also binding output actions
$\overline{a}(x)$ signifying the output of a local name $x$. We let $\bn(a(x)) = \bn(\overline{a}(x)) = \{x\}$, and for all other actions $\alpha$ we let $\bn(\alpha) = \emptyset$.

The relevant effects are the name substitutions, i.e., functions $\sigma$ from names to names that are identity almost everywhere. For any substitution, the set of names which are substituted to something else, or something else is substituted to, is finite. Clearly if a permutation only permutes names outside this set, then its action does not change the substitution. Thus every substitution has finite support, and we can let  $\mathcal{F}$ be the set of substitutions.
The effect of applying $\sigma$ to a state $P$ is notated $P\sigma$, in conformance with most of the literature on the pi-calculus.

Write $\stateid$ for the identity function on names, and let $\substx = \{\sigma \setsep \forall y\neq x.\ \sigma(y) = y\}$ be the set of substitutions for $x$. Note that $\n(\substx) = \{x\}$. We now get:

\begin{itemize}
\item {\bf Early bisimilarity} $\bisim_{\rm E}$ does not use binding input actions, instead it uses non-binding actions where the received object is already present. Early bisimilarity thus is precisely as defined in Definition~\ref{def:bisim}, which is the same as  $\setid \,/\, L_{\rm E}$-bisimilarity where $L_{\rm E}(\alpha)=\setid$ for all $\alpha$. No substitutive effect is needed; the substitution of output object for input object is included already in the semantics and thus already present in the corresponding nominal transition system.

\item {\bf Early equivalence} $\sim_{\rm E}$ is early bisimilarity closed under all possible substitutions, i.e., $P \sim_{\rm E} Q$ if for all $\sigma$ it holds
$P\sigma \bisim_{\rm E} Q\sigma$. Therefore $\sim_{\rm E}$ is ${\mathcal{F}}\,/\, L_{\rm E}$-bisimilarity where $L_{\rm E}$ is as above. Any substitution can be applied initially, and thereafter no substitution is needed. Early equivalence is the smallest congruence including early bisimilarity.

\item{\bf Late bisimilarity} $\bisim_{\rm L}$ has a binding input action and should consider all possible instantiations of the bound input object before the next transition. In other words,~$\bisim_{\rm L}$ is $\setid \,/\, L_{\rm L}$-bisimilarity where $L_{\rm L}(a(x)) = \substx$ and $L_{\rm L}(\alpha) = \setid$ for all other actions $\alpha$.

\item {\bf Late equivalence} $\sim_{\rm L}$ is late bisimilarity closed under all possible substitutions, i.e., $P \sim_{\rm L} Q$ if for all $\sigma$ it holds $P\sigma \bisim_{\rm L} Q\sigma$. Therefore $\sim_{\rm L}$ is $\mathcal{F} \,/\, L_{\rm L}$-bisimilarity where~$L_{\rm L}$ is as above. Late equivalence is the smallest congruence including late bisimilarity.

\item{\bf Hyperbisimilarity} $\sim_{\rm H}$ means that any name can be substituted at any time, thus it is ${\mathcal{F}}\,/\,L_{\rm H}$-bisimilarity where $L_{\rm H}(\alpha) = \mathcal{F}$ for all $\alpha$.
\end{itemize}

Open bisimilarity is more involved and requires a generalisation of the $L$-function to take additional parameters. We begin by quoting the definition from Sangiorgi 1993~\cite{open}.
\begin{defi}[Open bisimilarity, Sangiorgi]
\label{def:open-sangiorgi}
A {\em distinction} is a finite symmetric and irreflexive relation on names. A substitution $\sigma$ {\em respects} a distinction $D$ if $(a,b) \in D$ implies $\sigma(a) \neq \sigma(b)$. A distinction-indexed family of symmetric relations $\{S_D\}_D$ is an {\em open bisimulation} if for all $S_D$ and for each $\sigma$ which respects $D$, $(P,Q) \in S_D$ implies
\begin{enumerate}
\item If $P\sigma \trans{\overline{a}(b)}P'$ with $b$ fresh then $Q'$ exists s.t.\ $Q\sigma \trans{\overline{a}(b)}Q'$ and $(P',Q')\in S_{D'}$ where $D' = D\sigma \cup (\{b\} \times \mbox{fn}(P\sigma, Q\sigma))$ (with symmetric closure)
\item If $P\sigma \trans{\alpha} P'$ with $\alpha$ not a binding output and $\bn(\alpha)$ fresh then $Q'$ exists s.t.\ $Q\sigma \trans{\alpha} Q'$ and $(P',Q')\in S_{D\sigma}$
\end{enumerate}
Write $P \sim_{\rm O} Q$ to mean that $(P,Q) \in S_\emptyset$ for some open bisimulation $\{S_D\}_D$.
\end{defi}
In this definition, ``$b$ is fresh'' means that it ``is supposed to be different from any other name appearing in the objects of the statement, like processes or distinctions.'' The function~$\mbox{fn}$ extracts the free names of a process; in nominal terms this is the support.

Clearly, the distinctions here correspond to our effect sets since they determine which substitutions should be taken into account. A complication is then that the distinction $D'$ after the transition in clause~(1) depends not only on the action but also on $D$, $\sigma$, $P$ and $Q$. We therefore present an alternative definition of open bisimulation, where distinctions may be infinite but still finitely supported. Write $D+b$ for the distinction $D \cup (\{b\} \times (\nameset-\{b\}))$ (and its symmetric closure) and $D-b$ for the distinction $\{(x,y) \in D\setsep x,y \neq b\}$. Note that $\n(D+b) \subseteq \n(D)\cup \{b\}$ and $\n(D-b) \subseteq \n(D)\cup \{b\}$. It is easy to see that the function that maps a distinction to the set of substitutions respecting it is an equivariant injection.
Therefore we let  sets of substitutions be represented as distinctions, and say ``$\sigma \in D$'' instead of ``$\sigma$ respects $D$.''
\begin{defi}[Alternative definition of open bisimulation]
\label{def:altopen}
A distinction-indexed family of symmetric relations $\{S_D\}_D$ is an open bisimulation if for all $S_D$ and for each $\sigma \in D$, $(P,Q) \in S_D$ implies
\begin{enumerate}
\item If $P\sigma \trans{\overline{a}(b)}P'$ with $b\freshin Q\sigma,D,\sigma$ then $Q'$ exists s.t.\ $Q\sigma \trans{\overline{a}(b)}Q'$ and $(P',Q')\in S_{D\sigma+b}$
\item If $P\sigma \trans{\alpha} P'$ with $\bn(\alpha)=\emptyset$ then $Q'$ exists s.t.\ $Q\sigma \trans{\alpha} Q'$ and $(P',Q')\in S_{D\sigma}$
\item If $P\sigma \trans{a(b)}P'$ with $b\freshin Q\sigma,D,\sigma$ then $Q'$ exists s.t.\ $Q\sigma \trans{a(b)}Q'$ and $(P',Q')\in S_{D\sigma-b}$
\end{enumerate}
\end{defi}

Compared to Definition~\ref{def:open-sangiorgi}, clause~(1) now requires $b$ to be distinct from {\em all} names, not just the names in $P\sigma$ and $Q\sigma$. In the pi-calculus, it is known that if $P\sigma \trans{\overline{a}(b)} P'$ then $\n(P') \subseteq \n(P\sigma) \cup \{b\}$. Therefore, the difference between the clauses only concerns names outside $\n(P',Q')$, meaning that the additional substitutions allowed by clause~(1) of Definition~\ref{def:open-sangiorgi} will be injective when restricted to $\n(P',Q')$, and open bisimilarity (with any distinction) is closed under injective substitutions.

All names ``occur'' in  $D\sigma+b$ even though it has finite support. Thus, in a subsequent input $a(b)$ it will be impossible to choose $b$ fresh according to Sangiorgi. Instead we get the same effect with clause~(3): removing an input binder $b$ from the distinction is the same as choosing a new one that does not occur there. We also tighten and make explicit the necessary freshness condition on $b$.

To capture open bisimulation in our framework, we extend the simple $L$-functions to take more parameters as follows:

\begin{defi}[$L$-bisimulation]
\label{def:generalalternatebisim}
A {\em (general) L-bisimulation}, for $L\colon\actions \times \effectset \times \mathcal{F}
\rightarrow \effectset$ equivariant, is a $\effectset$-indexed
family~$\{R_F\}$ of symmetric binary relations on states satisfying
the following:

For all $F \in \effectset$, $R_F (P, Q)$ implies
\begin{enumerate}
\item {\em Static implication}: For all $f \in F$, $\entails{f(P)}{\varphi}$ implies  $\entails{f(Q)}{\varphi}$.
\item {\em Simulation}: For all $f\in F$ and $\alpha, P'$ such that $\bn(\alpha)\freshin  f(Q), F, f$ there exist $Q'$ such that
\[\mbox{if $f(P) \trans{\alpha} P'$ then $f(Q) \trans{\alpha} Q'$ and $R_{L(\alpha,F,f)}(P',Q')$}\]
\end{enumerate}
\end{defi}

The simple $L$-bisimulation of Definition~\ref{def:alternatebisim} is thus the special case when $L$ does not depend on $F$ or $f$. In the following when we write $L$-bisimulation we always refer to the general case of Definition~\ref{def:generalalternatebisim}.

To represent open bisimulation as an $L$-bisimulation we let $\mathcal{F}$ be the set of all substitutions and use distinctions to represent sets of substitutions (thus $\mathcal{F}$ is the empty distinction), and
define the function $L_{\rm O}$ by:
\[
  L_{\rm O}(\alpha,D,\sigma) = \left\{
     \begin{array}{ll}
         D\sigma +b  & \mbox{if}\; \alpha = \overline{a}(b) \\
         D\sigma     & \mbox{if}\; \bn(\alpha) = \emptyset \\
         D\sigma -b  & \mbox{if}\; \alpha = a(b)
     \end{array}
  \right.
\]
Strictly speaking this $L_{\rm O}$ is partial since there are sets of substitutions in $\effectset$ that cannot be represented by a distinction. These substitution sets will not matter and we can make~$L_{\rm O}$ total by assigning an arbitrary value in those cases. We now immediately get, using Definition~\ref{def:altopen}:
\begin{prop}
$P \sim_{\rm O} Q \quad$ iff $\quad P \stackrel{\mathcal{F}/L_{\rm O}}{\sim} Q$
\end{prop}
In conclusion we have a uniform framework to define most  variants of bisimulation co-inductively. The common claim that open bisimulation is the only co-inductively defined congruence in the pi-calculus is somewhat contradicted by our framework, which represents both late and early equivalence conveniently.

\subsection{Variants of the logic}

In view of the previous subsection, we only need to provide a modal logic adequate for (general) $F/L$-bisimilarity; it can then immediately be specialised to all of the above variants.  To this end we introduce a new kind of logical {\em effect consequence} operator $\effect{}{}$, which appears in front of state predicates and actions in formulas.  We define the formulas that can directly use effects from~$F$ and after actions use effects according to~$L$, ranged over by $A^{F/L}$, in the following way:
\begin{defi}\label{def:fl-formulas}
Given~$L$ as in Definition~\ref{def:generalalternatebisim}, for all $F \in \effectset$ define ${\mathcal{A}}^{F/L}$ as the set of formulas given by the mutually recursive definitions:
\[ A^{F/L} \quad ::= \quad
   \conjunc A^{F/L}_i \casesep
    \neg A^{F/L} \casesep
    \effect{f}{\varphi} \casesep
    \effect{f}{\may{\alpha}} A^{L(\alpha,F,f)/L}
\]
where we require~$f\in F$ and $\bn(\alpha)\freshin f,F$ and that the conjunction has bounded cardinality and finite support.
\end{defi}

Validity of a formula for a state~$P$ is defined as in Definition~\ref{def:validity}, where in the last clause we assume that~$\may{\alpha} A$ is a representative of its alpha-equivalence class such that $\bn(\alpha) \freshin P,F,f$.  Validity of formulas involving the effect consequence operator is defined as follows.
\begin{defi}
\label{def:consequence}
For each $f \in \mathcal{F}$,
\[\valid{P}{\effectfa} \quad\mbox{if}\quad \valid{f(P)}{A}\]
\end{defi}
Thus the formula $\effect{f}{A}$ means that~$A$ holds when the effect~$f$ is applied to the state.
The effect consequence operator is similar in spirit to the action modalities: both~$\effect{f}A$ and~$\may{\alpha}A$ assert that something (an effect or action) must be possible and that~$A$ holds afterwards.  Indeed, effects can be viewed as a special case of transitions (as formalised in Definition~\ref{def:lftransform} below).

\begin{lem}
\label{lemma:distsup2}
If $A \in {\mathcal{A}}^{F/L}$ is a distinguishing formula for $P$ and $Q$, then there exists a distinguishing formula $B\in {\mathcal{A}}^{F/L}$ for $P$ and $Q$  such that $\n(B) \subseteq \n(P,F)$.
\end{lem}
\begin{proof}
  The proof has been formalised in Isabelle. It is an easy generalisation of the proof of Lemma~\ref{lemma:distsup}, just replace $\mathcal{A}$ by $\mathcal{A}^{F/L}$ and $\n(P)$ by $\n(P,F)$ everywhere. We additionally need to prove that $B \in \mathcal{A}^{F/L}$. Since~$L$ is equivariant, $A \in {\mathcal{A}}^{F/L}$ implies $\pi \cdot A \in {\mathcal{A}}^{(\pi \cdot F)/L} = {\mathcal{A}}^{F/L}$ for all $\pi \in \Pi_{(P,F)}$, this establishes $B \in {\mathcal{A}}^{F/L}$.
\end{proof}

Let $P \flogeq Q$ mean that $P$ and $Q$ satisfy the same formulas in ${\mathcal{A}}^{F/L}$.

\begin{thm}
\label{thm:bisimloglf}
$P \fbisim Q \quad$ iff $\quad P \flogeq Q$
\end{thm}
\begin{proof}
The proof has been formalised in Isabelle. Direction $\Rightarrow$ is a generalisation of Theorem~\ref{thm:bisimlog}.
\begin{enumerate}
\item
Base case: $A = \effect{f}{\varphi}$ and $f \in F$. Then $\entails{f(P)}{\varphi}$. By static implication $\entails{f(Q)}{\varphi}$, which means $\valid{Q}{A}$.
\item
  Inductive step $\effect{f}{\may{\alpha} A}$ where $A \in {\mathcal{A}}^{L(\alpha,F,f)/L}$ with $f \in F$ and $\bn(\alpha) \freshin f,F$: Assume $\valid{P}{\effect{f}{\may{\alpha} A}}$. Then $\exists P'. \, f(P) \trans{\alpha} P'$ and $\valid{P'}{A}$. Without loss of generality we assume also $\bn(\alpha) \freshin f(Q)$, otherwise just find an alpha-variant of the transition where this holds. Then by simulation $\exists Q'. \, f(Q) \trans{\alpha} Q'$ and $P' \fabisim Q'$. By induction and $\valid{P'}{A}$ we get $\valid{Q'}{A}$, whence by definition $\valid{Q}{\effect{f}{\may{\alpha} A}}$.
\end{enumerate}

The direction $\Leftarrow$ is a generalisation of Theorem~\ref{thm:logbisim}: we prove that $\flogeq$ is an $F/L$-bisimulation.  The modified clauses are:
\begin{enumerate}
\item Static implication. Assume $f \in F$, then $\entails{f(P)}{\varphi}$ iff $\valid{P}{\effect{f}{\varphi}}$ iff $\valid{Q}{\effect{f}{\varphi}}$ iff $\entails{f(Q)}{\varphi}$.
\item Simulation. The proof is by contradiction. Assume that $\flogeq$ does not satisfy the simulation requirement. Then there exist $f \in F$, $P$, $Q$, $P'$, $\alpha$ with $\bn(\alpha) \freshin f(Q), F, f$ such that $P \flogeq Q$ and
$f(P) \trans{\alpha}P'$ and, letting ${\mathcal{Q}} = \{Q' \setsep f(Q) \trans{\alpha} Q'\}$, for all $Q' \in {\mathcal{Q}}$ it holds that not $P' \falogeq Q'$. Thus, for all $Q' \in {\mathcal{Q}}$ there exists a distinguishing formula in ${\mathcal{A}}^{L(\alpha,F,f)/L}$ for $P'$ and $Q'$. The formula may depend on~$Q'$, and by Lemma~\ref{lemma:distsup2} we can find such a distinguishing formula $B_{Q'} \in {\mathcal{A}}^{L(\alpha,F,f)/L}$ for~$P'$ and~$Q'$ with $\n(B_{Q'}) \subseteq\n(P',L(\alpha,F,f))$.  Therefore the formula
 \[B = \bigwedge_{Q' \in {\mathcal{Q}}} B_{Q'}\]
 is well formed in ${\mathcal{A}}^{L(\alpha,F,f)/L} $ with support included in $\n(P',L(\alpha,F,f))$. We thus get that
 $\valid{P}{\effect{f}{\may{\alpha} B}}$ but not $\valid{Q}{\effect{f}{\may{\alpha} B}}$, contradicting $P\flogeq Q$. \qedhere
 \end{enumerate}
\end{proof}

\subsection{Effects as transitions}
It is possible to view effects as transitions, and thus not use the effect consequence operator explicitly.
The idea is to let the effect functions be part of the transition relation, thus $f(P)=P'$ becomes $P \trans{f} P'$. We here make this idea fully formal. Given a transition system {\bf T} we construct a new transition system $L(\mathbf{T})$ without effects. The effects in {\bf T} are included among the actions in $L(\mathbf{T})$. The states in $L(\mathbf{T})$ are of two kinds. One is $\ef(F,P)$ where $P$ is a state in {\bf T} and $F \in \effectset$, corresponding to the state $P$ where $F$ contains the effects that should be considered before taking an action or checking a state predicate. The transitions from $\ef(F,P)$ are effects in $F$. The other kind is $\ac(f,F,P)$, corresponding to a state $P$ where $f \in F$ has been applied; the transitions from $\ac(f,F,P)$ are the actions from $P$.

To define this formally, and to talk about different nominal transition systems in the same definition, we write $\states_{\mathbf{T}}$ for the states of the nominal transition system {\bf T} and similarly for  actions, predicates, etc.
\begin{defi}
\label{def:lftransform}
Let {\bf T} be a nominal transition system with a set of effects $\mathcal{F}$. Let $L$ be as in Definition~\ref{def:generalalternatebisim}. The {\em $L$-transform} $L(\mathbf{T})$ of a nominal transition system~{\bf T} is a nominal transition system where:
\begin{itemize}
\item
$
    \begin{array}[t]{@{}ll}
\states_{L(\mathbf{T})} = & \{\ac(f,F,P) \setsep
     f \in \mathcal{F}, F\in\effectset, P \in \states_{\mathbf{T}} \}
     \\
     & \cup\;
     \{  \ef(F,P) \setsep F\in\effectset, P \in \states_{\mathbf{T}} \}
     \end{array}
$

\item $\predicates_{L(\mathbf{T})} = \predicates_{\mathbf{T}}$

\item $\ac(f,F,P) \vdash_{L(\mathbf{T})} \varphi$ if $P \vdash_\mathbf{T} \varphi$, and $\ef(F,P) \vdash_{L(\mathbf{T})} \varphi$ never holds.

\item $\actions_{L(\mathbf{T})} = \actions_\mathbf{T} \uplus \mathcal{F}$.

\item $\bn_{L(\mathbf{T})}(\alpha) = \bn_\mathbf{T}(\alpha)$ for $\alpha \in \actions_\mathbf{T}$; $\bn_{L(\mathbf{T})}(f) = \emptyset$ for $f \in \mathcal{F}$.

\item The transitions in $L(\mathbf{T})$ are of two kinds. If in {\bf T} we have $P \trans{\alpha} P'$ with $\bn(\alpha) \freshin f,F$, then in $L(\mathbf{T})$ there is a transition $\ac(f,F,P) \trans{\alpha} \ef(L(\alpha,F,f),P')$.  Additionally, for each $f \in F$ it holds $\ef(F,P) \trans{f} \ac(f, F, f(P))$.
\end{itemize}
\end{defi}

The intuition is that states of kind~\ac{} can perform ordinary actions, and states of kind~\ef{} can commit effects.  The analogous transform of modal formulas in {\bf T} to formulas in $L(\mathbf{T})$ simply replaces effects by actions: $L(\effect{f}A) = \may{f}L(A)$, and~$L$ is homomorphic on all other formula constructors.

The $L$-transform preserves satisfaction of formulas in the following sense:

\begin{thm}
  \label{thm:lfformulas}
  Assume $A \in {\mathcal{A}}^{F/L}$.  Then
  $\valid{P}{A}$\ \ iff\ \ $\valid{\ef(F,P)}{L(A)}$.
\end{thm}

\begin{proof}
  The proof has been formalised in Isabelle.  It is by induction over
  formulas in~$\mathcal{A}^{F/L}$, for arbitrary~$P$.  The cases
  conjunction and negation are immediate by induction.
  \begin{description}
  \item[Case $A=\effect{f}\varphi$] We then know~$f \in F$, and have
    $\valid{P}{\effect{f}\varphi}$ iff $f(P)\vdash_\mathbf{T}\varphi$
    iff, by construction of $\vdash_{L(\mathbf{T})}$,
    $\valid{\ac(f,F,f(P))}{\varphi}$.  Since $\ef(F,P) \trans{f}
    \ac(f, F, f(P))$, it follows that
    $\valid{\ef(F,P)}{\may{f}\varphi} = L(\effect{f}\varphi)$.
    Conversely, from $\valid{\ef(F,P)}{\may{f}\varphi}$ we obtain~$P'$
    with $\ef(F,P) \trans{f} P'$ and $\valid{P'}{\varphi}$.  Since
    $\bn_{L(\mathbf{T})}(f) = \emptyset$, there are no other
    alpha-variants of this transition.  It follows that $P' =
    \ac(f,F,f(P))$.

  \item[Case $A = \effect{f}\may{\alpha}A'$] We then know $A' \in
    \mathcal{A}^{L(\alpha,F,f)/L}$ and $f \in F$ and
    $\bn(\alpha)\freshin f,F$.  Without loss of generality we assume
    also $\bn(\alpha) \freshin P$, otherwise just find an
    alpha-variant where this holds.

    Assume $\valid{P}{\effect{f}\may{\alpha}A'}$.  Then $\exists P'.\,
    f(P) \trans{\alpha} P'$ and $\valid{P'}{A'}$.  We thus have
    $\valid{\ef(L(\alpha,F,f),P')}{L(A')}$ by induction.  With
    $\ac(f,F,f(P)) \trans{\alpha} \ef(L(\alpha,F,f),P')$ we have
    $\valid{\ac(f,F,f(P))}{\may{\alpha}L(A')}$.  Since $\ef(F,P)
    \trans{f} \ac(f,F,f(P))$, it follows that
    $\valid{\ef(F,P)}{\may{f}\may{\alpha}L(A')} =
    L(\effect{f}\may{\alpha}A')$.

    Conversely, from $\valid{\ef(F,P)}{\may{f}\may{\alpha}L(A')}$ we
    obtain~$P'$ with $\ef(F,P) \trans{f} P'$ and
    $\valid{P'}{\may{\alpha}L(A')}$.  Since $\bn_{L(\mathbf{T})}(f) =
    \emptyset$, there are no other alpha-variants of this transition.
    It follows that $P' = \ac(f,F,f(P))$.  Note that $\bn(\alpha)
    \freshin f,F,f(P)$.  By construction of~$L(\mathbf{T})$ we
    obtain~$P''$ with $\ac(f,F,f(P)) \trans{\alpha}
    \ef(L(\alpha,F,f),P'')$ and $f(P)\trans{\alpha}P''$ and
    $\valid{\ef(L(\alpha,F,f),P'')}{L(A')}$.  Thus $\valid{P''}{A'}$
    by induction.  Hence $\valid{f(P)}{\may{\alpha}A'}$.  Hence
    $\valid{P}{\effect{f}\may{\alpha}A'}$.\qedhere
  \end{description}
\end{proof}

As an immediate corollary we get:
\begin{thm}
  \label{thm:lftransform}
  Let the $L$-transform be as above. Then $\ef(F,P) \bisim \ef(F,Q)
  \Longrightarrow P \fbisim Q$.
\end{thm}

\begin{proof}
  The proof has been formalised in Isabelle.  We first apply
  Theorem~\ref{thm:bisimlog} to get $\ef(F,P) \logeq \ef(F,Q)$, then
  Theorem~\ref{thm:lfformulas} to get $\ef(F,P) \flogeq \ef(F,Q)$, and
  finally Theorem~\ref{thm:bisimloglf} to get $P \fbisim Q$.
\end{proof}

The converse cannot be proved in the same way since~$L$ is not a
surjection on formulas. As an example, the formula $\may{f}\top$ is not in the range of $L$, since any effect in any formula in ${\mathcal{A}}^{F/L}$ must be followed by a state predicate or an action.
Nevertheless the converse holds:

\begin{thm}
  Let the $L$-transform be as above. Then $P \fbisim Q \Longrightarrow
  \ef(F,P) \bisim \ef(F,Q)$.
\end{thm}

\begin{proof}
  The proof has been formalised in Isabelle.  Define a binary
  relation~$R$ on~$\states_{L(\mathbf{T})}$ by including $(\ef(F,P),
  \ef(F,Q)) \in R$ and $(\ac(f,F,f(P)), \ac(f,F,f(Q))) \in R$ for all
  $f$, $F$, $P$, $Q$ with $f \in F$ and $P \fbisim Q$.  We prove
  that~$R$ is a bisimulation.

  Symmetry is immediate since~$\mathop{\fbisim}$ is symmetric.  Now
  assume $R(S,T)$.
  \begin{enumerate}
  \item{Static implication.}  Assume $S \vdash_{L(\mathbf{T})}
    \varphi$.  Then $S = \ac(f,F,f(P))$ for some~$F$, $f \in F$
    and~$P$ with $f(P) \vdash_{\mathbf{T}} \varphi$, and $T =
    \ac(f,F,f(Q))$ with $P \fbisim Q$.  Thus $f(Q) \vdash_{\mathbf{T}}
    \varphi$, hence $T \vdash_{L(\mathbf{T})} \varphi$.
  \item{Simulation.} Assume $S \trans{\alpha} S'$.  By construction
    of~$L(\mathbf{T})$ there are two cases:
    \begin{itemize}
    \item $S = \ef(F,P) \trans{f} \ac(f,F,f(P))=S'$ and $f \in F$.
      Then $T = \ef(F,Q)$ with $P \fbisim Q$.  We get $T \trans{f}
      \ac(f,F,f(Q)) =: T'$ and $R(S',T')$.  Note here that
      $\bn_{L(\mathbf{T})}(f) = \emptyset$.
    \item $S = \ac(f,F,f(P)) \trans{\alpha} \ef(L(\alpha, F,f),P') =
      S'$ and $f(P) \trans{\alpha} P'$ and $f \in F$.  Then $T =
      \ac(f,F,f(Q))$ with $P \fbisim Q$.  We may assume $\bn(\alpha)
      \freshin T$, hence also $\bn(\alpha) \freshin f(Q),F,f$.  Then
      by simulation $\exists Q' \,.\, f(Q) \trans{\alpha} Q'$ and $P'
      \stackrel{L(\alpha,F,f)/L}{\sim} Q'$.  Thus
      we have
        $T \trans{\alpha}
      \ef(L(\alpha, F,f),Q') =: T'$, and $R(S',T')$ as required. \qedhere
    \end{itemize}
  \end{enumerate}
\end{proof}

A consequence is that for the variants of bisimilarity considered in
Section~\ref{sec:bisimvariants} no new machinery is really needed;
they can all be obtained by extending the transition relation. For the
examples on the monadic pi-calculus, there would be substitution
transitions, which take a state to another state where the
substitution has been performed. Here the $L$ function determines
exactly which substitutions are applicable at which states, and by
varying it we can obtain e.g.\ late, open and hyperbisimulation.

\section{Unobservable actions}
\label{section:weak}
The logics and bisimulations considered so far are all of the strong variety, in the sense that all transitions are regarded as equally significant. In many models of concurrent computation there is a special action which is {\em unobservable} in the sense that in a bisimulation, and also in the definition of the action modalities, the presence of extra such transitions does not matter. This leads to notions of {\em weak} bisimulation and accompanying weak modal logics. For example, a process that has no transitions is weakly bisimilar to any process that has only unobservable transitions, and these satisfy the same weak modal logic formulas. We shall here introduce these ideas into the nominal transition systems. One main source of complication over similar treatments in process algebras turns out to be the presence of state predicates.

To cater for unobservable transitions assume a special action $\tau$ with empty support. The following definitions are standard:
\begin{defi} \label{def:weaktransition}\
\begin{enumerate}
\item $P \Trans{} P'$ is defined by induction to mean $P = P'$ or $P \trans{\tau} \circ \Trans{} P'$.
\item $P \Trans{\alpha} P'$ means $P \Trans{} \circ \trans{\alpha} \circ \Trans{} P'$.
\item $P \HTrans{\alpha} P'$ means $P \Trans{} P'$ if $\alpha=\tau$ and $P \Trans{\alpha} P'$ otherwise.
\end{enumerate}
\end{defi}
Intuitively $P \HTrans{\alpha} P'$ means that $P$ can evolve to $P'$ through transitions with the only observable content $\alpha$. We call this a weak action $\alpha$ and it will be the basis for the semantics in this section.

\subsection{Weak bisimilarity}
\label{sec:weakbisim}
The standard way to define weak bisimilarity is to weaken $Q \trans{\alpha} Q'$ to $Q \HTrans{\alpha} Q'$ in the simulation requirement. This results in the weak simulation criterion:
\begin{defi}
\label{def:weaksim}
A binary relation $R$ on $\states$ is a {\em weak simulation} if $R(P,Q)$ implies that for all $\alpha, P'$ with $\bn(\alpha)\freshin Q$ there exists $Q'$ such that
\[\mbox{if $P \trans{\alpha} P'$ then $Q \HTrans{\alpha} Q'$ and $R(P',Q')$}\]
\end{defi}

However, just replacing the simulation requirement with weak simulation in Definition~\ref{def:bisim} will not suffice. The reason is that through the static implication criterion in Definition~\ref{def:bisim}, an observer can still observe the state predicates directly, and thus distinguish between a state that satisfies $\varphi$ and a state that does not but can silently evolve to another state that satisfies $\varphi$:

\medskip
\begin{exa}~
  \label{ex:1}

  \begin{center}
\includegraphics{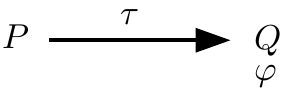}
  \end{center}

\noindent
Certainly $\{(P,Q), (Q,Q)\}$ is a weak simulation according to Definition~\ref{def:weaksim}. But $P \not\vdash \varphi$ and $Q \vdash \varphi$, thus they are in no static implication. We argue that if $\varphi$ is the {\em only} state predicate (in particular, there is no predicate $\neg\varphi$), then the only test that an observer can apply is ``if $\varphi$ then \ldots,'' and here $P$ and $Q$ will behave the same; $P$ can pass the test after an unobservable delay.  Thus $P$ and $Q$ should be deemed weakly bisimilar, and static implication as in Definition~\ref{def:bisim} is not appropriate.
\end{exa}

Therefore we need a weak counterpart of static implication where $\tau$ transitions are admitted before checking predicates, that is, if $P \vdash \varphi$ then $Q \Trans{} Q' \vdash \varphi$.  In other words, $Q$ can unobservably evolve to a state that satisfies~$\varphi$.  However, this is not quite enough by itself. Consider the following example where $P\vdash\varphi_0$, $P\vdash\varphi_1$, $R\vdash\varphi_1$ and $Q\vdash\varphi_0$, with transitions $P\trans{\tau}R$ and $Q\trans{\tau}R$:

\medskip
\begin{exa}~
  \label{ex:2}

  \begin{center}
\includegraphics{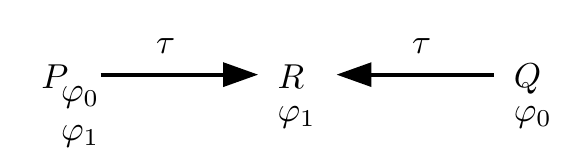}
  \end{center}

\noindent
Here we do not want to regard $P$ and $Q$ as weakly bisimilar. They  do have the same transitions and can satisfy the same predicates, possibly after a $\tau$ transition.  But an observer of $P$ can first determine that $\varphi_1$ holds, and then determine that~$\varphi_0$ holds.  This is not possible for $Q$: an observer who concludes~$\varphi_1$ must already have evolved to $R$.
\end{exa}

Similarly, consider the following example where the only difference between $P$ and $Q$ is that $\entails{P}{\varphi}$ but not $\entails{Q}{\varphi}$:

\medskip
\begin{exa}~
  \label{ex:3}

  \begin{center}
\includegraphics{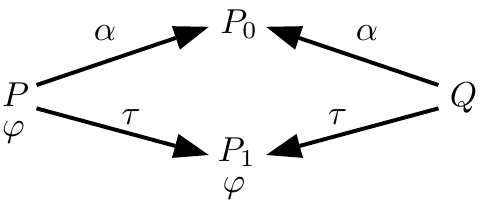}
  \end{center}

\noindent
Again we do not want to regard $P$ and $Q$ as weakly bisimilar.  Intuitively, an observer of~$Q$ who determines that $\varphi$ holds must already be at $P_1$  and thus have preempted the possibility to do~$\alpha$, whereas for~$P$, the predicate~$\varphi$ holds while retaining the possibility to do~$\alpha$.  For instance, $P$ in parallel with a process of kind ``if $\varphi$ then $\gamma$'' can perform~$\gamma$ followed by~$\alpha$, but~$Q$ in parallel with the same cannot do that sequence.
\end{exa}

In conclusion, the weak counterpart of static implication should allow the simulating state to proceed through unobservable actions to a state that {\em both} satisfies the same  predicate {\em and} continues to bisimulate. This leads to the following:
\begin{defi}
  \label{def:weakstaticimplication}
  A binary relation $R$ on states is a {\em weak static implication} if $R(P,Q)$ implies that for all $\varphi$ there exists $Q'$ such that
  \[\mbox{if $P \vdash \varphi $ then $Q \Trans{} Q'$ and $Q' \vdash \varphi$ and $R(P,Q')$ }\]
\end{defi}

\begin{defi}
  \label{def:weakbisim}
  A {\em weak bisimulation} is a symmetric binary relation on states satisfying both weak simulation and weak static implication. We write $P \wbisim Q$ to mean that there exists a weak bisimulation $R$ such that $R(P,Q)$.
\end{defi}

In Example~\ref{ex:1}, $\{(P, Q), (Q,P), (Q,Q)\}$ is a weak bisimulation.  In Examples~\ref{ex:2} and~\ref{ex:3}, $P$ and $Q$ are not weakly bisimilar.

It is interesting to compare Definition~\ref{def:weakbisim} with weak bisimilarities defined for psi-calculi~\cite{lics10}.  A psi-calculus contains a construct of kind ``if $\varphi$ then~\ldots''\ to test if a state predicate is true.  These constructs may be nested; for instance, ``if $\varphi_0$ then if $\varphi_1$ then~\ldots''\ effectively tests if both $\varphi_o$ and $\varphi_1$ are true simultaneously.  If state predicates are closed under conjunction, Definition~\ref{def:weakbisim} coincides with the definition of simple weak bisimulation in~\cite{lics10}.  In general, however, Definition~\ref{def:weakbisim} is less discriminating than in~\cite{lics10}.  Consider $P_0 \trans{\tau} P_1  \trans{\tau} P_0$ where for $i=0,1$: $P_i \vdash \varphi_i$.  Compare it to~$Q$ with no transitions where both $Q \vdash \varphi_0$ and $Q \vdash \varphi_1$:

\medskip
\begin{exa}~
  \label{ex:4}

  \begin{center}
\includegraphics{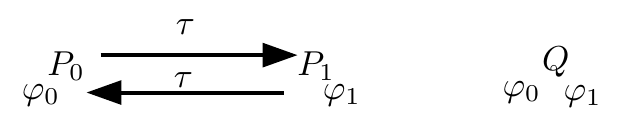}
  \end{center}

\noindent
Here all of $P_0$, $P_1$ and $Q$ are weakly bisimilar, unless the predicates are closed under conjunction, in which case the predicate $\varphi_0\wedge\varphi_1$ distinguishes between them.  In psi-calculi $Q$ would not be simply weakly bisimilar to $P_0$ or $P_1$ for the same reason. Thus the approach in the present paper is more general. In many cases it is reasonable to expect predicates to be closed under finite conjunctions, but there are circumstances when they are not, for example when the predicates can be checked only one at a time. 
  Consider for example a system where $\varphi_i$ means checking if the variable $i$ has value zero. If the system does not admit checking several variables simultaneously, then the predicates are not closed under conjunction, and there is no way to tell the systems in Example~\ref{ex:4} apart.
\end{exa}

We proceed to establish some expected properties of weak bisimilarity.

\begin{lem}
  \label{lem:weaktrans}
  If $P \wbisim Q$ and  $P \HTrans{\alpha} P'$ with $\bn(\alpha) \freshin Q$ then for some $Q'$ it holds $P' \wbisim Q'$ and  $Q \HTrans{\alpha} Q'$.
\end{lem}
\begin{proof}
  The proof has been formalised in Isabelle.  We first prove the lemma
  for the case $\alpha = \tau$ by induction on the derivation of $P
  \Trans{} P'$.  The base case is immediate.  In the inductive step,
  we have $P \trans{\tau} P'' \Trans{} P'$ for some~$P''$.  By weak
  simulation we obtain~$Q''$ with $Q \Trans{} Q''$ and $P'' \wbisim
  Q''$.  It then follows from the induction hypothesis that there
  exists~$Q'$ with $Q'' \Trans{} Q'$ and $P' \wbisim Q'$.

  In case $\alpha \neq \tau$, we obtain $P_1$, $P_2$ with $P \Trans{}
  P_1$ and $P_1 \trans{\alpha} P_2$ and $P_2 \Trans{} P'$.  By the
  previous case there exists~$Q_1$ with $Q \Trans{} Q_1$ and $P_1
  \wbisim Q_1$.  Without loss of generality we assume $\bn(\alpha)
  \freshin Q_1$, otherwise just find an alpha-variant of the
  transition where this holds.  By weak simulation we obtain~$Q_2$
  with $Q_1 \HTrans{\alpha} Q_2$ and $P_2 \wbisim Q_2$.  It then
  follows from the previous case that there exists~$Q'$ with $Q_2
  \Trans{} Q'$ and $P' \wbisim Q'$.
\end{proof}

\begin{lem}
  \label{lem:weakequiv}
  $\wbisim$ is an equivariant equivalence relation.
\end{lem}
\begin{proof}
  The proof has been formalised in Isabelle.  Equivariance is a simple
  calculation, based on the observation that if~$R$ is a weak
  bisimulation, then $\pi \cdot R$ is a weak bisimulation.  To prove
  reflexivity of~$\wbisim$, we note that equality is a weak
  bisimulation.  Symmetry is immediate from Definition~\ref{def:weakbisim}.
  To prove transitivity, we show that the composition of~$\wbisim$
  with itself is a bisimulation.  Assume $P \wbisim S \wbisim Q$.  The
  weak simulation requirement is proved by considering a
  transition~$P\trans{\alpha} P'$ where~$\bn(\alpha)$ is fresh
  for~$Q$.  Without loss of generality we also assume $\bn(\alpha)
  \freshin S$, otherwise just find an alpha-variant of the transition
  where this holds.  Since~$\wbisim$ satisfies weak simulation, we
  obtain~$S'$ with $S \HTrans{\alpha} S'$ and $P' \wbisim S'$.  Since
  $S \wbisim Q$, Lemma~\ref{lem:weaktrans} then implies that there
  exists~$Q'$ with $Q \HTrans{\alpha} Q'$ and $S' \wbisim Q'$.  To
  prove weak static implication, assume $P\vdash\varphi$.  We use $P
  \wbisim S$ and the fact that~$\wbisim$ satisfies weak static
  implication to obtain some~$S'$ with $S \Trans{} S'$ and $P \wbisim
  S'$ and $S'\vdash\varphi$.  Since $S \wbisim Q$,
  Lemma~\ref{lem:weaktrans} then implies that there exists~$Q'$ with
  $Q \Trans{} Q'$ and $S' \wbisim Q'$.  By weak static implication
  of~$\wbisim$ we obtain~$Q''$ with $Q' \Trans{} Q''$ and $S' \wbisim
  Q''$ and $Q''\vdash\varphi$.  So we have $Q \Trans{} Q''$ and $P
  \wbisim\!\circ\!\wbisim Q''$, as required.
\end{proof}

\subsection{Weak logic}
\label{subsec:weak-logic}

We here define a Hennessy-Milner logic adequate for weak bisimilarity. Since weak bisimilarity identifies more states than strong bisimilarity, the logic needs to be correspondingly less expressive: it must not contain formulas that distinguish between weakly bisimilar states.
Our approach is to keep the definition of formulas (Definition~\ref{def:formulas}) and identify an adequate sublogic.

One main idea is to restrict the action modalities $\may{\alpha}$ to occur only in accordance with the requirement of a weak bisimulation, thus checking for $\HTrans{\alpha}$ rather than for $\trans{\alpha}$. We therefore define the derived {\em weak action} modal operator $\wmay{\alpha}$ in the following way, where $\may{\tau}^i A$ is defined to mean $A$ if $i=0$ and $\may{\tau} \may{\tau}^{i-1} A$ otherwise.

\begin{defi}[Weak action modality]
  \[ \wmay{\tau}A = \bigvee_{i \in \omega} \may{\tau}^i A \qquad \qquad \qquad \wmay{\alpha}A = \wmay{\tau}\may{\alpha}\wmay{\tau}A \quad \mbox{for $\alpha \neq \tau$}\]
\end{defi}

Note that in $\wmay{\alpha}A$ the names in $\bn(\alpha)$ are abstracted and bind into $\alpha$ and $A$.
It is straightforward to show (and formalize in Isabelle) that $\wmay{\alpha}A$ corresponds to the weak transitions used in the definition of weak bisimilarity:

\begin{prop}
  Assume $\bn(\alpha)\freshin P$.  Then
  \[\valid{P}{\wmay{\alpha} A}  \quad \mbox{iff} \quad \exists P'.\, P \HTrans{\alpha} P' \mbox{ and } \valid{P'}{A}\]
\end{prop}

In particular, for $\alpha = \tau$, we have that $\wmay{\tau}A$ holds iff $A$ holds after zero or more $\tau$ transitions.

Thus a first step towards a weak sublogic is to replace $\may{\alpha}$ by $\wmay{\alpha}$ in Definition~\ref{def:formulas}.  By itself this is not enough; that sublogic is still too expressive.  For instance, the formula $\varphi$ asserts that $\varphi$ holds in a state; this holds for $Q$ but not for $P$ in Example~\ref{ex:1}, even though they are weakly bisimilar.

To disallow $\varphi$ as a weak formula we require that state predicates only occur guarded by a weak action $\wmay\tau$. This solves part of the problem. In Example~\ref{ex:1} we can no longer use $\varphi$ as a formula, and the formula $\wmay{\tau}\varphi$ holds of both $P$ and $Q$. Still, in Example~\ref{ex:1} there would be the formula $\wmay{\tau}\neg\varphi$ which holds for $P$ but not for $Q$, and in Example~\ref{ex:4} the formula $\wmay{\tau}(\varphi_0 \wedge \varphi_1)$ holds for $Q$ but not for $P_0$.  Clearly a logic adequate for weak bisimulation cannot have such formulas.  The more draconian restriction that state predicates occur {\em immediately} under $\wmay{\tau}$ would indeed disallow both $\wmay{\tau}\neg\varphi$ and $\wmay{\tau}(\varphi_0 \wedge \varphi_1)$ but would also disallow any formula distinguishing between $P$ and $Q$ in  Examples~\ref{ex:2} and~\ref{ex:3}.

A solution is to allow state predicates under $\wmay\tau$, and never directly under negation or in conjunction with another state predicate. The logic is:

\begin{defi}[Weak formulas]
The set of {\em weak formulas} is the sublogic of Definition~\ref{def:formulas} given by
\label{def:weaklogic}
\[
A \quad ::= \quad \conjunc A_i \casesep \neg A \casesep \wmay{\alpha}A  \casesep \wmay{\tau}(A\wedge\varphi)\]
\end{defi}

Note that since $P \HTrans\alpha \circ \Rightarrow P'$ holds iff $P \HTrans\alpha P'$ we have that $\wmay\alpha \wmay\tau A$ is logically equivalent to $\wmay\alpha A$.  We thus abbreviate $\wmay\alpha \wmay\tau (A \wedge \varphi)$ to $\wmay\alpha (A \wedge \varphi)$.  We also abbreviate $\wmay{\alpha}(\top \wedge \varphi)$ to $\wmay{\alpha}\varphi$.

Compared to Definition~\ref{def:formulas}, the state predicates can now only occur in formulas of the form $\wmay{\tau}(A\wedge\varphi)$, i.e., under a weak action, and not under negation or conjunction with another predicate.  For instance, in Example~\ref{ex:1} above, neither $\varphi$ nor $\wmay{\tau}\neg\varphi$ are weak formulas, and in fact there is no weak formula to distinguish between $P$ and $Q$.  Similarly, in Example~\ref{ex:4} $\wmay{\tau}(\varphi_0 \wedge \varphi_1)$ is not a weak formula, and no weak formula distinguishes between~$Q$ and $P_i$.

To argue that the logic still is expressive enough to provide distinguishing formulas for states that are not weakly bisimilar, consider Example~\ref{ex:2} and the formula $\wmay{\tau}((\wmay{\tau}\varphi_0) \wedge\varphi_1)$ which holds for $P$ but not for $Q$.  Similarly, in Example~\ref{ex:3} $\wmay{\tau}((\wmay{\alpha}\top) \wedge\varphi)$ holds for $P$ but not for $Q$.

\begin{defi}
  \label{def:weakequiv}
  Two states $P$ and $Q$ are {\em weakly logically equivalent}, written $P \wlogeq Q$, if for all weak formulas $A$ it holds that $\valid{P}{A}$ iff $\valid{Q}{A}$.
\end{defi}

\subsection{Logical adequacy and expressive completeness}
\label{subsec:weak-logical-adequacy-and-expressive-completeness}

We show that the logic defined in Section~\ref{subsec:weak-logic} is adequate for weak bisimilarity.  Moreover, every finitely supported set of states that is closed under weak bisimilarity can be described by a weak formula.

\begin{thm}
  \label{thm:wbisimlog}
  $P \wbisim Q\ \implies\ P \wlogeq Q$
\end{thm}

\begin{proof}
   The proof has been formalised in Isabelle. We prove by induction over weak formulas that $P \wbisim Q$ implies that $\valid{P}A$ iff $\valid{Q}A$. The cases for conjunction and negation are immediate by induction.

Case $\wmay{\alpha} A$: Assume
  $\valid{P}{\wmay{\alpha}A}$.  Then for some
  $\wmay{\alpha'}A' = \wmay{\alpha}A$, $\exists P' \,.\, P
  \HTrans{\alpha'} P'$ and $\valid{P'}{A'}$.  Without loss of
  generality we assume also $\bn(\alpha') \freshin Q$, otherwise just
  find an alpha-variant of~$\may{\alpha'}A'$ where this holds.  Then,
  by Lemma~\ref{lem:weaktrans}, $\exists Q' \,.\, Q \HTrans{\alpha'} Q'$ and $P' \wbisim Q'$.  By induction and~$\valid{P'}{A'}$ we get~$\valid{Q'}{A'}$,
  hence by definition~$\valid{Q}{\wmay{\alpha}A}$.  The proof
  of~$\valid{Q}{\wmay{\alpha}A} \implies \valid{P}{\wmay{\alpha}A}$ is
  symmetric, using the fact that~$P \wbisim Q$ entails~$Q \wbisim P$.

Case $A = \wmay\tau(B \wedge \varphi)$.
  Assume $P\wbisim Q$ and $P \models A$.  Then $P \Trans{} P'$ such that $P' \models B$ and $P' \vdash \varphi$. By $P \wbisim Q$ and Lemma~\ref{lem:weaktrans} we obtain~$Q'$ such that $Q \Trans{} Q'$ and $P' \wbisim Q'$. By weak static implication we obtain $Q''$ with $Q' \Trans{} Q''$ and $P' \wbisim Q''$ and $Q'' \vdash \varphi$. By induction $Q'' \models B$, thus $Q'' \models B \wedge \varphi$. From $Q \Trans{} Q' \Trans{} Q''$ we have $Q \Trans{} Q''$. Thus $Q \models \wmay{\tau}(B \wedge \varphi)$ as required.

The proof that $Q \models A$ implies $P \models A$ is symmetric, using the fact that $P \wbisim Q$ entails $Q \wbisim P$.
 \end{proof}
 
\begin{lem}
  \label{lemma:wdistsup}
  If $P\not\wlogeq Q$ then there exists a distinguishing weak formula $B$ for $P$ and $Q$ such that $\n(B) \subseteq \n(P)$.
\end{lem}

\begin{proof}
The proof has been formalised in Isabelle. Since $P\not\wlogeq Q$ there is a distinguishing weak formula $A$ for $P$ and $Q$.  Let $\Pi_P$ be the set of name permutations that leave $\n(P)$ invariant and choose $B = \bigwedge \{\pi \cdot A \,|\, \pi \in \Pi_P\}$.  In the terminology of Pitts~\cite{PittsNominalSets} ch.~5,  $B$ is the conjunction of $\mbox{hull}_{\n(P)} A$; this set is supported by $\n(P)$ (but not uniformly bounded).  Because~$\models$ is equivariant we get $\valid{P}{\pi \cdot A}$ for all conjuncts $\pi \cdot A$ of $B$, and since $Q \not\models A = \mbox{id} \cdot A$ we get $Q \not\models B$.
\end{proof}

\begin{thm}
  \label{thm:wlogbisim}
  $P \wlogeq Q\ \implies P \wbisim Q$
\end{thm}

\begin{proof}
  The proof has been formalised in Isabelle. We establish that $\wlogeq$ is a weak bisimulation. Obviously it is symmetric. So assume $P\wlogeq Q$, we need to prove the two requirements on a weak bisimulation.
\begin{enumerate}
\item Weak static implication.
The proof is by contradiction. Assume that $\wlogeq$ does not satisfy the weak static implication requirement. Then there exist $P$, $Q$, $\varphi$ such that $P \wlogeq Q$ and $P \vdash \varphi$ and for all $Q'$ such that $Q \Trans{} Q'$ and $\entails{Q'}{\varphi}$ there exists a distinguishing formula $A_{Q'}$ such that $\valid{P}{A_{Q'}}$ and not $\valid{Q'}{A_{Q'}}$. By Lemma~\ref{lemma:wdistsup} $\n(A_{Q'}) \subseteq \n(P)$, which means that the infinite conjunction $A$ of all these $A_{Q'}$ is well formed. We thus have that $\wmay{\tau}(A \wedge  \varphi)$ is a distinguishing formula for $P$ and $Q$, contradicting $P \wlogeq Q$.
\item Weak simulation.
The proof is by contradiction. Assume that $\wlogeq$
  does not satisfy the weak simulation requirement. Then there exist
  $P,Q,P',\alpha$ with $\bn(\alpha) \freshin Q$ such that $P\wlogeq Q$
  and $P \trans{\alpha}P'$ and, letting ${\mathcal{Q}} = \{Q' \setsep Q
  \HTrans{\alpha} Q'\}$, for all $Q' \in {\mathcal{Q}}$ it holds that
  $P'\not\wlogeq Q'$. Assume $\bn(\alpha)\freshin P$, otherwise just find an alpha-variant of the transition satisfying this. By $P' \not\wlogeq Q'$, for all $Q' \in {\mathcal{Q}}$ there exists a
  weak distinguishing formula for $P'$ and $Q'$. The formula may depend
  on~$Q'$, and by Lemma~\ref{lemma:wdistsup} we can find such a
  distinguishing formula $B_{Q'}$ for $P'$ and $Q'$ with $\n(B_{Q'})
  \subseteq \n(P')$.  Let $B$ be the conjunction of all $B_{Q'}$.    
 We thus get that
 $\valid{P}{\wmay{\alpha} B}$ but not $\valid{Q}{\wmay{\alpha} B}$,
 contradicting $P\wlogeq Q$. \qedhere
\end{enumerate}
\end{proof}

We omit the proofs of the following results (which have also been formalised in Isabelle), as they are similar to the corresponding proofs for the full logic in Section~\ref{subsec:expressive-completeness}.  In addition, we need to verify that the constructions used below yield weak formulas.  This is immediate from Definition~\ref{def:weaklogic}.

\begin{lem}
  \label{lemma:weak-distinguishing-equivariant}
  If $P\not\wlogeq Q$, write~$B_{P,Q}$ for a distinguishing weak formula for~$P$ and~$Q$ such that $\n(B_{P,Q}) \subseteq \n(P)$.  Then \[D(P,Q) := \bigwedge_\pi \pi^{-1} \cdot B_{\pi \cdot P, \pi \cdot Q}\] defines a distinguishing weak formula for~$P$ and~$Q$ with support included in~$\n(P)$.  Moreover, the function $D$ is equivariant.
\end{lem}

\begin{defi}
  \label{def:characteristic-weak}
  A \emph{characteristic weak formula} for $P$ is a weak formula $A$ such that for all~$Q$, $P \wbisim Q$ iff~$\valid{Q}{A}$.
\end{defi}

\begin{lem}
  \label{lemma:characteristic-weak}
  Let~$D$ be defined as in Lemma~\ref{lemma:weak-distinguishing-equivariant}.  The formula $\mathrm{Char}(P) := \bigwedge_{P\not\wlogeq Q} D(P,Q)$ is a characteristic weak formula for~$P$.
\end{lem}

\begin{lem}
  \label{lemma:characteristic-weak-equivariant}
  Let~$\mathrm{Char(P)}$ be defined as in Lemma~\ref{lemma:characteristic-weak}.  The function~$\mathrm{Char}\colon P \mapsto \mathrm{Char}(P)$ is equivariant.
\end{lem}

\begin{thm}[Weak Expressive Completeness]
  \label{thm:weak-expressive-completeness}
  Let~$S$ be a finitely supported set of states that is closed under weak bisimilarity, i.e., for all~$P\in S$ and~$Q$, $P \wbisim Q$ implies $Q\in S$.  Then~$P\in S$ iff $\valid{P}{\bigvee_{P'\in S} \mathrm{Char}(P')}$.
\end{thm}

\subsection{Disjunction elimination}
\label{sect:disjunctionelemination}

As defined in Section~\ref{sec:derived-formulas}, disjunction is a derived logical operator, expressed through conjunction and negation. This is still true in the weak modal logic, but there is a twist in that neither general conjunctions nor negations may be applied to unguarded state predicates. The examples in Section~\ref{sec:weakbisim} demonstrate why this restriction is necessary: negated or conjoined state predicates in formulas would mean that adequacy no longer holds. Interestingly, we can allow disjunctions of unguarded predicates while maintaining adequacy; in fact, adding disjunction would not increase the expressive power of the logic. In this section we demonstrate this claim. An uninterested reader may skip this section without loss of continuity.

The {\em extended} weak logic is as follows, where a simultaneous induction defines both extended weak formulas (ranged over by~$E$) and preformulas (ranged over by~$B$) corresponding to subformulas with unguarded state predicates.

\begin{defi}[Extended weak formulas $E$ and preformulas $B$]
\label{extendedweaklogic}
\[
\begin{array}{l}
E \quad ::= \quad \conjunc E_i \casesep \neg E \casesep \wmay{\alpha}E  \casesep \wmay{\tau}B \\ \\
B \quad ::= \quad E\wedge B \casesep \varphi \casesep \disjunc B_i
\end{array}\]
\end{defi}

The last clause in the definition of preformulas is what distinguishes this logic from the logic in Definition~\ref{def:weaklogic}. (Thus an extended weak formula is also an ordinary weak formula if it does not contain a disjunction of unguarded state predicates.) For instance, $\wmay{\tau}(\varphi_0 \vee \varphi_1)$ is an extended weak formula, as is
\[\wmay{\tau}( ((\wmay{\beta}\top) \wedge \varphi_0) \vee ((\wmay{\gamma}\top) \wedge \varphi_1))\]
saying that it is possible to do a sequence of unobservable actions such that either continuing with $\beta$ and satisfying~$\varphi_0$ hold, or continuing with $\gamma$ and satisfying $\varphi_1$ hold.

\begin{lem}
  \label{disjunctdistrib}
  $\wmay{\tau}\disjunc B_i$ is a formula iff $\disjunc \wmay{\tau} B_i$ is, and in this case
  \[\wmay{\tau}\disjunc B_i \wlogeq \disjunc \wmay{\tau} B_i\]
\end{lem}

\begin{proof}
  The proof is by directly expanding definitions, noting that $\n(\{\wmay{\tau} B_i \mid i \in I\})$ is equal to $\n(\{B_i \mid i \in I\})$.
\end{proof}

\begin{thm}
\label{thm:disjunctionelim}
  For any extended weak formula~$E$ there is an (ordinary) weak
  formula~$\Delta(E)$ such that~$E \wlogeq \Delta(E)$.
\end{thm}
\begin{proof}
  The idea is to push disjunctions in preformulas to top level using the fact that (finite) conjunction distributes over disjunction, and then use Lemma~\ref{disjunctdistrib}.
  Say that a preformula is in {\em normal form} if it either is $A \wedge \varphi$ where $A$ here and in the following stands for an ordinary weak formula, or is a disjunction of normal forms. We let $C$ range over normal preformulas, thus
\[C \quad ::= \quad A \wedge \varphi \casesep \bigvee_{i \in I} C_i\]
The intuition is that in a normal preformula no conjunction can contain a disjunction of preformulas.

  Define a function~$\delta$ from normal preformulas to (ordinary) weak
  formulas by
  \[
  \begin{array}{l@{\hspace{0.5em}}c@{\hspace{1em}}l}
    \delta(A\wedge\varphi) &=& \wmay{\tau}(A\wedge\varphi) \\[.5em]
    \delta(\disjunc C_i)                          &=& \disjunc \delta(C_i)
  \end{array}
  \]
  By induction~$\delta$ is equivariant, hence the disjunction in the
  above definition is finitely supported.  Moreover, $\delta(C)
  \wlogeq \wmay{\tau}C$ by Lemma~\ref{disjunctdistrib} and induction.

  Next, define a binary function~$\varepsilon$ from pairs consisting of an (ordinary) weak formula and a normal preformula, to normal preformulas,
by
  \[
  \begin{array}{l@{\hspace{0.5em}}c@{\hspace{1em}}l}
    \varepsilon(A, A'\wedge\varphi) &=& (A\wedge A') \wedge\varphi \quad \\[.5em]
    \varepsilon(A, \disjunc C_i) &=& \disjunc \varepsilon(A, C_i)
  \end{array}
  \]
  Also~$\varepsilon$ is equivariant, and $\varepsilon(A, C) \wlogeq
  A\wedge C$ by induction and distributivity of conjunction.

  We now provide an explicit transformation~$\Delta$ from extended
  weak formulas to (ordinary) weak formulas.  We also provide a
  transformation~$\Delta_\mathrm{pre}$ that maps preformulas to
  normal preformulas.  The transformations~$\Delta$
  and~$\Delta_\mathrm{pre}$ are defined by mutual recursion:
  \[
  \begin{array}{l@{\hspace{0.5em}}c@{\hspace{1em}}l}
    \Delta(\conjunc E_i)   &=& \conjunc \Delta(E_i)\\[.5em]
    \Delta(\neg E)         &=& \neg \Delta(E)\\[.5em]
    \Delta(\wmay{\alpha}E) &=& \wmay{\alpha}\Delta(E) \\[.5em]
    \Delta(\wmay{\tau}B)   &=& \delta(\Delta_\mathrm{pre}(B))\\[1em]
    \Delta_\mathrm{pre}(E\wedge B)    &=& \varepsilon(\Delta(E), \Delta_\mathrm{pre}(B))\\[.5em]
    \Delta_\mathrm{pre}(\varphi)      &=& \top\wedge\varphi\\[.5em]
    \Delta_\mathrm{pre}(\disjunc B_i) &=& \disjunc \Delta_\mathrm{pre}(B_i)
  \end{array}
  \]

  By simultaneous induction over extended weak formulas and
  preformulas it is easy to prove that both~$\Delta$
  and~$\Delta_\mathrm{pre}$ are equivariant.  Therefore, all
  conjunctions and disjunctions that appear on the right-hand side in
  the above definition are finitely supported.

  We now prove, again by simultaneous induction over extended weak
  formulas~$E$ and preformulas~$B$, that $\Delta(E) \wlogeq E$ and
  $\Delta_\mathrm{pre}(B) \wlogeq B$.
  The cases~$\Delta(\conjunc E_i)$ and~$\Delta(\neg E)$ and
  $\Delta(\wmay{\alpha}E)$ are immediate by induction.  For the
  case~$\Delta(\wmay{\tau}B)$ we have
  \[\Delta(\wmay{\tau}B) = \delta(\Delta_\mathrm{pre}(B)) \wlogeq \wmay{\tau}\Delta_\mathrm{pre}(B) \wlogeq \wmay{\tau}B\]
  by induction.  Likewise, for the case~$\Delta_\mathrm{pre}(E\wedge
  B)$ we have
  \[\Delta_\mathrm{pre}(E\wedge B) = \varepsilon(\Delta(E), \Delta_\mathrm{pre}(B)) \wlogeq \Delta(E)\wedge\Delta_\mathrm{pre}(B) \wlogeq E\wedge B\]
  by induction.  The case~$\Delta_\mathrm{pre}(\varphi)$ is trivial,
  and the case $\Delta_\mathrm{pre}(\disjunc B_i)$ is again immediate
  by induction.  
\end{proof}

\section{State predicates versus actions}
\label{sec:state-predicates-versus-actions}

\newcommand{\tsyst}{\mathbf{T}}
\newcommand{\transf}{\mathcal S}
\newcommand{\stsyst}{{\transf(\tsyst)}}
\newcommand{\validtr}[3]{#1 \nobreak \, \models_{#2} \nobreak \, #3 }

Transition system formalisms differ in how much information is considered to reside in states and how much is considered to reside in actions. One extreme is Lamport's TLA~\cite{specifying-systems-the-tla-language-and-tools-for-hardware-and-software-engineers} where all information is in the states. On the other extreme are most process algebras such as the pi-calculus, where states contain no information apart from the outgoing transitions.  Advanced process algebras such as psi-calculi use both state predicates and actions. Clearly, many ways are possible and the choice is more dependent on the traditions and modelling convenience in different areas than on hard theoretical results.

The question if state information can be encoded as actions and vice versa is old. Already in 1989 Jonsson et al.~\cite{Jonsson1990} provided a translation from a labelled transition system into an (unlabelled) Kripke structure, in order to use a model checker for CTL. With nominal transition systems it is hard to see how such a translation could work. The actions can contain binding names and encoding that into the states would require a substantial extension of our definitions. 

In the other direction, it has long been a folklore fact in process algebra that a state predicate can be represented as an action on a transition leading back to the same state. In this section we make this idea fully formal: for the purposes of checking bisimilarity this transformation is indeed sound, and there is a companion transformation on modal logic formulas. This is straightforward for strong bisimulation and modal logic, and a little more involved for weak modal formulas.

\subsection{Strong bisimulation and logic}
Our idea is that for any transition system $\tsyst$ there is another transition system $\stsyst$ where state predicates are replaced by self-loops.  To formulate this idea we again use the notation $\states_\tsyst$ to mean the states in the transition system $\tsyst$, and similarly for actions, $\bn$, transitions, bisimilarity, etc.

\begin{defi}
  The function $\transf$ from transition systems to transition systems is defined as follows:
  \begin{itemize}
  \item $\states_\stsyst = \states_\tsyst$
  \item $\actions_\stsyst = \actions_\tsyst \uplus \predicates_\tsyst$
  \item $\bn_\stsyst(\alpha) = \bn_\tsyst(\alpha)$ if $\alpha \in
    \actions_\tsyst$; $\bn_\stsyst(\varphi) = \emptyset$ if $\varphi \in
    \predicates_\tsyst$
  \item $\predicates_\stsyst = \;\vdash_\stsyst\; = \emptyset$
  \item $P \trans{\alpha}_\stsyst P'$ if $P \trans{\alpha}_\tsyst P'$
    (for $\alpha \in \actions_\tsyst$); $P \trans{\varphi}_\stsyst P$ if
    $P\vdash_\tsyst \varphi$ (for $\varphi \in \predicates_\tsyst$)
  \end{itemize}
\end{defi}

It is easy to see that if $\tsyst$ is a transition system then so is $\stsyst$. In particular equivariance of~$\rightarrow_\stsyst$ follows from equivariance of~$\rightarrow_\tsyst$ and~$\vdash_\tsyst$ and the fact that the union of equivariant relations is equivariant.\footnote{Nominal Isabelle does not admit empty types.  Therefore, the type~$\predicates_\stsyst$ is given by the unit type in our Isabelle formalisation.  We still define~$\vdash_\stsyst$ to be false everywhere, so that this minor difference has no further significance.}

\begin{thm}
  $P \bisim_\tsyst Q\ \implies\ P \bisim_\stsyst Q$
\end{thm}

\begin{proof}
  The proof has been formalised in Isabelle. We show that $\bisim_\tsyst$ is a $\stsyst$-bisimulation. Obviously it is symmetric. Static implication is trivial since $\vdash_\stsyst$ is false everywhere. For simulation, assume $P \bisim_\tsyst Q$ and that $P$ has a transition in $\stsyst$. If the transition has an action in $\tsyst$, then $P$ has the same transition in $\tsyst$. If the action is a state predicate in~$\tsyst$ then the transition must be $P \trans{\varphi}_\stsyst P$ where $P \vdash_\tsyst \varphi$, hence by static implication $Q \vdash_\tsyst \varphi$, thus also $Q \trans{\varphi}_\stsyst Q$. In both cases we therefore have a simulating transition from $Q$ in~$\stsyst$.
\end{proof}

\begin{thm}
  \label{thm:s_implies_t}
  $P \bisim_\stsyst Q\ \implies\ P \bisim_\tsyst Q$
\end{thm}

\begin{proof}
  The proof has been formalised in Isabelle. We show that $\bisim_\stsyst$ is a $\tsyst$-bisimulation. Obviously it is symmetric. Assume $P \bisim_\stsyst Q$. For static implication, if $P\vdash_\tsyst \varphi$ then~$P$ has a transition with action $\varphi$ in $\stsyst$, thus $Q$ has a transition with the same action, hence $Q \vdash_\tsyst \varphi$. For simulation, any transition in $\tsyst$ is also a transition in $\stsyst$.
 \end{proof}
 
Modal logic formulas over $\stsyst$ use predicates as actions. We extend $\mathcal{S}$ to formulas as follows.
\begin{defi}
  The function $\transf$ from formulas over the transition system $\tsyst$ to formulas over the transition system $\stsyst$ is defined by
  \[\transf(\varphi) = \may{\varphi}\top\]
  and is homomorphic on the other cases in Definition~\ref{def:formulas}.
\end{defi}

\begin{thm}
  \label{thm:slog}
  $P \models_\tsyst A \quad \mbox{iff} \quad P \models_\stsyst \transf(A)$
\end{thm}

\begin{proof}
  The proof has been formalised in Isabelle.  It is by induction over formulas~$A$, for arbitrary~$P$.  If $A = \varphi$ then clearly $P \models_\tsyst \varphi$ iff $P \models_\stsyst \may{\varphi}\top$.  For the case $A = \may{\alpha}A'$ we assume without loss of generality $\bn(\alpha) \freshin P$, otherwise just find an alpha-variant where this holds.  This and the other cases are then immediate by induction.
\end{proof}

This provides an alternative proof of Theorem~\ref{thm:s_implies_t}. If $P \bisim_\stsyst Q$ then by Theorem~\ref{thm:bisimlog} for any $A$ it holds that $P \models \transf(A)$ iff $Q \models \transf(A)$, and thus $P \models A$ iff $Q \models A$, therefore by Theorem~\ref{thm:logbisim} $P\bisim_\tsyst Q$. The converse does not follow since $\transf$ is not surjective on formulas.

\subsection{Weak bisimulation }

For weak bisimulation the corresponding proofs are only a little bit more complicated. 

\begin{thm}
  \label{thm:timpliesst}
  $P \wbisim_\tsyst Q\ \implies\ P \wbisim_\stsyst Q$
\end{thm}

\begin{proof}
  The proof has been formalised in Isabelle. We show that $\wbisim_\tsyst$ is a weak $\stsyst$-bisimulation. Obviously it is symmetric. Weak static implication is trivial since $\vdash_\stsyst$ is false everywhere. For weak simulation, assume $P \wbisim_\tsyst Q$. There are two cases.
  \begin{enumerate}
  \item $P \trans{\alpha}_\stsyst P'$ where $\alpha \in \actions_\tsyst$. Then $P \trans{\alpha}_\tsyst P'$ and by $P \wbisim_\tsyst Q$ we get $Q \HTrans{\alpha}_\tsyst Q'$,  which implies $Q \HTrans{\alpha}_\stsyst Q'$, with $P' \wbisim_\tsyst Q'$.
  \item $P \trans{\varphi}_\stsyst P$ where $\varphi \in \predicates_\tsyst$ and $P\vdash_\tsyst \varphi$. By weak static implication $Q \Trans{}_\tsyst Q'$ and $Q' \vdash_\tsyst \varphi$ and $Q' \wbisim_\tsyst P$. Thus $Q\Trans{}_\stsyst Q'$ and $Q' \trans{\varphi}_\stsyst Q'$, hence $Q \HTrans{\varphi}_\stsyst Q'$. \qedhere
  \end{enumerate}
\end{proof}

The converse uses the following lemma, which is familiar in many process algebras and interesting in its own right.

\begin{lem}
  \label{lem:interpolate}
  If $P \Trans{} Q \Trans{} R$ and $P \wbisim R$ then $Q \wbisim R$.
\end{lem}

\begin{proof}
  The proof has been formalised in Isabelle.  We establish that $\{(Q,R), (R,Q)\}\; \cup \wbisim$ is a weak bisimulation.  Obviously it is symmetric.

  For weak static implication, clearly if $R \vdash \varphi$ then $Q \Trans{} R \vdash \varphi$; and if $Q \vdash \varphi$ then $P\Trans{} Q \vdash \varphi$, so by $P \wbisim R$, Lemma~\ref{lem:weaktrans}, and weak static implication we obtain $R'$ with $R \Trans{} R' \vdash \varphi$ as required.

  For weak simulation, clearly if $R \trans{\alpha} R'$ then $Q \Trans{} R \trans{\alpha} R'$, hence $Q \HTrans{\alpha} R'$ (and obviously $R' \wbisim R'$); and if $Q \trans{\alpha} Q'$ then $P \HTrans{\alpha} Q'$, so again by $P \wbisim R$ and Lemma~\ref{lem:weaktrans} we obtain~$R'$ with $R \HTrans{\alpha} R'$ and $Q' \wbisim R'$.
\end{proof}

\begin{thm}
  \label{thm:stimplt}
  $P \wbisim_\stsyst Q\ \implies\ P \wbisim_\tsyst Q$
\end{thm}

\begin{proof}
  The proof has been formalised in Isabelle.  We show that $\wbisim_\stsyst$ is a weak $\tsyst$-bisimulation. Obviously it is symmetric.

  For weak static implication, assume $P \wbisim_\stsyst Q$ and $P \vdash_\tsyst \varphi$. Thus $P \trans{\varphi}_\stsyst P$. By weak simulation we obtain~$Q'$ with $Q \HTrans{\varphi}_\stsyst Q'$ and $P \wbisim_\stsyst Q'$. By construction of~$\stsyst$ this means that there is a $Q''$ such that $Q \Trans{}_\stsyst Q'' \trans{\varphi}_\stsyst Q'' \Trans{}_\stsyst{Q'}$. Thus $Q \Trans{}_\tsyst Q''$ and $Q'' \vdash_\tsyst \varphi$. Moreover, by applying Lemma~\ref{lem:interpolate} to $Q \Trans{}_\stsyst Q'' \Trans{}_\stsyst{Q'}$ and $Q \wbisim_\stsyst P \wbisim_\stsyst Q'$ we get $P \wbisim_\stsyst Q''$ as required.

  For weak simulation, assume $P \wbisim_\stsyst Q$ and $P \trans{\alpha}_\tsyst P'$. This implies $P \trans{\alpha}_\stsyst P'$, and thus by weak simulation $Q \HTrans{\alpha}_\stsyst Q'$, which implies $Q \HTrans{\alpha}_\tsyst Q'$, with $P' \wbisim_\stsyst Q'$ as required.
\end{proof}

\subsection{Weak logic}
For the weak logic, the correspondence between state predicates and actions is less obvious, and
 there appear to be more than one alternative.
The transformation $\transf$ of formulas does not preserve the property of being a weak formula since
\[\transf(\wmay{\tau}(A \wedge \varphi)) =  \wmay{\tau}( \transf(A) \wedge \may{\varphi}\top)\]
and the subformula $\may{\varphi}\top$ is not weak. For the transition system $\stsyst$ we thus get an alternative logic adequate for weak bisimilarity by taking the formulas $\{\transf(A) \setsep \mbox{$A$ is a weak formula}\}$, i.e., with the last clause of Definition~\ref{def:weaklogic} replaced by formulas of kind $\wmay{\tau}(A \wedge \may{\varphi}\top)$.

\begin{cor}
  $P \wbisim_\stsyst Q$ \quad iff \quad for all weak formulas $A$: $P\models_\stsyst \transf(A)$ iff $Q\models_\stsyst \transf(A)$
\end{cor}

\begin{proof}
  The proof has been formalised in Isabelle.  It is immediate by Theorems~\ref{thm:wbisimlog}, \ref{thm:wlogbisim} and \ref{thm:slog}--\ref{thm:stimplt}.
\end{proof}

For this corollary to hold it is critical that the transition system is an image of $\transf$, i.e., that the transitions involving actions $\varphi$ only occur in loops. If this is not the case, there may be formulas in the image of $\transf$ that can distinguish between weakly bisimilar states. An example is the following:

\vspace{0.3cm}

\includegraphics{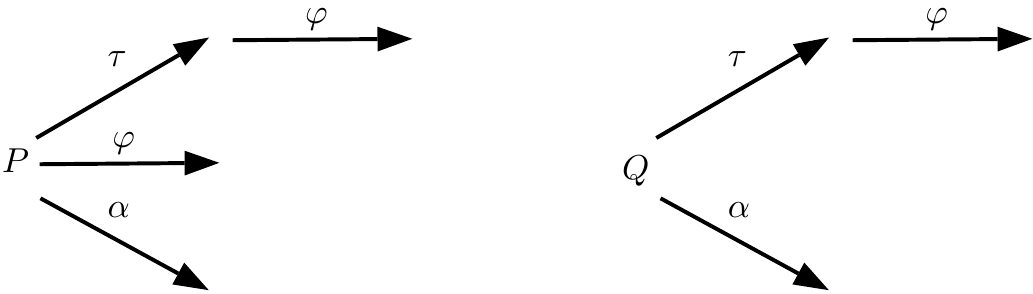}

\vspace{0.3cm}

\noindent Here $P$ and $Q$ are weakly bisimilar. Let $A$ be the weak formula $\wmay{\tau}(\wmay{\alpha}\top \wedge \varphi)$. Then $\transf(A)= \wmay{\tau}(\wmay{\alpha}\top \wedge \may{\varphi}\top)$, and $P\models\transf(A)$ but not $Q \models\transf(A)$.

Ideally, we would want an alternative transformation onto weak formulas, but it seems difficult to formulate this in a succinct way. Through expressive completeness (Section~\ref{subsec:weak-logical-adequacy-and-expressive-completeness}) we get the following:

\begin{thm}
  \label{thm:weak-transformation}
  Let~$A$ be a weak formula over~$\tsyst$.  Then there exists a weak formula~$A'$ over~$\stsyst$ such that
  \[ P \models_\tsyst A \quad \mbox{iff} \quad P \models_\stsyst A' \]
  for all states~$P$.
\end{thm}

\begin{proof}
  The proof has been formalised in Isabelle.  Note that $\{Q\mid Q \models_\tsyst A\}$ is finitely supported (with support included in~$\n(A)$), and moreover closed under~$\mathop{\wbisim_\stsyst}$ by Theorems~\ref{thm:wbisimlog} and~\ref{thm:stimplt}.  It then follows from Theorem~\ref{thm:weak-expressive-completeness} that $A' = \bigvee_{Q \models_\tsyst A} \mathrm{Char}(Q)$ has the desired properties.
\end{proof}

Inherent in this proof is a construction of the weak formula $A'$ from the weak formula~$A$,
but the construction depends implicitly on distinguishing formulas (cf.~Theorem~\ref{thm:logbisim}) and thus on the entire transition system.

We have failed to find a transformation of weak formulas that yields equivalent weak formulas over~$\stsyst$ and is defined by induction over formulas, or, at the very least, is independent of the transition system.  We have also failed to prove that such a transformation cannot exist, and thus have to leave this problem open. The following counterexamples shed light on the difficulties:

\paragraph{{\bf Counterexample}} Define the transformation $\transf$ on weak formulas by

\[\transf(\wmay{\tau}(A \wedge \varphi))= \wmay{\varphi}\transf(A)\]

With this definition, the counterpart of Theorem~\ref{thm:slog} fails. A counterexample is $A = \neg\wmay{\alpha}\top$,  $P \vdash_\tsyst{\varphi}$ with $P \trans{\tau}_\tsyst Q$ and $P \trans{\alpha}_\tsyst Q$ for some $\alpha \neq \tau$, where $Q$ has no outgoing transitions, cf.\ the diagrams below:

\vspace{0.3cm}

\includegraphics{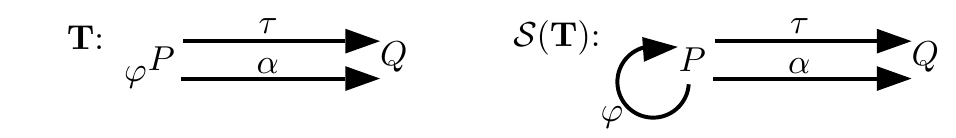}

\noindent
Since $P \Trans{\varphi}_\stsyst 
Q$ and $Q$ has no $\wmay{\alpha}$ action, we have that
\[\validtr{P}{\stsyst}{\wmay{\varphi} \neg \wmay{\alpha} \top}\]
The only state in {\bf T} that satisfies $\varphi$ also has an $\wmay{\alpha}$ action, thus it does {\em not} hold that
\[\validtr{P}{\tsyst}{\wmay{\tau}(\neg \wmay{\alpha} \top \wedge \varphi)}\]

\paragraph{\bf Counterexample} Define the transformation $\transf$ on weak formulas by
\[\transf(\wmay{\tau}(A \wedge \varphi))= \wmay{\tau}(\transf(A) \wedge \wmay{\varphi}\top)\]
Then again the counterpart of Theorem~\ref{thm:slog} fails. A counterexample here is $A = \wmay{\alpha}\top$, with $P \trans{\tau}_\tsyst Q$ and $Q \vdash_\tsyst{\varphi}$, $P \trans{\alpha}_\tsyst Q$ for some $\alpha \neq \tau$, cf.\ the diagrams below:

\vspace{0.3cm}

\includegraphics{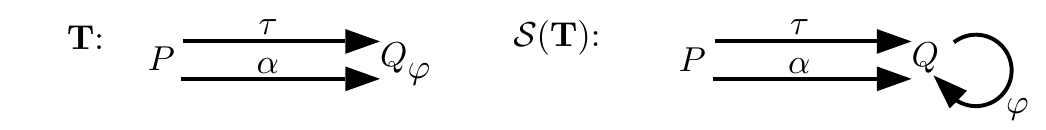}

\noindent
Here it holds that
\[\validtr{P}{\stsyst}{\wmay{\tau}(\wmay{\alpha}\top \wedge \wmay{\varphi}\top)}\]
and it does {\em not} hold that
\[\validtr{P}{\tsyst}{\wmay{\tau}(\wmay{\alpha} \top \wedge \varphi)}\]

Finally, consider
the {\em partial} transformation $\transf'$ on weak formulas by
  \[\transf'(\wmay{\tau}((\wmay{\tau}A) \wedge \varphi)) = \wmay{\varphi}\transf'(A)\]
  where $\transf'$ is homomorphic on the first three cases in Definition~\ref{def:weaklogic}.

$\transf'$ is not total since a formula $\wmay{\tau}(A \wedge \varphi)$ is in its domain only when $A=\wmay{\tau}A'$ for some~$A'$. It is easy to see that $\transf'$ is injective and surjective, i.e., every weak formula $A$ on~$\stsyst$ has a unique formula $B$ on $\tsyst$ such that $\transf'(B)=A$. We write $\transf'^{-1}$ for the inverse of $\transf'$. Thus 
\[\transf'^{-1}(\wmay{\varphi}A) = \wmay{\tau}((\wmay{\tau}\transf'^{-1}(A)) \wedge \varphi)\]
and $\transf'^{-1}$ is homomorphic on all other operators.

\begin{thm}
  \label{thm:validtransform}
  $\validtr{P}{\stsyst}{A} \quad \mbox{iff} \quad \validtr{P}{\tsyst}{\transf'^{-1}(A)}$
\end{thm}

\begin{proof}
  By induction over weak formulas on $\stsyst$. All cases are trivial by induction except for the case $\wmay{\varphi}A$ with $\varphi \in \predicates_\tsyst$. Suppose $\validtr{P}{\stsyst}{\wmay{\varphi}A}$. Then $P \Trans{}_\stsyst P' \trans{\varphi}_\stsyst P'' \Trans{}_\stsyst Q$ with $\validtr{Q}{\stsyst}{A}$. By construction of $\stsyst$ we get $P'=P''$, and $P \Trans{}_\tsyst P' \Trans{}_\tsyst Q$ with $P' \vdash_\tsyst \varphi$. By induction $\validtr{Q}{\tsyst}{\transf'^{-1}}(A)$. This establishes that $\validtr{P'}{\tsyst}{\wmay{\tau}\transf'^{-1}(A) \wedge \varphi}$ and thus $\validtr{P}{\tsyst}{\wmay{\tau}(\wmay{\tau}\transf'^{-1}(A) \wedge \varphi)}$ as required. Conversely, if $\validtr{P}{\tsyst}{\wmay{\tau}(\wmay{\tau}\transf'^{-1}A \wedge \varphi)}$ then $P \Trans{}_\tsyst P' \Trans{}_\tsyst Q$, and by construction of $\stsyst$ and induction we get $\validtr{P}{\stsyst}{\wmay{\varphi} A}$.
\end{proof}

An interesting consequence is that to express the distinguishing formulas guaranteed by~Theorem~\ref{thm:wlogbisim}, it is enough to consider formulas in $\mbox{dom}(\transf')$, i.e., in the last clause of Definition~\ref{def:weaklogic}, it is enough to consider $A = \wmay{\tau}A'$. The reason is that if $P \not\wbisim_\tsyst Q$ then by Theorem~\ref{thm:stimplt} also $P \not\wbisim_\stsyst Q$, which by Theorem~\ref{thm:wlogbisim} means there is a distinguishing formula~$B$ for~$P$ and $Q$ in $\stsyst$, which by Theorem~\ref{thm:validtransform} means that $\transf'^{-1}(B)$ is a distinguishing formula in $\tsyst$. 

In conclusion, the results in this section show that the transformation $\transf$ indicates an alternative weak logic with formulas $\wmay{\tau}(A \wedge \may{\varphi}\top)$, which is appropriate if $\varphi$ only occurs in self-loops. The transformation $\transf'$, on the other hand, indicates that a sublogic with formulas $\wmay{\tau}(\wmay{\tau}A \wedge \varphi)$ is expressively equivalent. It admits a smooth transformation of predicates into actions at the cost of a more complicated definition of a weak logic.

Through expressive completeness we obtain a transformation on our original weak logic,  but it does not really shed light on the nature of state predicates and actions. Certainly, it means that any impossibility result for a transformation onto weak formulas must be qualified with a notion of independence from the transition system. This would require a proper definition of the set of transition systems that act as models for a given logic.
One difficulty here is that the disjoint union of a set of transition systems is itself a transition system. By expressive completeness we get a transformation (dependent on this union); that same transformation thus applies to all the members in the set. Consider the set of all transition systems of cardinality at most $\kappa$; this set may have a cardinality higher than $\kappa$.
Thus an impossibility result may actually depend on the cardinality limit in Definition~\ref{def:formulas}.

\newcommand{\outprefix}[1]{\overline{#1}}
\newcommand{\inprefix}[1]{\underline{#1}}

\newcommand{\inlabel}[1]{\underline{#1}\,}
\newcommand{\outlabel}[2]{\overline{#1}#2}
\newcommand{\boutlabel}[3][\tilde{c}]{\overline{#2}(\nu\,#1)#3}
\newcommand{\substlabel}[2]{{#1}/{#2}}

\newcommand{\free}[1]{\langle{#1}\rangle}
\newcommand{\procpair}[2]{#1 \mathrel{\triangleright} #2}
\newcommand{\pair}[2]{(#1,#2)}

\newcommand{\fresh}[1][\;]{#1\operatorname{fresh}}
\newcommand{\names}[1]{\n(#1)}
\newcommand{\dom}[1]{\operatorname{dom}(#1)}
\newcommand{\range}[1]{\operatorname{range}(#1)}

\newcommand{\class}[2][P]{[#2]_{#1}}
\newcommand{\encode}[2][P]{[\![{#2}]\!]_{#1}}

\newcommand{\known}[1]{\mathcal{S}(#1)}
\newcommand{\eval}[1]{\mathbf{e}(#1)}
\newcommand{\enc}[2]{\operatorname{E}(#1,#2)}
\newcommand{\dec}[2]{\operatorname{D}(#1,#2)}

\section{Applications} \label{sec:applications} 
In this section we consider standard process calculi and their accompanying labelled bisimilarities, and investigate how to obtain an adequate modal logic using our framework. In each of the final two examples, no HML has to our knowledge yet been proposed, and we immediately obtain one by instantiating the logic in the present paper.

\topic{Pi-calculus} The pi-calculus by Milner et al.~\cite{MPWpi} has a labelled transition system with early input, defined using a structural operational semantics. This transition system satisfies the axioms of a nominal transition system, and the logic of Section~\ref{sec:nominaltransitionsystems} (with an empty set of state predicates) is adequate for early bisimilarity.  Other variants of bisimilarity can be obtained by instead using the original transition system (with late input) and one of the variants of our logic described in Section~\ref{sec:variants}.

The pi-calculus already has several notions of weak bisimulation, and Definition~\ref{def:weakbisim} corresponds to the early weak bisimulation. In the pi-calculus there are no state predicates, thus the weak static implication is unimportant. The weak logic of Definition~\ref{def:weaklogic} is adequate for early weak bisimulation.

\topic{Applied pi-calculus} The applied pi-calculus by Abadi and Fournet (2001)~\cite{abadi.fournet:mobile-values} comes equipped with 
a labelled transition system and a notion of weak labelled bisimulation.  
States contain a record of emitted messages; 
this record has a domain and can be used to equate open terms $M$ and $N$ modulo some rewrite system.  
The definition of bisimulation requires bisimilar processes to have the same domain and equate the same open terms, i.e., to be strongly statically equivalent.  
We model these requirements using state predicates ``$x\in\operatorname{dom}$'' and ``$M\equiv N$''.
Since satisfaction is invariant under silent transitions, weak and strong static implication coincide, and our weak HML is adequate for Abadi and Fournet's early weak labelled bisimilarity.

\topic{Spi-calculus}
The spi-calculus by Abadi and Gordon (1999)~\cite{abadi.gordon:calculus-cryptographic} has a formulation as an environment-sensitive labelled transition system by Boreale et al.~(2001)~\cite{bdp01cryptographic} equipped with state formulae (predicates) $\phi$.  Adding state predicates ``$x\in\operatorname{dom}$'' to the state predicates makes our weak HML adequate with respect to Boreale's early weak bisimilarity.

\topic{Concurrent constraint pi calculus}
The concurrent constraint pi calculus (CC-pi) by Buscemi and Montanari (2007)~\cite{buscemi.montanari:cc-pi} 
extends the explicit fusion calculus~\cite{gardner.wischik:explicit-fusions} with a more general notion of constraint stores $c$.
  Using the labelled transition system of CC-pi and the associated
  bisimulation (Definition~\ref{def:bisim}), we immediately get an
  adequate modal logic.

The reference equivalence for CC-pi is open bisimulation~\cite{buscemi08openCCpi} 
(closely corresponding to hyperbisimulation in the fusion calculus~\cite{fusion}), 
which differs from labelled bisimulation in two ways:
First, two equivalent processes must be equivalent under all store extensions.
  To encode this, we let the effects $\mathcal{F}$ be the set of constraint stores $c$ different from $0$, and let $c(P)= c\mid P$.
Second,
  when simulating a labelled transition $P\trans\alpha P'$, 
  the simulating process $Q$ can use any transition $Q\trans\beta Q'$ with an equivalent label, 
  as given by a state predicate $\alpha=\beta$.
  As an example, if $\alpha = \overline a\free x$ is a free output label then $P\vdash\alpha=\beta$
  iff $\beta=\overline b\free y$ where $P\vdash a=b$ and $P\vdash x=y$.
  To encode this, we transform the labels of the transition system by replacing them with their equivalence classes, i.e., 
  $P\trans{\alpha}P'$ becomes $P\trans{\class{\alpha}}P'$ where $\beta\in\class{\alpha}$ iff $P\vdash \beta=\alpha$.
Hyperbisimilarity (Definition~\ref{def:alternatebisim}) on this transition system then corresponds to open bisimilarity, 
and the modal logic defined in Section~\ref{sec:variants} is adequate.

\topic{Psi-calculi}
In psi-calculi by Bengtson et al.~(2011)~\cite{PsiLMCS}, the labelled transitions take the form $\procpair \Psi P \trans\alpha P'$,
where the assertion environment $\Psi$ is unchanged after the step.  
We model this as a nominal transition system by letting the set of states be pairs $\pair \Psi P$ of assertion environments and processes, 
and define the transition relation by
  $\pair \Psi P \trans\alpha \pair \Psi{P'}$\linebreak if $\procpair \Psi P \trans\alpha P'$.
The notion of bisimulation used with psi-calculi also uses an assertion environment and is required to be closed under environment extension, 
i.e., if $\procpair \Psi P \sim Q$, then $\procpair{\Psi\otimes\Psi'} P \sim Q$ for all $\Psi'$.
We let the effects~$\mathcal{F}$ be the set of assertions, and\linebreak define $\Psi(\pair{\Psi'}P)= \pair{\Psi\otimes\Psi'}P$.
Hyperbisimilarity on this transition system then subsumes the standard psi-calculi bisimilarity, 
and the modal logic defined in Section~\ref{sec:variants} is adequate.

\topic{Weak bisimilarity for CC-pi and psi}
Both CC-pi and psi-calculi has a special unobservable action $\tau$, but until now only psi-calculi have a notion of weak labelled bisimulation (as remarked in Section~\ref{sec:weakbisim}), and neither has a weak HML. Through this paper they both gain both weak bisimulation and logic, although more work is needed to establish how compatible the bisimulation equivalence is with their respective syntactic constructs.
A complication is that the natural formulation of bisimulation makes use of effects (store or assertion extensions) which are bisimulation requirements on neither predicates nor actions. In order to map them into our framework these would need to be cast as actions. Their interactions with the silent action could be an interesting topic for further research.

\section{Related work}
\label{sec:related}
We here discuss other modal logics for process calculi, 
with a focus on how their constructors can be captured by finitely supported conjunction in our HML.
This comparison is by necessity somewhat informal:
 formal correspondence fails to hold due to differences in the conjunction operator of the logic (finite, uniformly bounded or unbounded vs.~bounded support).

\topic{HML for CCS}
The first published HML is Hennessy and Milner (1980--1985)~\cite{DBLP:conf/icalp/HennessyM80,DBLP:conf/caap/Milner81,DBLP:journals/jacm/HennessyM85}.  
They work with image-finite CCS processes, 
  where finite (binary) conjunction suffices for adequacy, 
  and define both strong and weak versions of the logic. 
The logic by Hennessy and Liu (1995)~\cite{HennessyLiu1995} for a value-passing calculus also uses binary conjunction, 
  where image-finiteness is due to a late semantics and the logic contains quantification over data values. 
A similar idea and argument is in a logic for LOTOS by Calder et al.~(2002)~\cite{calder2002modal}, 
   though that paper only considers stratified bisimilarity up to $\omega$.

 Hennessy and Liu's value-passing calculus is based on 
ordinary CCS. 
In this calculus, a receiving process $a(x).P$ can participate in a synchronisation on $a$, becoming an abstraction $(x)P$ where $v$ is a bound variable. 
Dually, a sending process $\overline{a}\,v.Q$ becomes a concretion $(v, Q)$ where $v$ is a value.
The abstraction and concretion above react as part of the synchronisation on $a$, yielding $P\ssubst vx \mid Q$.
To capture the operations of abstractions and concretions in our framework,
we add effects $\stateid$ and $?v$, with $?v((x)P)= P\ssubst vx$, and transitions $(v, P)\trans{!v}P$.  
Letting $L(a?,\_)=\{?v \setsep v\in\text{values}\}$ and $L(\alpha,\_)=\{\stateid\}$ otherwise,
late bisimilarity is $\{\stateid\}/L$-bisimilarity as defined in Section~\ref{sec:variants}.
We can then encode their 
  universal quantifier $\forall x.A$ as $\bigwedge_{v}\may{?v}A\ssubst vx$,  
  which has support $\names{A}\setminus\{x\}$,
and their 
  output modality $\may{c!x}A$ as $\may{c!}\bigvee_{v}\may{!v}A\ssubst vx$,  
  with support $\{c\}\cup(\names{A}\setminus\{x\})$.

An infinitary HML for CCS is discussed in Milner's book (1989)~\cite{MilnerCCS}, where also the process syntax contains infinite summation.
There are no restrictions on the indexing sets and no discussion about how this can exhaust all names.
The adequacy theorem is proved by stratifying bisimilarity  and using transfinite induction over all ordinals, where the successor step basically is the contraposition of the argument in Theorem~\ref{thm:logbisim}, though without any consideration of finite support. A more rigorous treatment of the same ideas is by Abramsky (1991)~\cite{DBLP:journals/iandc/Abramsky91} where uniformly bounded conjunction is used throughout.

Koutavas et al.~(2018)~\cite{Koutavas18logicTransactions} study a transactional CCS,
extending the natural transition system in order to provide three equivalent adequate HMLs.
Neither processes nor actions contain bound names, which is a stark difference from the problem treated in the present paper.
Formulas contain countably infinite conjunction and thus may contain an infinite number of variable names, which are ``interpreted nominally''. 
The paper does not define alpha-equivalence for formulas,
and states that variable instantiation "does not require any notion of alpha-equivalence, as in~\cite{DBLP:conf/concur/ParrowBEGW15}".
In order to find a fresh variable name, the adequacy proof instead uses an unspecified notion of renaming (Prop.~4.5).

Simpson (2004)~\cite{Simpson04:GSOSsequentHML} gives a sequent calculus for proving HML properties of process calculi in GSOS format, using binary conjunction and assuming finite branching and a finite number of actions.

\topic{$\mu$-calculus}
Kozen's modal $\mu$-calculus (1983)~\cite{Kozen:1983} 
  subsumes several other weak temporal logics including CTL* 
  (Cranen et al.~2011)~\cite{Cranen11LinearCTL}, 
  and can encode weak transitions using least fixed points. 
Dam (1996)~\cite{Dam:1996pd} gives a modal $\mu$-calculus for the pi-calculus, 
  treating bound names using abstractions and concretions, 
  and provides a model checking algorithm. 
Bradford and Stevens (1999)~\cite{bradfield99:observational_mu_calculus} 
  give a generic framework for parameterising the $\mu$-calculus 
  on data environments, state predicates, and action expressions.
The logic defined in the present paper 
  can encode the least fixpoint operator of $\mu$-calculi 
  by a disjunction of its finite unrollings, 
  as seen in Section~\ref{sec:encod-least-fixp}.
We can immediately encode the atomic $\mu$-calculus by Klin and {\L}e{\l}yk (2017)~\cite{klin_et_al:LIPIcs:2017:7699}.

\topic{Pi-calculus}
The first HML for the pi-calculus is by Milner et al.~(1993)~\cite{Milner:1993ys}, 
  where infinite conjunction is used in the early semantics 
  and conjunctions are restricted to use a finite set of free names. 
The adequacy proof has the same structure as in this paper. 
The logic of Section~\ref{sec:nominaltransitionsystems},
applied to the pi-calculus transition system from which bound input actions $x(y)$ have been removed, contains the logic $\mathcal{F}$ of Milner et al., 
or the equipotent logic $\mathcal{FM}$ if we take the set of name matchings $[a=b]$ as state predicates.
Gabbay~\cite{Gabbay2003piinFM} gives the first nominal syntax and operational semantics for the pi-calculus, in Fraenkel-Mostowski set theory.

Koutavas and Hennessy (2012)~\cite{DBLP:journals/cl/KoutavasH12} give a weak HML for a higher-order pi-calculus with both higher-order and first-order communication using an en\-vironment-sensitive LTS. 
Xu and Long (2015)~\cite{Xu2015} define a weak HML with countable conjunction for a purely higher-order pi-calculus. The adequacy proof uses stratification.

Ahn et al.~(2017)\cite{ahn_et_al:LIPIcs:2017:7789} give an innovative intuitionistic logic characterising open bisimilarity. In their setting, substitution effects are only employed inside the definition of the logical implication operator.  Their soundness proof is standard; for the completeness proof the distinguishing formula for $P$ wrt.~$Q$ is constructed in parallel with one for $Q$ wrt.~$P$.

There are several extensions of HML with spatial modalities. The one most closely related to our logic is by Berger et al.~(2008)~\cite{Berger2008}. They define an HML with both strong and weak action modalities, fixpoints, spatial conjunction and adjunction, and a scope extrusion modality, to study a typed value-passing pi-calculus with selection and recursion. The logic has three (may, must, and mixed) proof systems that are sound and relatively complete.

\topic{Spi Calculus}
Frendrup et al.~(2002)~\cite{DBLP:journals/entcs/FrendrupHJ02} provide three
Hennessy-Milner logics for the spi calculus~\cite{abadi.gordon:calculus-cryptographic}. 
All three logics use infinite quantification without any consideration of finite support.
The transition system used is a variant of the one by Boreale et
al.~(2001)~\cite{bdp01cryptographic}, where a state is a pair $\procpair \sigma P$ of a process $P$ and its environment~$\sigma$: 
a substitution that maps environment variables to public names and messages received from the process.
This version of the spi calculus has expressions $\xi$, that are terms constructed from names and environment variables using encryption and decryption operators, and messages $M$, that only contain names and encryption.
Substitution $\xi\sigma$ replaces environment variables in $\xi$ with their values in $\sigma$,
and evaluation $\eval{\xi}$ is a partial function that attempts to perform the decryptions in $\xi$, 
yielding a message $M$ if all decryptions are successful.

As usual for the spi calculus, the bisimulation (and logic) is defined in terms of the environment actions, 
rather than the process actions.  In Frendrup's version of Boreale's environment-sensitive transition system, 
the transition labels are related to the process actions in the following way:
when a process $P$ receives message $M$ on channel $a$, the label of a corresponding environment-sensitive transition
$\procpair \sigma P \trans{a\,\xi}\procpair{\sigma'}{P'}$ describes how the environment $\sigma$ computed the message $M=\eval{\xi\sigma}$. 
For process output of message $M$ on channel $a$, the corresponding environment-sensitive transition is simply $\procpair \sigma P \trans{\overline a}\procpair{\sigma'}{P'}$; the message $M$ can be recovered from the updated environment $\sigma'$.

The logics of Frendrup et al.~include a matching modality $[M=N]A$ 
that is defined using implication: 
$\procpair \sigma P \models[M=N]A$ iff 
  $\eval{M\sigma}=\eval{N\sigma}$ implies $\procpair \sigma P \models A$.
This is equipotent to having matching as a state predicate, 
since we can rewrite all non-trivial guards by $[M=N]A \iff A\lor\neg[M=N]\top$.

The logic of Section~\ref{sec:nominaltransitionsystems},
applied to the nominal transition system with the environment labels $\tau$, $\overline a$ and $a\,\xi$ 
above has the same modalities as the logic $\mathcal{F}$ of Frendrup et al.

  The logic $\mathcal{EM}$ by Frendrup et al.~replaces the simple input modality by an
  early input modality $\may{\inlabel{a}(x)}^{E}A$,
  which (after a minor manipulation of the input labels) can be encoded as the conjunction
$\bigwedge_{\xi}\may{{a}\,\xi}A\ssubst{\xi}{x}$
  with support  $\names{A}\setminus\{x\}$.
  We do not consider their logic $\mathcal{LM}$ that uses a late input
  modality, since its application relies on sets that do not have
  finite support \cite[Theorem 6.12]{DBLP:journals/entcs/FrendrupHJ02}, which are not meaningful in
  nominal logic.

Frendrup et al.~claim to ``characterize early and late versions of the environment sensitive bisimilarity of~\cite{bdp01cryptographic}'',
but this claim only holds with some modification.
First the definition of static equivalence \cite[Definition 22]{frendrup01:MSc} that is used in the adequacy proofs 
is strictly stronger than the one that appears in the published paper \cite[Definition 3.4]{DBLP:journals/entcs/FrendrupHJ02}.
Thus, the adequacy results \cite[Theorems 6.3, 6.4, and 6.14]{DBLP:journals/entcs/FrendrupHJ02} are false as stated,
but can be repaired by substituting the stronger notion of static equivalence.
Then Frendrup's logics and bisimilarities become sound, but not complete, 
with respect to the bisimilarity of~\cite{bdp01cryptographic}, since the latter uses the weaker notion of static equivalence (Definition 3.4).
In Section~\ref{sec:applications} we sketched an instance of our logic, for Boreale's labelled transition system, that is adequate for early bisimilarity.

\topic{Applied Pi-calculus}
 A more recent work is a weak HML by H{\"u}ttel and Pedersen (2007)~\cite{journals/entcs/HuttelP07} for the applied pi-calculus by Abadi and Fournet (2001)~\cite{abadi.fournet:mobile-values},
where completeness relies on an assumption of image-finiteness of the weak transitions. 
Similarly to the spi calculus, there is a requirement that terms $M$ received by a process $P$ can be computed from the current knowledge available to an observer of the process, which we here write $M \in \known P$.

The logic contains atomic formulae for term equality (indistinguishability) in the frame of a process, corresponding to our state predicates.
However, H{\"u}ttel and Pedersen use a notion of equality (and thus static equivalence) that is stronger than the corresponding relation by Abadi and Fournet, 
and that is \emph{not} well-defined modulo alpha-renaming~\cite[p.~20]{pedersen2006}. 
Since our framework is based on nominal sets, we must identify alpha-equivalent processes, 
and instead use Abadi and Fournet's notion of term equality.

H{\"u}ttel and Pedersen's logic includes an early input modality and an existential quantifier. 
The early input modality $\may{\inlabel{a}(x)}A$ can be straightforwardly encoded as the conjunction 
$\bigwedge_{M}\may{\inlabel{a}M}A\ssubst Mx$, 
with support  $\{a\}\cup(\names{A}\setminus\{x\})$.
The definition of the existential quantifier takes the observer knowledge into account: 
$P$ satisfies $\exists x.A$ if $x\freshin P$ and there is $M\in\known P$ such that $\ssubst Mx\mid P$ satisfies $A$.
The condition $M \in \known P$ makes the quantifier difficult to encode using effects, 
since there is no corresponding state predicate 
(for good reason: the main property modelled by cryptographic process calculi is that different
cipher texts $\enc{M}k$ and $\enc{N}k$ are indistinguishable unless the key $k$ is known).
To treat the existential quantifier, we instead add an action $(x)$ with $\bn((x))=x$ and transitions $P \trans{(x)}\ssubst Mx\mid P$ if 
$M \in \known P$ and $x\freshin P$.
We can then encode $\exists x. A$ as $\may{(x)} A$.

\topic{Fusion calculus}
In an HML for the fusion calculus by Haugstad et al.~(2006)~\cite{haugstad2006modal} the fusions (i.e., equality relations on names) are action labels $\varphi$.  
The corresponding modal operator $\may\varphi A$ has the semantics that the formula $A$ must be satisfied for all substitutive effects of~$\varphi$ (intuitively, substitutions that map each name to a fixed representative for its equivalence class). 
  In order to represent fusion actions in the logics in this paper,
  we add substitution effects $\sigma$ such that $\sigma(P')=P'\sigma$.
  The fusion modality $\may\varphi A$ can then be encoded in our
  framework as $\may\varphi\bigvee_{\!\sigma}\effect\sigma A\sigma$,
  where the parameter $\sigma$ of the disjunction ranges over the (finite set of)
  substitutive effects of $\varphi$. 
Their adequacy theorem uses the contradiction argument with infinite conjunction, 
with no argument about finiteness of names for the distinguishing formula.

\topic{Nominal transition systems}
De Nicola and Loreti~(2008)~\cite{NicolaLoreti08} define a general format for multiple-labelled transition systems
with labels for name revelation and resource management,
and an associated modal logic with name equality predicates, name quantification ($\exists$ and $\new$), and a fixed-point modality.
In contrast, we seek a small and expressive HML for general nominal transition systems.
Indeed, the logic of De Nicola and Loreti can be seen as a special case of ours: 
their different transition systems can be merged into a single one, and 
we can encode their quantifiers and fixpoint operator as described in Section~\ref{sec:derived-formulas}. 
Nominal SOS of Cimini et al.~(2012)~\cite{cimini12nominalSOS} is also a special case of our nominal transition systems.
Aceto et al.~(2017)~\cite{aceto_et_al:LIPIcs:2017:7786} give conditions on rule formats for transitions from processes to residuals so that they generate nominal transition systems.

\section{Formalisation}
\label{sec:formalisation}
\newcommand{\rawformset}{\mathcal{R}}
We have formalised results of Sections~\ref{sec:nominaltransitionsystems} and~\ref{sec:variants}--\ref{sec:state-predicates-versus-actions}.
We use Nominal Isabelle~\cite{nominal2}, an implementation of nominal
logic in Isabelle/HOL~\cite{isabelle}.  This is a popular interactive
proof assistant for higher-order logic with convenient specification
mechanisms for, and automation to reason about, data types with
binders.  Our Isabelle theories are available from the Archive of
Formal Proofs~\cite{Modal_Logics_for_NTS-AFP}.  We here comment on
some of the interesting aspects, and the difficulties involved for the
results yet without a formal proof.

The main challenge, and perhaps most prominent contribution from the
perspective of formalisation, is the definition of formulas
(Definitions~\ref{def:formulas} and~\ref{def:fl-formulas}).  Nominal
Isabelle does not directly support infinitely branching data types.
We therefore construct formulas from first principles in higher-order
logic, by defining an inductive data type of \emph{raw} formulas
(where alpha-equivalent raw formulas are \emph{not} identified).  The
constructor for conjunction recurses through sets of raw formulas of
bounded cardinality, a feature made possible only by a recent
re-implementation of Isabelle/HOL's data type package~\cite{datatype}.

\begin{defi} \label{def:rawformulas}
  The set of raw formulas $\rawformset$ ranged over by $R$ is defined by induction as follows:
  \[ R \quad ::= \quad
    \conjunc R_i \casesep
    \neg R \casesep
    \varphi \casesep
    \alpha.R \]
  In $\conjunc R_i$ it is required that $\lvert I\rvert < \kappa$.
\end{defi}

There are no name abstractions or binders in raw formulas: $\alpha.R$,
unlike~$\may{\alpha} A$, is \emph{not} an abbreviation for an
equivalence class, and~$\bn(\alpha)$ plays no role in the definition
of raw formulas.  Name permutation distributes over all raw
constructors.  Raw formulas need not have finite support: for
instance, consider the raw formula $\bigwedge_{a\in S} \varphi_a$
where $\n(\varphi_a)=\{a\}$ and $S\subset\nameset$ is not finitely
supported.

We then define the \emph{concrete} alpha-equivalence of raw formulas
(not to be confused with nominal alpha-equivalence as defined in
Section~\ref{sec:background}), in the following just called
alpha-equivalence, by well-founded recursion.

\begin{defi}\label{def:rawalpha}
  Two raw formulas $\conjunc R_i$ and $\conjunc S_i$ are
  \emph{alpha-equivalent}~($\mathop{\approx_\alpha}$) if for every
  conjunct~$R_i$ there is an alpha-equivalent conjunct~$S_j$, and vice
  versa.  Two raw formulas $\alpha.R$ and $\beta.S$ are
  alpha-equivalent if there exists a permutation~$\pi$ with $\pi \cdot
  \alpha = \beta$
such
  that $\pi \cdot R \approx_\alpha S$.  Moreover,~$\pi$ must leave
  names that are in $(\n(\alpha) \cup N)\setminus \bn(\alpha)$ invariant,
  for some set of names $N$ that supports $[R]_{\approx_\alpha}$.
  The other cases in the definition of alpha-equivalence are standard.
\end{defi}

We note that alpha-equivalence is equivariant, i.e., $R \approx_\alpha
S$ iff $\pi \cdot R \approx_\alpha \pi \cdot S$.  Moreover, all raw
constructors respect alpha-equivalence.  To obtain formulas, we
quotient raw formulas by alpha-equivalence, and finally carve out the
subtype of all terms that can be constructed from finitely supported
ones.

\begin{defi}\label{def:hfs}
  An alpha-equivalence class $[R]_{\approx_\alpha}$ is
  \emph{hereditarily finitely supported~(h.f.s.)}\ if, for each
  subformula~$S$ of~$R$, $[S]_{\approx_\alpha}$ has finite support.
\end{defi}

Definitions~\ref{def:rawformulas}--\ref{def:hfs} have been formalised
in Isabelle.  To our knowledge, this is the first mechanisation of
infinitely branching nominal data types in a proof assistant.
Fortunately, we need not keep the details of this construction in
mind, since the formulas obtained in this way agree with those
obtained from Definition~\ref{def:formulas}.  We give an explicit
bijection.

\begin{defi}
  If~$R$ is a raw formula, the corresponding formula~$A_R$ is defined
  homomorphically except for the case $R=\alpha.S$ where
  $A_{\alpha.S}=\may\alpha A_{S}$.
\end{defi}

This defines a partial function $\Gamma\colon \mathcal{R}
\hookrightarrow \mathcal{A}$, $R \mapsto A_R$.  $\Gamma$ is partial
because $\Gamma(\conjunc R_i) = \conjunc \Gamma(R_i)$ is well-formed
only when $\{\Gamma(R_i) \setsep i\in I\}$ is finitely supported.

\begin{lem}\label{lem:a-equivariant}
  The partial function~$\Gamma$ is equivariant on its domain: $R \in
  \dom \Gamma$ iff $\pi \cdot R \in \dom \Gamma$, and in this case
  $\pi \cdot \Gamma(R) = \Gamma(\pi \cdot R)$.
\end{lem}

\begin{proof}
  By structural induction on~$R$, using the fact that all constructors
  of (raw) formulas are equivariant.
\end{proof}

\begin{thm}\label{thm:a-iff}
  For all $R$, $S \in \dom \Gamma$, we have $R \approx_\alpha S$ iff
  $\Gamma(R) = \Gamma(S)$.
\end{thm}

\begin{proof}
  By structural induction on~$R$, for arbitrary~$S \in \dom{\Gamma}$.
  The case $R=\varphi$ is trivial.  The cases $R=\neg R'$ and
  $R=\conjunc R_i$ follow from the induction hypothesis.  Finally
  consider $R=\alpha.R'$.

  $\Longrightarrow$: From $R \approx_\alpha S$ we obtain $S =
  \beta.S'$ and a permutation~$\pi$ with $\pi \cdot \alpha = \beta$
  and $\pi \cdot R' \approx_\alpha S'$.  Equivariance
  of~$\approx_\alpha$ implies $R' \approx_\alpha \pi^{-1} \cdot S'$,
  hence $\pi \cdot \Gamma(R') = \Gamma(S')$ by induction and
  Lemma~\ref{lem:a-equivariant}.
  Moreover, $\pi$ must leave names in $(\n(\alpha)\cup N) \setminus
  \bn(\alpha)$ invariant, where~$N$ supports~$[R']_{\approx_\alpha}$.
  The induction hypothesis implies $\n([R']_{\approx_\alpha}) =
  \n(\Gamma(R'))$.  Since $\n([R']_{\approx_\alpha})$ is the smallest
  set that supports $[R']_{\approx_\alpha}$ we have
  $\n([R']_{\approx_\alpha}) \subseteq N$.  Hence~$\pi$ witnesses
  $\may{\alpha}\Gamma(R') = \may{\beta}\Gamma(S')$.

  $\Longleftarrow$: From $\Gamma(R) = \Gamma(S)$ we obtain $S =
  \beta.S'$ and a permutation~$\pi$ with $\pi \cdot \alpha = \beta$
  and $\pi \cdot \Gamma(R') = \Gamma(S')$.
  Lemma~\ref{lem:a-equivariant} implies $\Gamma(R') = \Gamma(\pi^{-1}
  \cdot S')$, hence $\pi \cdot R' \approx_\alpha S'$ by induction and
  equivariance of~$\approx_\alpha$.
  Moreover, $\pi$ must leave names in $(\n(\alpha)\cup\n(\Gamma(R')))
  \setminus \bn(\alpha)$ invariant.  The induction hypothesis implies
  $\n([R']_{\approx_\alpha}) = \n(\Gamma(R'))$.  Hence $\pi$ witnesses
  $\alpha.R' \approx_\alpha \beta.S'$.
\end{proof}

Theorem~\ref{thm:a-iff} implies that~$\Gamma$ can be lifted to a
partial function on alpha-equivalence classes, $\hat{\Gamma} \colon
\mathcal{R}/_{\approx_\alpha} \hookrightarrow \mathcal{A}$,
$[R]_{\approx_\alpha} \mapsto A_R$ for $R\in\dom \Gamma$.  The
following lemma shows that this function is defined on all
alpha-equivalence classes that are h.f.s.

\begin{lem}
If~$[R]_{\approx_\alpha}$ is h.f.s, then $R \approx_\alpha S$ for some
$S \in \dom \Gamma$.
\end{lem}
  
\begin{proof}
  By structural induction on~$R$.
\end{proof}

To prove that~$\hat{\Gamma}$ is a bijection between
h.f.s.\ alpha-equivalence classes and formulas, it remains to show
that every formula is the image of some h.f.s\ alpha-equivalence
class.

\begin{thm}
  For any formula~$A$, there is h.f.s~$[R]_{\approx_\alpha}$ such that
  $\hat{\Gamma}([R]_{\approx_\alpha}) = A$.
\end{thm}

\begin{proof}
  By Theorem~\ref{thm:a-iff}, $\hat{\Gamma}$ is injective.  It follows
  that the inverse function $\hat{\Gamma}^{-1} \colon \mathcal{A}
  \hookrightarrow \mathcal{R}/_{\approx_\alpha}$ is well-defined
  on~$\hat{\Gamma}(\mathcal{R}/_{\approx_\alpha})$.  Equivariance
  of~$\Gamma$ (Lemma~\ref{lem:a-equivariant}) implies
  that~$\hat{\Gamma}^{-1}$ is equivariant on its domain.
  We show by induction on~$A$ that $A \in \dom{\hat{\Gamma}^{-1}}$ and
  that $\hat{\Gamma}^{-1}(A)$ is h.f.s.
\end{proof}

In order to conveniently work with formulas in Isabelle we have proved
important lemmas; for instance, a strong induction principle for
formulas that allows the bound names in~$\may{\alpha} A$ to be chosen
fresh for any finitely supported context.
This is why we prefer the logic for weak bisimulation
(Definition~\ref{def:weaklogic}) to be a sublogic: that way we do not
have to re-do these proofs.  Armed with these definitions and
induction principles, most of the proof mechanisation proceeded
smoothly.  The formal proofs follow closely the sketches provided in
this paper, and often helped clarify points and correct minor errors.
Our development currently amounts to approximately 8800 lines of Isar
proof scripts; this is divided according to
Table~\ref{table:isabelle-loc}.

\begin{table}
\begin{tabular}{llr}
  \toprule
\textbf{Topic}         & \textbf{Section}                    & \textbf{LOC} \\
\midrule
Basics                 & \ref{sec:nominaltransitionsystems}  & 3062 \\
L-bisimilarity         & \ref{sec:variants}                  & 3215 \\
Unobservables          & \ref{section:weak}                  & 1730 \\
Predicates vs actions & \ref{sec:state-predicates-versus-actions} & 817 \\
\bottomrule
\end{tabular}
\caption{Size (lines of code) of the Isabelle formalisation.}
\label{table:isabelle-loc}
\end{table}

Some parts of this paper lack formal proofs.  Among them are the
definitions and results in Section~\ref{sec:fixedpoints} on fixpoint
operators.  Here proof mechanisation would be tedious, since Isabelle
does not provide enough support for transfinite induction: explicit
reasoning about ordinals and cardinals would require further library
development~\cite{DBLP:conf/itp/Blanchette0T14}.  The results of
Sections~\ref{sec:applications} and~\ref{sec:related}, on applications
and comparisons with related work, lack proofs for two reasons.
First, the correspondences are not exact since we allow arbitrary
finitely supported conjunctions, and second, any formal proof must
begin by constructing a formal model of the related work under study,
which would require a large effort including resolving any ambiguities
in that work.

Finally, the result on disjunction elimination
(Section~\ref{sect:disjunctionelemination},
Theorem~\ref{thm:disjunctionelim}) is challenging for a different
reason.  The proof uses a case analysis and induction over derived
operators.  While we believe that our proof is correct, to formalise
it we would need to prove that this is well-defined, despite the
ranges of the derived operators not being disjoint (recall that
$\wmay{\alpha}A$ is shorthand for a disjunction, which in turn is
shorthand for a negation).  This turned out to be more difficult than
first anticipated, and at the time of writing, we have not finished
the formal proofs.

\section{Conclusion}
\label{sec:conclusion}

We have given a general account of transition systems and Hennessy-Milner Logic using nominal sets.  The advantage of our approach is that it is more expressive than previous work.  We allow infinite conjunctions that are not uniformly bounded, meaning that we can encode, e.g., quantifiers and the next-step operator.  We have given ample examples of how the definition captures different variants of bisimilarity and how it relates to many different versions of~HML in the literature. Our main results have been formalised in the interactive proof assistant Isabelle.

There are many interesting avenues for further research. Many process algebras can be given a semantics  where the operators of the algebra correspond to functions on nominal transition systems with designated initial states. There may be interesting classes of such functions, for example finitely supported bisimilarity preserving functions, that merit study. For particular such functions it would be interesting to explore compositionality of~$\models$. Similar work abounds for ordinary transition systems and modal logics; for nominal transition systems in general this area is very much open.

\section*{Acknowledgements}

We thank Andrew Pitts for enlightening discussions on nominal data types with infinitary constructors, and Dmitriy Traytel for providing an Isabelle/HOL formalisation of cardinality-bounded sets. We also thank Ross Horne for discussions on modal logics for open bisimilarity, and for pointing out an inadequacy of an earlier version of $F/L$-bisimilarity. We are very grateful to the anonymous referees for many constructive comments.

\bibliographystyle{alpha}
\bibliography{pi,logics}

\end{document}